\documentclass[amsmath,twocolumn,
superscriptaddress,
amssymb,
aps,
pra,
]{revtex4-2}

\usepackage{graphicx}% Include figure files
\usepackage{bm}% bold math
\usepackage[mathlines]{lineno}% Enable numbering of text and display math
\usepackage{amsthm}
\usepackage{dsfont}
\usepackage{xcolor}
\usepackage{enumitem}
\usepackage{amsmath}
\usepackage{braket}
\usepackage{booktabs}

\newtheorem{lemma}{Lemma}
\newtheorem{prop}{Proposition}
\newtheorem{thm}{Theorem}
\newcommand{\Id}{\mathds{1}}

\newcommand{\tr}{\mathrm{Tr}}

\usepackage{hyperref}
\hypersetup{
    colorlinks,
    linkcolor={red!50!black},
    citecolor={blue!50!black},
    urlcolor={blue!80!black}
}

\begin{document}

\preprint{APS/123-QED}

\title{Efficiently learning non-Markovian noise in many-body quantum simulators}% Force line breaks with \\

\author{Jordi A. Montañà-López}
\email{jordi.montana-lopez@mpq.mpg.de}

\affiliation{Max Planck Institute of Quantum Optics, Garching bei München - 85748, Germany}
\affiliation{Department of Electrical and Computer Engineering, University of Washington - 98195, USA}
%\altaffiliation{Physics Department, XYZ University.}%Lines break automatically or can be forced with \\
\author{Andreas Elben}
  \affiliation{Laboratory for Theoretical and Computational Physics, Paul Scherrer Institute, CH-5232 Villigen-PSI, Switzerland}
    \affiliation{ETHZ-PSI Quantum Computing Hub, Paul Scherrer Institute, CH-5232 Villigen-PSI, Switzerland}
    \author{Joonhee Choi}
\affiliation{Department of Electrical Engineering, Stanford University, Stanford, CA, USA}
\author{Rahul Trivedi}

\affiliation{Max Planck Institute of Quantum Optics, Hans-Kopfermann-Straße 1, 85748 Garching, Germany}
\affiliation{Department of Electrical and Computer Engineering, University of Washington - 98195, USA}

\date{\today}% It is always \today, today,

\begin{abstract}
As quantum simulators are scaled up to larger system sizes and lower noise rates, non-Markovian noise channels are expected to become dominant. While provably efficient protocols for Markovian models of quantum simulators, either closed system models (described by a Hamiltonian) or open system models (described by a Lindbladian), have been developed, it remains less well understood whether similar protocols for non-Markovian models exist. In this paper, we consider geometrically local lattice models with both quantum and classical non-Markovian noise and show that, under a Gaussian assumption on the noise, we can learn the noise with sample complexity scaling logarithmically with the system size. Our protocol requires preparing the simulator qubits initially in a product state, introducing a layer of single-qubit Clifford gates and measuring product observables. 

%based on further prior knowledge of the structure of the simulator Hamiltonian, with the goal of measuring only low-weight local operators. As a particular example, we show that we can estimate the mean and covariance matrix of an ensemble of Hamiltonians whose coefficients are jointly Gaussian random variables, noting that this is a particular example of a non-Markovian coupling to a bath.

%A strategy for closed-system Hamiltonian learning with an exponential improvement in the locality of the terms by leveraging the sparsity of the Hamiltonian is also discussed.
%\begin{description}
%You may use the \texttt{description} environment to structure your abstract; use the optional argument of the \verb+\item+ command to give the category of each item. 
%\end{description}
\end{abstract}

\maketitle

%\tableofcontents
\section{\label{sec:intro} Introduction}
Recent experimental progress in developing quantum simulators in different technological platforms \cite{bluvstein2024logical,ai2024quantum,main2025distributed,reichardt2024demonstration,google2025quantum,google2023suppressing,zhang2023superconducting} has opened up the possibility of using them to solve computationally challenging many-body problems. In order to both understand and characterize the accuracy of these simulators \cite{mark2023benchmarking,miessen2024benchmarking,eisert2020quantum,shaw2024benchmarking} as well as to effectively use error mitigation strategies \cite{cai2023quantum,takagi2022fundamental,quek2024exponentially}, we would like to be able to learn an accurate noise model that captures their dynamics. As a first approximation, typically errors in a quantum simulator implementing a target Hamiltonian are modeled as Markovian. This can include errors modeled as either unwanted Hamiltonian terms or as jump operators in a master equation describing the simulator dynamics. However, the Markovian model is only an approximation --- in many experimental quantum platforms, gates and Hamiltonian evolution are realized through global driving fields that couple multiple qubits simultaneously. This shared control introduces correlated noise and memory effects, giving rise to non-Markovian dynamics. As both the size of the quantum simulators are increased and the noise rates are reduced, such non-Markovian effects are expected to become important and increasingly determine the noise floor in the quantum device. This opens up the theoretical question of developing efficient measurement protocols to learn non-Markovian noise models of quantum simulators.

For quantum devices containing only a few qubits, the quantum channel generated by its dynamics (including both the target Hamiltonian and unwanted errors) can be learned by performing quantum process tomography \cite{chuang1997prescription,mohseni2008quantum}. In particular, for digital circuits, recent work has also developed ``process tensor tomography" which generalizes process tomography to learn a non-Markovian noise model independent of the target evolution \cite{pollock2018non,white2020demonstration,white2022non,mangini2024tensor, guo2020tensor,luchnikov2024scalable}. However, these approaches usually require a number of samples that scale exponentially with the number of qubits and thus become infeasible to employ at large system sizes. For digital quantum simulators implementing a quantum circuit built out of a universal gate set, noise is often characterized by the average error rates of the individual gates used in the circuit, which can be extracted using the randomized benchmarking protocol \cite{knill2008randomized,wallman2014randomized,emerson2005scalable,dankert2009exact,kelly2014optimal,veldhorst2014addressable,helsen2022general,proctor2017randomized}. While traditionally randomized benchmarking was developed under the the Markovian assumption, where individual gate errors are assumed to be independent and uncorrelated, it has recently been extended to the non-Markovian setting as well \cite{gandhari2025quantum,brillant2025randomized}. Despite these significant advances, employing randomized benchmarking for learning noise models remains restricted to digital simulators and only furnishes an ``average" measure of the noise rate as opposed to a more detailed model of the noise in the quantum device.

In analog quantum simulators, alternative approaches have been developed that leverage the natural system dynamics in the many-body regime to infer both the system Hamiltonian and Lindbladian noise processes \cite{haah2024learning,stilck2024efficient,mobus2023dissipation,mobus2025heisenberg,hangleiter2024robustly,wu2025hamiltonian,gu2024practical,ma2024learning,olsacher2025hamiltonian,francca2025learning}. Such strategies have been explored both for the closed system setting, where the unitary evolution of the simulator is generated by a Hamiltonian \cite{haah2024learning,gu2024practical}, as well as in the Markovian open system setting, where dissipative evolution is described by a Lindbladian \cite{stilck2024efficient,francca2025learning}. The general strategy involves evolving the analog simulator for a short time starting from an initial product state and measure a set of local observables --- since the evolution is performed only for a short time, the parameters of the Hamiltonian or the Lindbladian describing the simulator can be backed out from the measured observables. This can be done either by measuring the observables at a single, sufficiently small, time and using a non-linear equation solver \cite{haah2024learning} or by measuring time-traces of the observables and employing a robust fitting strategy \cite{stilck2024efficient}. For Markovian geometrically local models, it has also been rigorously proven that the sample complexity of a learning protocol employing either of these strategies scales at-most logarithmically with the system size and attains the standard-quantum-limit scaling with respect to the desired precision in the simulator parameters \cite{stilck2024efficient}. Furthermore, for the case of Hamiltonian learning, a strategy that combines long-time evolution and Pauli operator reshaping has also been developed to attain Heisenberg scaling in the sample complexity with respect to the desired precision \cite{huang2023learning,ma2024learning}. However, the development of these strategies beyond the Markovian approximation still remains an experimentally relevant open question and has largely been considered only for small systems \cite{zhang2022predicting,ivander2024unified,vezvaee2024fourier,aguiar2025quantum,yang2024control,gullans2024compressed, dong2025efficient,varona2025lindblad}.

%Recent progress has been achieved by considering noise described by perturbations in the generators of a desired time evolution. In a closed system, the Hamiltonian generating a time evolution can be learned efficiently from the short-time dynamics of the system \cite{haah2024learning}. This has prompted several works that relax many of the original assumptions \cite{bakshi2024structure,gu2024practical,yu2023robust,bakshi2024learning,wu2025hamiltonian,olsacher2025hamiltonian,mobus2023dissipation}, achieve the Heisenberg limit \cite{huang2023learning,ma2024learning} and leverage quantum memories for this task \cite{caro2024learning}. A robust fitting technique was introduced in \cite{stilck2024efficient} to learn the Lindbladian in Markovian open quantum systems from observable time traces. 

In this paper, we address this question for geometrically local quantum simulators. We consider a reasonably general non-Markovian noise model, where the noise can act on multiple neighboring qubits in a correlated manner, with the only assumption being that the noise arises from an interaction with a stationary and Gaussian environment \cite{trivedi2022description}. We provide a measurement strategy that uses short-time evolution to learn the two-point correlation functions of the non-Markovian environment to a desired precision. Similar to the algorithms available for the Markovian case, our strategy provably requires a number of samples scaling only logarithmically with the system size, and the postprocessing time is similarly efficient. Furthermore, we provide an explicit construction of the initial states and observables that need to be measured for our protocol. In the worst case, while the sample complexity of our protocol scales exponentially with the size of the support of the system-environment coupling operators, it is independent of the locality of the terms in the system Hamiltonian.

Finally, we also adapt our strategy to a coherent noise model, where the errors are Gaussian random geometrically local terms in the Hamiltonian albeit with possible all-to-all correlations. Physically, this would be the case where a global control field is used to implement single- or two-qubit Hamiltonian terms in an experiment, but the control itself is noisy --- then the laser noise would result in coherent errors that can be modeled as random, but coherent, terms added to the target Hamiltonian. On averaging the unitary evolution with respect to the noise realization, we in general still obtain a non-Markovian channel for the simulator dynamics. In this case, we show that even if the noise is all-to-all correlated, we can still learn the mean and covariance matrix of the noise model with almost the same experimental requirements as Hamiltonian learning. 

% As is the case for Lindbladian learning methods \cite{stilck2024efficient}, the dependence on the locality of the interaction terms is exponential. However, we show that we only need to perform state tomography on regions slightly larger than the supports of the non-Markovian interaction terms. That is, our method works well even if the system Hamiltonian has very large locality, as long as the non-Markovian dissipation acts on few-qubit terms. 

This paper is structured as follows: in Section \ref{sec:model_nonmarkovian} we introduce the non-Markovian noise model, describing the non-Markovian effects via so-called memory kernels. In Section \ref{sec:results} we present our results for estimating the derivatives at $t=0$ of these memory kernels, on two classes of non-Markovian models that satisfy different locality assumptions. In Section \ref{sec:measurement_strategy} we lay out the strategy to learn these derivatives from a linear system of equations and provide the set of measurements that need to be performed to invert it. In Section \ref{sec:numerics} we present numerical results for the performance of our protocol with both sparse non-Markovian kernels and in the case of correlated errors across all interaction terms. In Section \ref{sec:polynomial_fit} we explain how this system of equations arises from a polynomial fit to the time traces of the measured observables and analyze the number of samples required to estimate the derivatives of the kernels to the desired precision.

\section{\label{sec:model_nonmarkovian}Describing Non-Markovian noise}
\subsection{Model}
\begin{figure}[b]
    \centering
    \includegraphics[width=0.9\linewidth]{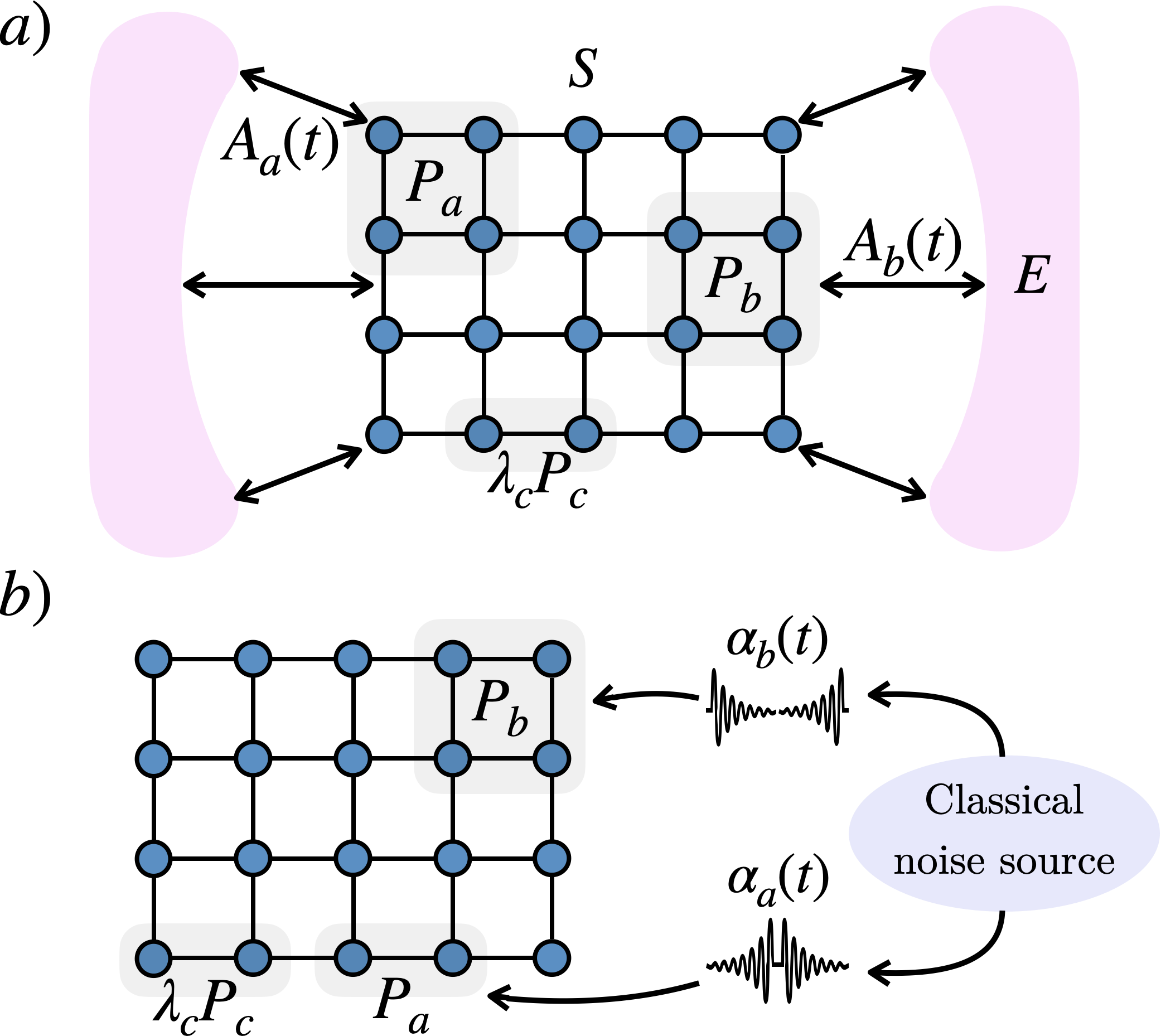}
    \caption{Open system models considered in this paper. (a) The system $S$ (blue circle lattice) is assumed to interact with an environment $E$ (pink), with a coupling $P_a A_a(t)$ between a Pauli string $P_a$ acting on the system and a time-dependent Hermitian operator $A_a(t)$ acting on the environment. (b) A non-dissipative environment, where the environment operators are commuting $[A_a(t), A_{b}(s)] = 0$, can be equivalently described by a noisy term in the system Hamiltonian  $\alpha_a(t) P_a$, where $\alpha_a(t)$ is a classical noise. }
    \label{fig:models_intro}
\end{figure}

We will model the quantum simulator as a system with $N$ qubits interacting with a decohering environment. In the Pauli basis, the Hamiltonian describing the system-environment evolution can generically be written as sum of a static system Hamiltonian $H_S$ and a time-dependent system-environment coupling $V_{SE}(t)$:
\begin{subequations}\label{eq:non_markovian_model_def}
\begin{align}
    H(t) = H_S + V_{SE}(t), 
\end{align}
with
\begin{align}
      H_S &= \sum_{a \in \mathcal{P}_S} \lambda_a P_a , \\
      V_{SE}(t) &= \sum_{b \in \mathcal{P}_{SE}} P_b A_b(t),
\end{align}
\label{eq:H(t)}
\end{subequations}
where $\mathcal{P}_S$ and $\mathcal{P}_{SE}$ are the index sets of Pauli operators (i.e., operators of the form $\sigma_1 \otimes \sigma_2 \otimes \dots \sigma_N$ where $\sigma_i \in \{\Id, \sigma^x, \sigma^y, \sigma^z\}$) appearing in the system Hamiltonian, $H_S$, and the coupling Hamiltonian, $V_{SE}(t)$, respectively. Note that we choose to express the coupling Hamiltonian directly in the interaction picture with respect to the environment Hamiltonian, which makes the Hermitian operator $A_b(t)$ dependent on time. While the system-environment coupling will in general result in the system and the environment evolving into an entangled state, we will assume that initially the system and the environment are not entangled i.e., $\rho(0) = \rho_S \otimes \gamma_E$ where $\rho_S$ and $\gamma_E$ are the initial system and environment state, respectively. We assume that the state $\gamma_E$ is the same for each realization of the experiment. This is reasonable as, for example, the environment could be evolving under a quadratic Hamiltonian and $\gamma_E$ be its generalized Gibbs state. After a projective measurement on the system and turning off the system-environment coupling, the environment will equilibrate to $\gamma_E$. 

From the Dyson expansion  it follows that the state of the system at time $t$, $\rho_S(t) = \tr_E(U(t, 0)(\rho_S \otimes \gamma_E)U(0, t))$ where $U(t, s) = \mathcal{T}\exp(-i\int_s^t H(\tau) d\tau)$, is entirely determined by the $n-$point environment correlation functions $\mathcal{C}_{a_1, a_2 \dots a_n}(t_1, t_2 \dots t_n)$ \cite{feynman2000theory,huang2024unified,cirio2023pseudofermion}:
\begin{align*}
    \mathcal{C}_{a_1, a_2 \dots a_n}(t_1, t_2 \dots t_n) = \tr_E\bigg[\bigg(\prod_{i=1}^n A_{a_i}(t_i)\bigg) \gamma_E\bigg].
\end{align*}
We will make two physically reasonable assumptions on these environment correlation functions which are reasonable in several leading quantum information processing platforms, such as superconducting systems, \cite{chirolli2008decoherence,zhang2009controlling,tang2012measuring,malekakhlagh2016non,heidler2021non}, trapped ion systems as well as quantum optical systems \cite{gardiner2004quantum,liu2011experimental,chiuri2012linear,de2008matter}. First, \emph{Stationarity}, i.e.,~the correlation functions $\mathcal{C}_{a_1, a_2 \dots a_n}(t_1, t_2 \dots t_n)$ depend only on time differences $t_2 - t_1, t_3 - t_1 \dots t_n - t_1$ i.e.
    \begin{align}
    &\mathcal{C}_{a_1, a_2 \dots a_n}(t_1, t_2, \dots t_n) =\nonumber\\
    &\qquad  \mathcal{C}_{a_1, a_2 \dots a_n}(0, t_2 - t_1, \dots t_n - t_1).
    \end{align}
Second, \emph{Gaussianity}, i.e.,~the correlation functions $\mathcal{C}_{a_1, a_2 \dots a_n}(t_1, t_2 \dots t_n)$ satisfy the Wick's theorem i.e.,
    \begin{align}
        &\mathcal{C}_{a_1, a_2 \dots a_n}(t_1, t_2 \dots t_n) =\nonumber\\
        &\qquad \sum_{\mathcal{S}} \prod_{(i, j) \in \mathcal{S}} \mathcal{C}_{a_i, a_j}(t_i, t_j) \prod_{k \in \mathcal{S}^c} \mathcal{C}_{a_k}(t_k), 
    \end{align}
    where $\mathcal{S}$ is a subset of $\{1, 2 \dots n\}$ with even number of elements that are divided into pairs $(i, j)$ with $i < j$.
Thus, a stationary Gaussian environment is described entirely by 
\begin{align}
&\lambda_a = \mathcal{C}_a(0) = \tr(\gamma_E A_a(0)) \text{ and }\\
&K_{a b}(t) =\mathcal{C}_{a, b}(t, 0) = \tr(\gamma_E A_a(t) A_b(0)). \label{eq:kernel_def}
\end{align}
Since we can always transform the environment operators $A_a(t) \to A_a(t) - \lambda_a$ and the system Hamiltonian $H_S \to H_S + \sum_{a \in \mathcal{P}_{SE}}\lambda_a P_a$ without changing the coupling Hamiltonian, without loss of generality we will assume that $\lambda_a= 0$.

A special class of such systems are \emph{non-dissipative} systems, in which the environment operators, $A_b(t)$, commute with one another: $[A_b(t), A_{b'}(s)] = 0$. In this case, the evolution of the system qubits can be alternately described by a Hamiltonian ensemble i.e.,
\[
\rho_S(t) = \mathbb{E}_{\alpha}\big(U_{\alpha}(t, 0) \rho_S(0) U_{\alpha}(0, t)\big),
\]
where $U_{\alpha}(t, s) = \mathcal{T}\exp(-i\int_s^t H_{\alpha}(\tau) d\tau)$ with
\begin{align}
    H_{\alpha}(t) = \sum_{a\in \mathcal{P}_S} \lambda_a P_a + \sum_{b \in \mathcal{P}_{SE}}\alpha_b(t) P_b,\label{eq:H_alpha}
\end{align}
with $\alpha_b(t)$ being a classical stationary real-valued Gaussian stochastic process with $\mathbb{E}(\alpha_b(t)) = 0$ and $\mathbb{E}(\alpha_a(t) \alpha_b(s))= K_{a b}(t - s)$. Thus, non-Markovian non-dissipative noise can equivalently be considered to be classical random noise in Hamiltonian parameters.

\subsection{The learning problem}
Learning the quantum simulator model with either dissipative or non-dissipative noise, under the stationary Gaussian assumption, requires learning the Hamiltonian parameters $\lambda_a$ as well as the kernels $K_{a,b}(t)$. Since, in the most general setting, the kernels can be almost arbitrary functions of time, learning them could require learning an infinite number of parameters and is thus not a well-posed problem. However, in most physical situations, the kernel is parameterized by a small number of parameters. Under the additional physically reasonable assumption that the kernels $K_{a, b}(t)$ are smooth functions of time, we can then aim to learn its derivatives $K_{a, b}^{(m)}(0)$ for $m \in \{0, 1, 2 \dots M \}$ upto a pre-specified order $M$, and then use these derivatives to infer the parameters that $K_{a, b}(t)$ depends on. 

As an example of a physically relevant non-Markovian noise model, consider the setting where the environment itself consists of a set of discrete bosonic modes undergoing Markovian particle loss. This family of models, referred to as the pseudomode model in open-system theory \cite{tamascelli2018nonperturbative,ferracin2024spectral}, is a well-studied model for non-Markovian open systems and can also approximate the dynamics of \emph{any} non-Markovian models with smooth kernels if a sufficiently large number of discrete modes are used \cite{trivedi2021convergence}. For such models, it can be shown that the kernels $K_{a, b}(t)$ are  given by linear combinations of decaying exponentials \cite{tamascelli2018nonperturbative}:
\begin{align}
    &K_{a,b}(t) = \sum_l v_{a,l}^* v_{b,l} e^{(i\epsilon_l -\gamma_l/2)|t|}. \label{eq:kernel_ansatz}
\end{align}
Here the sum is over environment modes $l$ that are coupled to $P_a,P_b$ and $\epsilon_l,\gamma_l$ are the resonant frequency and dissipation rate of the modes. In Section \ref{sec:numerics_nonmarkovian} we simulate this kind of non-Markovian dynamics in the fermionic case and show that we can estimate the kernel derivatives at $t=0$ for $m=0,1,...,M$. With an ansatz for the kernel as in Eq.\eqref{eq:kernel_ansatz} this yields the system of equations relating the kernel derivatives $K_{a, b}^{(m)}(0)$ to the parameters $v_{a, l}, \epsilon_l, \gamma_l$:
\begin{align}
    &K_{a,b}^{(m)}(0)=\sum_l v_{a,l}^* v_{b,l}\Big(i\epsilon_l-\frac{\gamma_l}{2}\Big)^m. \label{eq:K_ansatz_MFD}
\end{align}
Inverting either this system of non-linear equations or by using a filter diagonalization algorithm \cite{mandelshtam1998multidimensional,mandelshtam2000multidimensional}, we can then obtain the parameters $v_{a, l}, \epsilon_l, \gamma_l$ from the learned derivative $K_{a, b}^{(m)}(0)$.

Note that this differs from previous work on characterizing many-body non-Markovian processes \cite{white2020demonstration,white2022non} that could also be applied to our setting. Their general method of process tensor tomography requires a number of samples that grows exponentially with the system size and Markov order --- their measure of non-Markovianity. Our work focuses on a particular non-Markovian model with geometrically-local kernels, which are relevant in many experimental platforms. We expect our learning problem to directly provide experimentally-relevant information, since kernel ansatzs such as Eq.\ref{eq:kernel_ansatz} provide resonant frequencies and dissipation rates of the environment modes. 

In the next section, we summarize the main results of this paper showing that this learning problem can be solved with a number of copies of $\rho_S(t)$ that scales at most logarithmically with the system size and outline the measurement protocols that attain this scaling.

\section{\label{sec:results}Results}
% \emph{Non-Markovian models}. We will consider the model of non-Markovian evolution given by a system $S$ of $N$ spin-$\frac{1}{2}$ particles interacting with stationary Gaussian environment (E) consisting of $O(N)$ bosonic baths. We assume that the Hamiltonian $H(t)$ constants $M=O(N)$ terms and is geometrically local: the Pauli strings in $H_S$ and $V_{SE}(t)$ are supported on at most $k_S,k_{SE}$ sites, respectively, and their support diameter is upper bounded by some constant $a_0$. This ensures that each Pauli string in $H(t)$ overlaps with at most  $\mathfrak{d}$ other Pauli strings in $H(t)$. We further assume that the baths are geometrically local: each term $P_aA_a(t)$ does not commute with at most $\mathfrak{d}$  other terms $P_bA_b(t)$ and that for a fixed $P_a$, there's at most $\mathfrak{d}$ kernels $K_{a,b}(t)$ that are nonzero. That is, if $P_a,P_b$ act on far away regions their kernel vanishes $K_{a,b}(t)=0$. In a $D$-dimensional lattice, in general $\mathfrak{d}= O(D^{\max(k_S,k_{SE})})$, but this might be very loose in many problems of interest, so we assume that $\mathfrak{d}$ is a constant with respect to $k_S,k_{SE}$. With a technical assumption capturing the fact that the kernel has ``finite memory", we will be able to use the non-Markovian Lieb-Robinson bound \cite{trivedi2024lieb}. Our goal will be to obtain estimates for the system parameters $\lambda_a$ and for the derivatives $K_{a,b}^{(m)}:=\partial_t^mK_{a,b}(t)|_{t=0}\in\mathbb{C}$ of the function $K_{a,b}(t)$ at $0$ up to some degree $0\leq m\leq L$.\\
\subsection{Non-Markovian models}
We will consider the model of non-Markovian evolution given by a system $S$ of $N$ spin-$\frac{1}{2}$ particles interacting with stationary Gaussian environment $E$. We assume that the system-environment Hamiltonian $H(t)$ is a sum of $O(N)$ geometrically local terms --- more specifically, we assume that the Pauli strings in $H_S$ and $V_{SE}(t)$ are supported on at most $k_S,k_{SE}=O(1)$ sites, respectively, and their support diameter is also upper bounded by a constant $a_0=O(1)$. Note that this ensures that each Pauli string in the Hamiltonian $H(t)$ overlaps with at most $\mathfrak{d}=O(1)$ other Pauli strings in $H(t)$. Additionally, we assume that the environment does not introduce non-local interactions in $H(t)$ by assuming that each term $P_aA_a(t)$ in the coupling Hamiltonian [Eq.~(\ref{eq:non_markovian_model_def}c)] is non-commuting with at most $\mathfrak{d}=O(1)$ other terms $P_bA_b(t)$ and that for a fixed $P_a$, there's at most $s=O(1)$ kernels $K_{a,b}(t)$ that are nonzero. Finally, we will assume that the non-Markovian environment has ``finite memory" as formalized by the total variation of the kernels $K_{a,b}(t)$ in Ref.~\cite{trivedi2024lieb} i.e., we will assume that
\begin{align}\label{eq:total_variation_condition}
\sup_a \sum_{b} \int_{-\infty}^\infty \left| {K_{a,b}(s)}\right |ds \leq O(1).
\end{align}
As established in Ref.~\cite{trivedi2024lieb}, this condition is sufficient to ensure that the system has a finite velocity of information propagation, thus making it a reasonable assumption for an experimentally realistic model of a noisy quantum simulator.

Our first result shows that there is an efficient protocol that allows us to estimate both the parameters $\lambda = \{\lambda_a\}_a$ of the system Hamiltonian as well as the Kernel derivatives
\begin{align}
    &K_{a,b}^{(m)}(0):=\partial_t^mK_{a,b}(t)|_{t=0},\quad K^{(m)}=\{K_{a,b}^{(m)}(0)\}_{(a,b)} \label{eq:Km_def}
\end{align} by only initializing and measuring the system.

\begin{prop}\label{prop:non_markovian_main}
    There is a protocol which uses only initial product states, single qubit gates and Pauli measurements and with probability $\geq 1-\delta$ obtains estimates $\hat{\lambda},\hat{K}^{(m)} (0)$ satisfying $||\hat{\lambda}-{\lambda}||_\infty,||\hat{K}^{(m)}(0)-{K}^{(m)}(0)||_\infty<\epsilon$ for every $m\in \{0,1...,M\}$ with $N_S$ samples, where:
\begin{align}
    &N_S=O\bigg(e^{O(M^2\log M)} \frac{1}{\epsilon^2}\log\bigg(\frac{N}{\delta}\bigg)\bigg).\notag
\end{align}

This protocol requires evolving initial states up to time $t=O(1)$ and the classical postprocessing is $O(N\cdot N_S)$.
\end{prop}

Note that the provable sample complexity in Proposition \ref{prop:non_markovian_main} grows super-exponentially with $M$, which is the degree of the highest derivative of the kernels $K_{a,b}(t)$ that we want to learn. This indicates that this protocol would perform poorly in accurately learning noise models with very rapidly varying kernels --- however in many physically relevant settings, these kernels are expected to be a sum of a small number of exponentials in $t$, as in Eq.\eqref{eq:kernel_ansatz}, \cite{tamascelli2018nonperturbative,ferracin2024spectral} and can thus be specified by only low-order derivatives at $t = 0$.\\

% For a $t_{\max} = O(1)$, the observable time evolution can be approximated by a polynomial of degree $d=O(\text{poly}(t_{\max},\log(\frac{1}{\epsilon})))$. We will pick $t_{\min}\leq \frac{1}{d^2},t_{\max} = t_{\min}+O(1)$. \\

\begin{figure}
    \centering
    \includegraphics[width=0.9\linewidth]{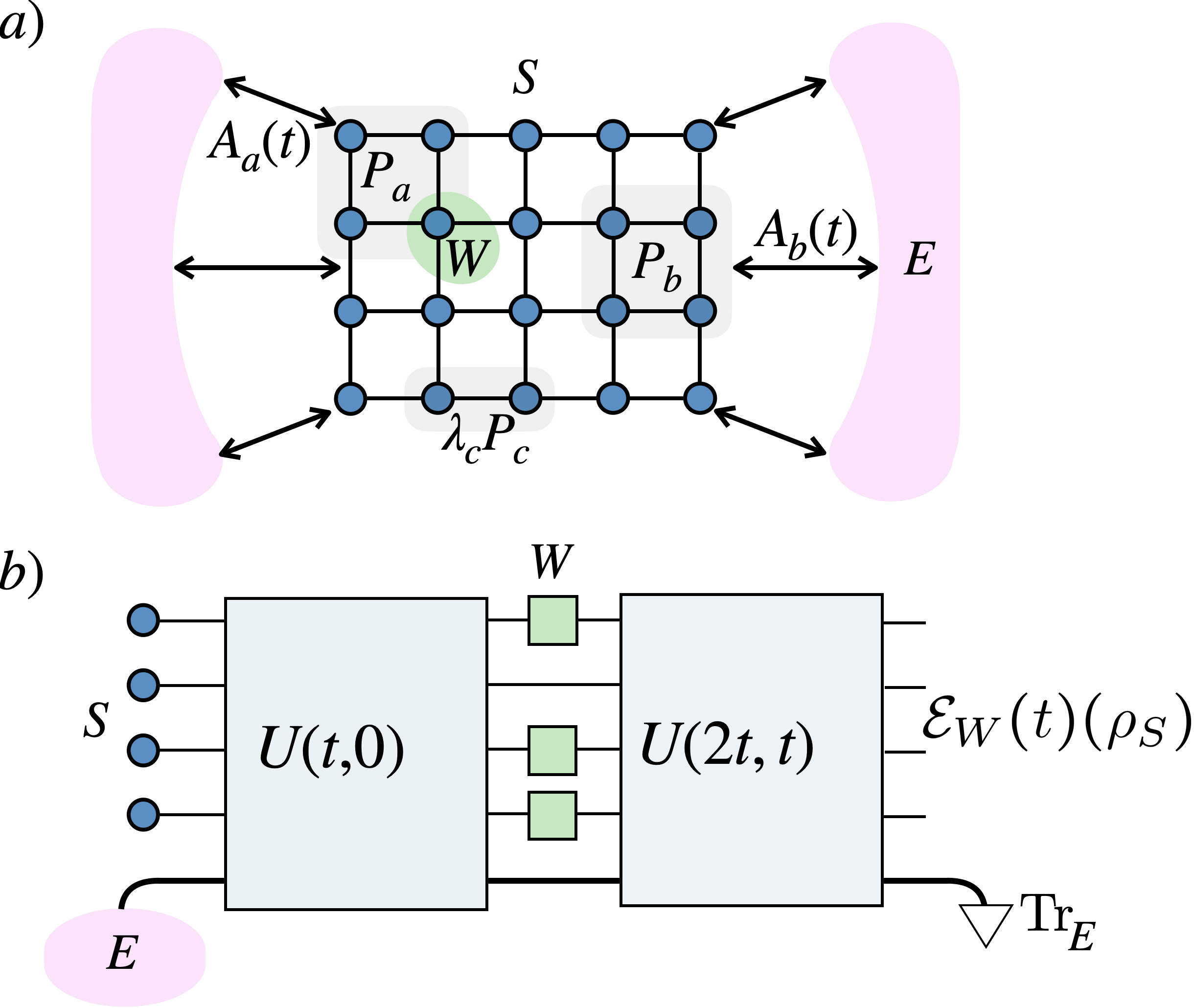}
    \caption{a) Single-qubit gate $W$ applied on the support of $P_a$ when trying to estimate kernel $K_{a,b}(t)$. b) Quantum channel used to prepare the state: an initial state $\rho(0)=\rho_S\otimes \gamma_E$ is prepared, where $\rho_S $ is a product state chosen depending on the kernels we are trying to learn while $\gamma_E$ is a Gaussian state. The state is time evolved under $U(t,0)$ and a unitary $W$, consisting of single-qubit Clifford gates on certain sites, is applied. The state is further evolved under $U(2t,t)$ and the environment is traced out.}
    \label{fig:state_preparation}
\end{figure}
Our protocol builds on a strategy recently developed for many-body Lindbladian learning \cite{stilck2024efficient}, which is based on measuring system observables after a short-time evolution of an initial controllable system state. More specifically, we start with an initial system-environment state $\rho(0)=\rho_S\otimes\gamma_E$, where $\rho_S$ is a product state of our choice, and then prepare a state $\mathcal{E}_W(t)(\rho_S)$ by performing the evolution shown in Figure \ref{fig:state_preparation}:
\begin{align}\label{eq:rho_W_t}
    &\mathcal{E}_W(t)(\rho_S) =\tr_E(V_W(t)\rho_S \otimes \gamma_EV_W^\dagger(t)),
\end{align}
where $V_W(t) = U(2t,t)WU(t,0)$, with $W$ being a layer of single qubit gates applied on the system qubits and $U(t, s) = \mathcal{T}\exp(-i \int_s^t H(\tau) d\tau)$. For small enough $t$, $\mathcal{E}_W(t)(\rho_S)$ can be well approximated by a second order Dyson series expansion of the system-environment evolution $U(t,0)$ and $U(2t, t)$:
\begin{subequations}\label{eq:dyson_series}
\begin{align}
    \mathcal{E}_W(t)(\rho_S) \approx W \rho_S W^\dagger + \mathcal{F}_W^{(1)}(t)(\rho_S) + \mathcal{F}_W^{(2)}(t)(\rho_S),
\end{align}
where $\mathcal{F}_W^{(j)}(t)$ are system super-operators that correspond to the $j^\text{th}$ order term in the Dyson series expansion:
\begin{align}
    &\mathcal{F}_W^{(1)}(t)(\rho_S) =-it\big(W[H_S,\rho_S]W^\dagger+[H_S,W \rho_S W^\dagger]\big),\\
    &\mathcal{F}_W^{(2)}(t)(\rho_S) =\nonumber\\
    &\quad-\bigg(\int_0^t \int_{0}^{s_1}  \tr_E(W[H(s_1),[H(s_2),\rho_S\otimes \gamma_E]]W^\dagger) d^2s+\notag\\
    &\qquad \hspace{0pt} \int_t^{2t} \int_{0}^{t} \tr_E([H(s_1),W[H(s_2),\rho_S \otimes \gamma_E ]W^\dagger]) d^2 s+\notag \\
    &\qquad \hspace{0pt}\int_t^{2t} \int_{t}^{s_1} \tr_E([H(s_1),[H(s_2),W\rho_S W^\dagger \otimes \gamma_E]])d^2 s\bigg).
\end{align}
\end{subequations}
where in Eq.~(\ref{eq:dyson_series}a) we have used that, by assumption, $\text{Tr}_E(A_a(t) \gamma_E) = 0$. From Eq.~\ref{eq:dyson_series}, it is clear that $\mathcal{F}_W^{(1)}(t)$ is linear in the system Hamiltonian coefficients $\lambda_a$. Furthermore, since $\mathcal{F}_W^{(2)}(t)$ is quadratic in the environment operators $A_a(t)$ and the environment is traced out, it is linear in $K_{a,b}(t)$. Consequently, it follows from Eq.~\ref{eq:rho_W_t} that the first derivative of $\mathcal{E}_W(t)$ at $t = 0$ carries information about the system Hamiltonian coefficients $\lambda_a$, while second and higher derivatives carry information about $K_{a,b}^{(m)}(0)$. The problem of learning $\lambda_a, K_{a,b}^{(m)}(0)$ thus reduces to finding an informationally complete set of initial states $\rho_S$, observables $O$ and unitaries $W$ for which the derivatives of $\text{Tr}_S[O\mathcal{E}_W(t)(\rho_S)]$ at $t = 0$ can be used to uniquely determine $\lambda_a, K_{a,b}^{(m)}(0)$.

For learning system Hamiltonian parameters from first derivative information, it was shown in Ref.~\cite{haah2024learning} that we can always choose $W = \Id^{\otimes N}$, and for every $\lambda_a$, there is a unique choice of an initial state $\rho_{S}$ as well as an observable $O$ such that $\partial_t\text{Tr}_S[O\mathcal{E}_{W = \Id^{\otimes N}}(t) (\rho_{S})]|_{t=0} = -8\lambda_a$, as shown in Eq. \eqref{eq:H_learning}. In order to learn $K_{a,b}^{(m)}(0)$, we need to use higher-order derivative information and the nested commutator structure of $\mathcal{F}_W^{(2)}(t)(\rho_S)$ makes it more complicated to develop an explicit tomography protocol. One of the main contributions of our work, detailed in Section \ref{sec:measurement_strategy}, is to determine a set of initial product states $\rho_S$, product observables $O$ and single qubit unitaries $W$ such that $K_{a,b}^{(m)}(0)$ can be learned from the derivatives of $\text{Tr}_S[O\mathcal{E}_W(t)(\rho_S)]$ at $t = 0$.  In contrast to the case of Hamiltonian learning, where a single observable and initial state can be used to directly infer a given Hamiltonian parameter, we need to use several initial-state and observable pairs $(\rho_S,O)$ to learn a particular $K_{a,b}^{(m)}(0)$. Nevertheless, we show that it is sufficient to use product states and measure observables on a region that is slightly larger than the supports $P_a,P_b$ to learn $K_{a,b}^{(m)}(0)$. Furthermore, unlike Hamiltonian or Lindbladian learning \cite{stilck2024efficient}, where the learning strategy involve preparing an initial state and measuring an observable after short time evolution (thus setting $W = \Id^{\otimes N}$), we show in Section \ref{sec:measurement_strategy} that the short-time evolution of observables as described by a second-order Dyson series expansion does not depend on the kernels $\text{Im}[K_{c, c}^{(m)}(0)]$. Nevertheless, we find that by applying a layer of single qubit gates $W$ in the middle of the time-evolution allows us to also learn these Kernel parameters --- we provide an explicit procedure to pick these gates $W$ in addition to the initial state $\rho_S$ and observables $O$ that allows us to learn all the Kernel parameters in Section \ref{sec:measurement_strategy}.

% one needs to introduce a layer of single-qubit unitaries, $W$, on a few sites to selectively introduce relative phases between the terms of interest.\\

% In the Lindbladian setting, the necessary configurations $(\rho,O)$ have been provided for the case of single-qubit dissipative terms \cite{stilck2024efficient}. Our analysis of the sufficient configurations $(\rho,O,W)$ to invert the linear systems of equations for the higher order derivatives in the non-Markovian setting also cover the general many-body Lindbladian case by setting $W=\Id$. The need for $W$ in the non-Markovian setting arises from the fact that the Lindbladian case corresponds to an even kernel, $\delta(t)$, in which case $W=\Id $ is sufficient. More general kernels also have odd terms, which require $W\neq\Id$.\\

While our construction of the required measurement settings $(\rho_S,O,W)$ is based on analyzing a second order Dyson series expansion, we rigorously account for the higher order terms in establishing the sample complexity result in Proposition \ref{prop:non_markovian_main}. \emph{First}, while the first and second order derivatives of the measured observable $\text{Tr}_S(O\mathcal{E}_W(t)(\rho_S))$, at $t = 0$, is entirely determined by the first and second order term in the Dyson series expansion of $\mathcal{E}_W(t)$, the third and higher order derivatives of $\text{Tr}_S(O\mathcal{E}_W(t)(\rho_S))$, needed to obtain the higher derivatives of $K_{a,b}(t)$, in general have a contribution from higher order terms in the Dyson series for $\mathcal{E}_W(t)$. Furthermore, the third and higher order terms in the Dyson series of $\mathcal{E}_W(t)$ would be non-linear in the parameters $\lambda_a, K_{a,b}^{(m)}(0)$ --- consequently, the derivatives of $\text{Tr}_S(O\mathcal{E}_W(t)(\rho_S))$ at $t = 0$ depend nonlinearly on $\lambda_a, K_{a,b}^{(m)}(0)$. Despite this, we show that with the choice of $(\rho_S,O,W)$ made on the basis of a second order Dyson series expansion of $\mathcal{E}_W(t)$, the resulting non-linear equations relating $\lambda_a, K_{a,b}^{(m)}(0)$ to the derivatives of $\text{Tr}_S(O\mathcal{E}_W(t)(\rho_S))$ at $t = 0$ can be uniquely and efficiently solved to obtain $\lambda_a, K_{a,b}^{(m)}(0)$. Furthermore, in an actual experiment we can only measure the time-traces $\text{Tr}_S(O\mathcal{E}_W(t)(\rho_S))$ to an accuracy determined by the number of measurement shots and its derivatives have to be obtained by fitting a polynomial to the measured data --- to obtain rigorous sample complexity bounds for the tomography protocol, following a strategy laid out in \cite{stilck2024efficient, arora2024outlier} for Lindbladian learning, we use Lieb-Robinson bounds \cite{trivedi2024lieb} to show that for evolution times $t \leq t_\text{max} = O(1)$, these time traces can be well approximated by polynomials of degree $d=O(\text{poly}(t_{\max},\log({\epsilon}^{-1})))$, where $\epsilon$ is the target precision desired in the parameters $\lambda_a, K_{a,b}^{(m)}(0)$, and with the observables measured only at times $n\tau$ where $n \in \{1, 2, 3 \dots\} $ and $\tau = \Theta(1/d^2)$. These fitted polynomials are then used for extracting the derivatives of $\text{Tr}_S(O\mathcal{E}_W(t)(\rho_S))$ which are the subsequently used for learning $\lambda_a, K_{a,b}^{(m)}(0)$.

% we sample $\tr(O\rho_W(t,\tau))$ at different times and, using a Lieb-Robinson bound in the non-Markovian setting \cite{trivedi2024lieb}, argue that we can approximate the time traces accurately by low-degree polynomials for constant time. Using the robust fitting strategy of \cite{stilck2024efficient,arora2024outlier}, we obtain polynomial fits whose derivatives at $t=\tau=0$ are close to those of $\tr(O\rho_W(t,\tau))$. Inverting the systems of equations given by these derivatives we obtain the desired kernel coefficients.

\subsection{\label{sec:ensemble_model}Ensemble Hamiltonian model}
The second class of models that we address in this paper are Hamiltonians with noisy parameters --- we consider the simulator to be described by
\[
H_\Lambda=\sum_{a \in \mathcal{P}_S} \Lambda_a P_a,
\]
where, again, $P_a$ are geometrically local Paulis which are supported on at most $O(1)$ sites and with diameter upper bounded by an $O(1)$ constant and $\Lambda_a$ are jointly Gaussian random variables with mean $\lambda_a = \mathbb{E}(\Lambda_a)$ and covariances $\Sigma_{a,b} = \mathbb{E}(\Lambda_a \Lambda_b) - \lambda_a \lambda_b$. We assume the Hamiltonian contains $|\mathcal{P}_S|=O(N)$ terms and allow the covariance matrix to be dense (i.e., allow for all-to-all correlations between the coefficients $\Lambda_a$) and only assume that $
|{\lambda_a}| \leq 1$ and $\text{var}(\Lambda_a) = \Sigma_{aa} \leq  1$. We remark that, as described in Section \ref{sec:model_nonmarkovian}, this model is equivalent to a non-dissipative open system model with system Hamiltonian $\sum_a \lambda_a P_a$ and a time-independent kernel $K_{a,b}(t)=\mathbb{E}_\Lambda(\Lambda_a\Lambda_b)$. However, since the kernels $K_{a,b}(t)$ are not necessarily sparse and do not satisfy the total variation condition in Eq.~\eqref{eq:total_variation_condition}, Proposition \ref{prop:non_markovian_main} cannot be directly applied to this setting. Nevertheless, we show that both the means $\lambda = \{\lambda_a\}_a$ and covariances $\Sigma_{a,b}$ can be learned efficiently for this model following a strategy similar to the one used for the more general non-Markovian model.
\begin{figure}
    \centering
    \includegraphics[width=0.9\linewidth]{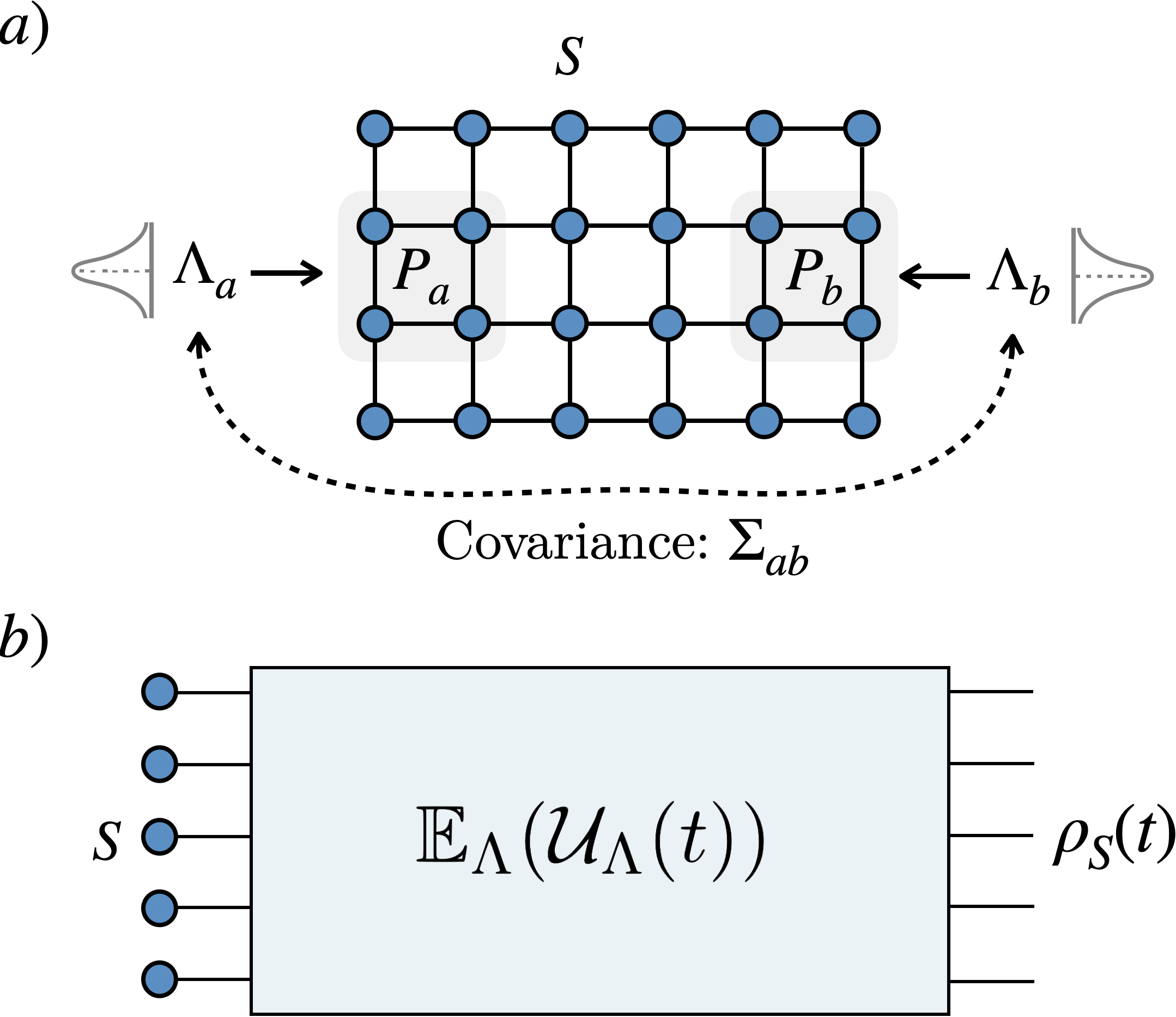}
    \caption{a) In the ensemble Hamiltonian model, while the Pauli strings are geometrically local, the correlations between their coefficients can be all-to-all: $\Lambda_a,\Lambda_b$ can have covariance $\Sigma_{ab} = \Omega(1)$ even if the supports of $P_a,P_b$ are far away. b) Quantum channel used to prepare the sate: a state $\rho(0)=\rho_S\otimes \gamma_E$ is prepared, where $\rho_S $ is a product state chosen depending on the kernel we are trying to learn while $\gamma_E$ is a Gaussian state. The state is time evolved under $\mathcal{U}_\Lambda (t)(\rho_S)=U_\Lambda(2t,0)\rho_S U_\Lambda(2t,0)^\dagger$ with $U_\Lambda(t,0) = e^{-itH_\Lambda}$, without introducing any intermediate gate $W$. We take the expectation value with respect to the jointly Gaussian random variables $\Lambda$. }
    \label{fig:ensemble_state_preparation}
\end{figure}

% . Therefore,  $\mathbb{E}_\Lambda(\rho_\Lambda(t))$, where $\rho_\Lambda(t) = e^{-itH_\Lambda}\rho e^{itH_\Lambda}$, does not satisfy a master equation. Since our covariance matrix is not necessarily sparse, the kernels need not be geometrically local. This is a necessary condition for the the non-Markovian Lieb-Robinson bound \cite{trivedi2024lieb}. Even for a sparse covariance matrix and geometrically local kernels, constant kernels don't satisfy a finite memory condition required by the non-Markovian Lieb-Robinson bound \cite{trivedi2024lieb} used in the previous section. Therefore, we need some other way of showing that time traces of local observables can be approximated by low-degree polynomials. Since $\Lambda$ are not bounded variables, we can not directly use Lieb-Robinson bounds for geometrically local Hamiltonians. However, they are jointly Gaussian random variables, so we can bound the probability that they deviate significantly from their mean. By controlling the tail contribution we obtain an effective Lieb-Robinson bound for the expected time trace of local observables for $t\leq t_{\max}=O(\log(\frac{M}{\epsilon})^{-\frac{1}{2}})$.

\begin{prop}\label{prop:hamiltonian_ensemble}
    There is a protocol which uses only initial product states and Pauli measurements and with probability $\geq 1-\delta$ obtains estimates $\hat{\lambda},\hat{\Sigma}$ satisfying $||\hat{\lambda}-\lambda||_\infty,||\hat{\Sigma}-\Sigma||_{\max}<\epsilon$ for an $s$-sparse $\Sigma$ with $N_S$ samples, where:     
\begin{align}
    &N_S=O\bigg(\frac{s}{\epsilon^2}\log\bigg(\frac{N}{\delta}\bigg)\bigg).\notag
\end{align}
This protocol requires evolving initial states up to time $t=O((\log (N) \log(1/\epsilon))^{-\frac{1}{2}})$ and the classical postprocessing is $O(sN\cdot N_S)$.
\end{prop}
\noindent The learning protocol for this model is considerably simpler than that for the non-Markovian model. \emph{First}, we do not need an intermediate layer of single qubit gates shown in Fig.~\ref{fig:ensemble_state_preparation}. \emph{Second}, we only need to extract the first and second derivatives of the measured observables --- the mean $\lambda_a$ are estimated from the first derivative, as in Hamiltonian learning \cite{haah2024learning,gu2024practical} and  $\mathbb{E}_\Lambda(\Lambda_a\Lambda_b)$ from the second derivatives. We describe the precise choice of initial product states as well as observables in the next section.

To obtain rigorous sample complexity bounds for this tomography protocol, we again use the Lieb-Robinson bounds \cite{hastings2006spectral} to show that the measured observables are well approximated by low degree polynomials --- while the Lieb-Robinson bounds derived for non-Markovian lattice models \cite{trivedi2024lieb} does not apply to this model directly (since $\Lambda_a$ might have all-to-all correlations), we use the fact that since each $\Lambda_a$ is Gaussian with mean and variance $\leq 1$, it is going to be $\leq O(1)$ with high probability and consequently the $H_\Lambda$ will have an $O(1)$ Lieb-Robinson velocity with high probability. By controlling the tail contribution of the low probability event where $H_\Lambda$ has a large Lieb-Robinson velocity, in Appendix \ref{app:RandomHam} we show that the time trace up to $t_{\max}=O((\log(N)\log(1/\epsilon))^{-1/2})$ can be well-approximated by a polynomial of degree $d=O(\text{poly}(t_{\text{max}},\log(\epsilon^{-1})))$, where $\epsilon$ is the target precision in the means and covariances. The first and second time-derivatives of the observable at $t= 0$ can be extracted by evolving it from $0$ to $t_{\text{max}}$, measuring it at times $n \tau$ where $n \in \{1, 2, 3\dots \}$ and $\tau = \Theta(1/d^2) $, and obtaining the polynomial fit. \\

\section{\label{sec:measurement_strategy} Measurement strategy}
We now present the details on the measurement protocols that have been outlined in Section ~\ref{sec:results} and achieve the sample complexities in propositions \ref{prop:non_markovian_main} and \ref{prop:hamiltonian_ensemble} --- we focus on describing the measurement protocol for the non-Markovian setting since the same protocol also works for the ensemble Hamiltonian setting. We first briefly review the measurement protocol for learning the system Hamiltonian parameters $\lambda_a$ which have been provided in \cite{haah2024learning} and then describe the more complex measurement protocol needed for measuring the kernel derivatives $K_{a,b}^{(m)}(0)$.

\subsection{Learning the system Hamiltonian parameters $\lambda_a$ (Review of Ref.~\cite{haah2024learning})}
To learn the system Hamiltonian parameter $\lambda_a$, we choose the intermediate unitary $W = \Id^{\otimes N}$. For every Pauli string $P_a, a \in \mathcal{P}_S$ in the system Hamiltonian, we can always find a single-qubit Pauli $P_{I}$ such that $\{P_{I}, P_a\} = 0$ --- suppose we choose the initial state $\rho_S = (I^{\otimes N}+P_{I}) / 2^N$ and the observable $O = i[P_a, P_{I}] = 2i P_a P_{I}$, then from Eqs.~\ref{eq:rho_W_t} and \ref{eq:dyson_series} it follows that the first derivative of the expected observable $\text{Tr}[O \mathcal{E}_{W = \Id^{\otimes N}}(t)(\rho_S)]$ at $t = 0$ is given by
\begin{align*}
&\partial_t \text{Tr}[O \mathcal{E}_{W = \Id^{\otimes N}}(t)(\rho_S)]|_{t = 0} \nonumber\\
&\qquad =\frac{4}{2^N} \sum_{b \in \mathcal{P}_S}\lambda_b \text{Tr}(P_a P_{I}[P_b, P_{I}]).
\end{align*}
Noting that since $P_{I}$ and $P_a$ anticommute by construction, $\text{Tr}(P_a P_{I}[P_b, P_{I}]) = -2\text{Tr}(P_a P_b) = -2^{N + 1}\delta_{a, b}$, and consequently 
\begin{align}
\partial_t \text{Tr}[O \mathcal{E}_{W = \Id^{\otimes N}}(t,\rho_S)]|_{t = 0}  =-8\lambda_a. \label{eq:H_learning}
\end{align}
This shows that using $W = \Id^{\otimes N}, \rho_S = (I + P_{I})/2^N$ and $O = 2i P_a P_{I}$ is enough to learn $\lambda_a$ from the first derivative of the expected observable at $t= 0$. 

In practice, it is more convenient to prepare initial product states instead of states of the form $\rho_S = (I  +P_I)/2^N$ for a Pauli string $P_I$. In particular, we can recognize that $\rho_S$ has product states as its eigenvectors, so it can be expressed as
\begin{align}\label{eq:initial_paulc_State_to_product}
\rho_S =\frac{1}{2^N} \sum_{r_1, r_2 \dots r_{|\mathcal{S}|} = \pm 1} (1 + r_1r_2 \dots r_{|\mathcal{S}|}) \bigotimes_{i\in \mathcal{S}}\ket{r_i}_i\!\bra{r_i},
\end{align}
where we assume that $P_I$ is supported on qubits in $\mathcal{S}$ with $|\mathcal{S}| \leq O(1)$ and $|\pm1\rangle_i$ is the $\pm 1$ eigenstate of the pauli at the $i^\text{th}$ qubit in $P_I$. Therefore, if we measure the expected value of an observable, and its derivatives, with $2^{|\mathcal{S}|}$ states of the form $\bigotimes_{i\in \mathcal{S}}\ket{r_i}_i\!\bra{r_i} \bigotimes_{i \in \mathcal{S}^c}I/2$, then they can be classically postprocessed as per Eq.~\eqref{eq:initial_paulc_State_to_product} to obtain its expected value in the state $(I + P_I)/2^N$ in $O(2^{|\mathcal{S}|}) \leq O(1)$ time.

% As was shown in Ref.~\cite{haah2024learning}, the system Hamiltonian parameter $\lambda_a$ 
% Consider initializing the system qubits in a state $\rho_S = (\Id^{\otimes N} + P_I)/2^N$, where $P_I$ is a Pauli string
% The initial states required for our protocol are states of $N$ qubits in the system that are non-identity on some support $\mathcal{S}$, where they are an eigenstate of a Pauli string: $\rho = \underset{i\in \mathcal{S}}{\bigotimes}|r_i,\alpha_i\rangle\langle r_i,\alpha_i|\underset{i\not\in \mathcal{S}}{\bigotimes}\frac{\Id}{2}$, where $r_i=\pm 1, \alpha_i=x,y,z$. Here $|\pm 1,\alpha\rangle$ is the $\pm 1$-th eigenstate of the Pauli matrix $\sigma^{\alpha}$. At each site, our measurements will consist of either tracing out the qubit or measuring it in a basis $|\pm 1,\beta\rangle$. From these initial states and measurements we will estimate the following time traces, where we initialize a Pauli string $P_I$ and measure another Pauli string $P_O$:
% \begin{align}
%     &B_{W,(O,I)}(t):=\frac{1}{2^N}  \tr(P_O P_{I,W}(t)).
% \end{align}

% Suppose we want to learn a Hamiltonian coefficient $\lambda_a$. In the Hamiltonian setting it has been shown that it can be estimated from the first derivative of the time trace with $P_I$ being a single-qubit Pauli that anticommutes with $P_a$ and $P_O = (i[P_a,P_I])^\dagger$ \cite{haah2024learning}. This still holds in the non-Markovian case, and we can choose $W=\Id$:
% \begin{align}
%     &\partial_t B_{\Id,(O,I)}(t)|_{t=0}=\partial_t D_{\Id,1}(t)|_{t=0}=8\lambda_a. \label{eq:Hamiltonian_learning}
% \end{align}

\subsection{Learning the Kernel parameters $K_{a,b}^{(m)}(0)$}
\subsubsection{Notation and preliminaries}
Next, we consider learning the derivatives $K_{a,b}^{(m)}(0)$ of $K_{a,b}(t)$ at $t=0$. To explicitly construct initial states and observables that would be needed to learn $K_{a,b}^{(m)}(0)$, it will be instructive to consider the action of the channel $\mathcal{E}_W(t)$ defined in Eq.~\ref{eq:rho_W_t} and depicted schematically in Fig.~\ref{fig:inversion_schematic} on an ``input" Pauli $P_I$ and take its overlap with an ``output" Pauli $P_O$ i.e., 
\begin{align}\label{eq:experimental_time_traces}
B_{W,(O,I)}(t) =\frac{1}{2^N} \text{Tr}[P_O \mathcal{E}_W(t)(P_I)].
\end{align}
Even though $P_I/2^N$ is not a valid initial state of the system qubits, as long as $P_I$ is supported on $O(1)$ qubits, we can follow the strategy outlined in the previous Section to efficiently obtain $B_{W, (O, I)}(t)$ --- first initialize the system qubits in the eigenstates of $P_I$, measuring the observable $P_O$ and then classically postprocess the resulting measurements to obtain $B_{W, (O, I)}(t)$.

Performing a Dyson series expansion of $\mathcal{E}_W(t)$, we can express
\begin{align}
    B_{W, (O, I)}(t) =\frac{1}{2^N}\sum_{n = 1}^\infty   \text{Tr}[P_O \mathcal{F}_W^{(n)}(t)(P_I)],
\end{align}
where $\mathcal{F}_W^{(n)}(t)(P_I)$ is the $n^\text{th}$ order term in the Dyson series expansion of $\mathcal{E}_W(t)(P_I)$ as introduced in Eq.~\ref{eq:dyson_series}. We remark that $\mathcal{F}_W^{(n)}(t)(\rho_S)$ will have $n$ commutators with the full system-environment Hamiltonian $H(t)$ and $n$ integrals over time. Furthermore, the Taylor expansion of $\mathcal{F}_W^{(n)}(t)$ as a function of $t$ around $t = 0$ will have the form
\[
\mathcal{F}_W^{(n)}(t) = \sum_{m \geq n} \frac{t^m}{m!} [\partial_s^m \mathcal{F}_W^{(n)}(s)]_{s = 0},
\]
i.e., due to the $n$ time integrals in $\mathcal{F}_W^{(n)}(t)$, the leading order of its Taylor expansion will be $t^n$ and hence the summation in the above equation starts from $n$. We emphasize that if the system-enviroment Hamiltonian is time-independent (for e.g., for the ensemble Hamiltonian model with the trace over the environment replaced with a classical ensemble average), then $\mathcal{F}_W^{(n)}(t)$ will only have a term proportional to $t^n$ in its Taylor expansion.

% Let $D_{W,n}(P_I,t)$ be the $n$-th order in the Dyson series of $P_{I,W}(t)$, i.e. the term with $n$ commutators. Then its $m$-th derivative at $0$ can be written as:
% \begin{align}
%     &\partial_t^mP_{I,W}(t)|_{t=0}=\sum_{n=0}^\infty \partial_t^mD_{W,n}(P_I,t)|_{t=0}.
% \end{align}
\begin{figure*}
    \centering
    \includegraphics[width=0.8\linewidth]{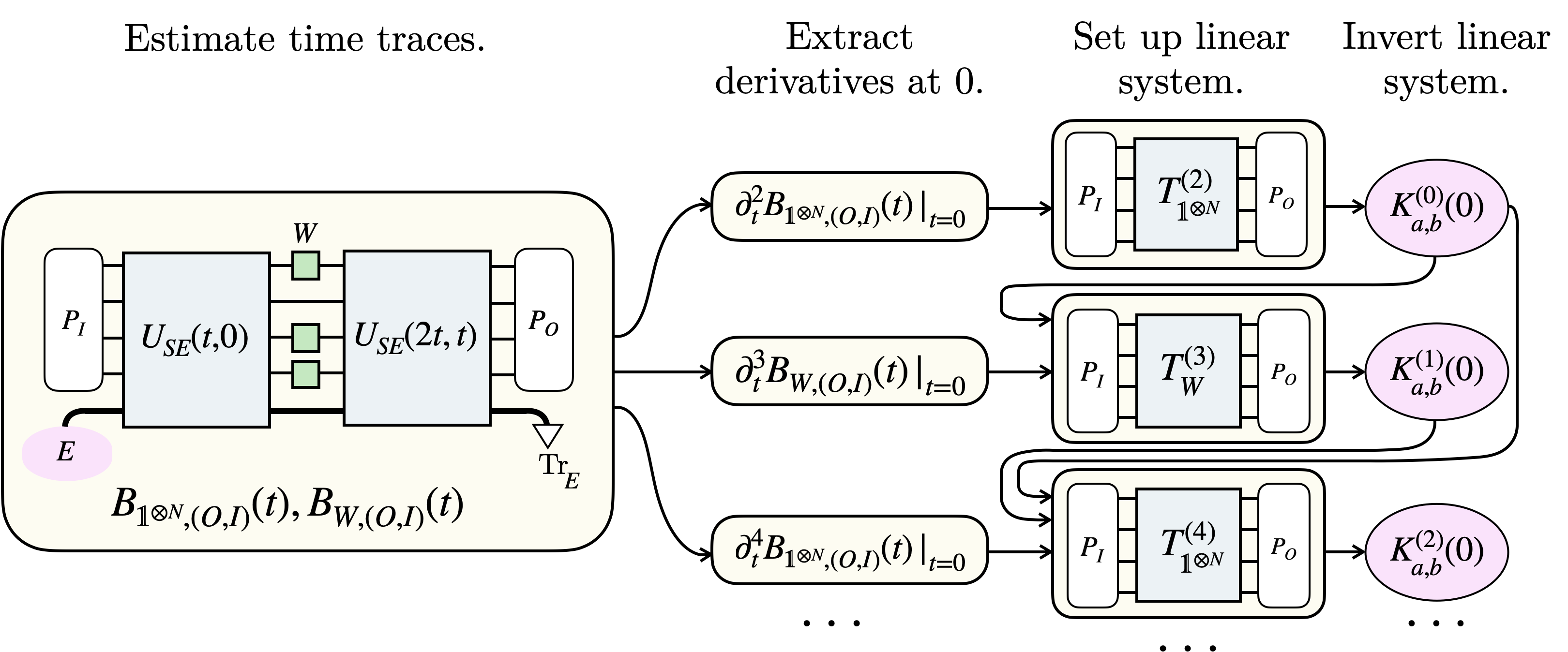}
    \caption{Schematic for the procedure to learn the kernel derivatives $K_{a,b}^{(0)}(0),K_{a,b}^{(1)}(0),...$. Given the Pauli strings $P_a,P_b$ associated to the kernel $K_{a,b}(t)$, we perform tomography on a region $\mathcal{I}_{a,b}$ slightly larger than the supports of $P_a,P_b$. We estimate $\frac{1}{2^N}\tr(P_O \mathcal{E}_W(t)(P_I))$ by sampling eigenstates of $P_I$, which is supported only on $\mathcal{I}_{a,b}$, evolving for time $t$, introducing a layer of single-qubit gates $W$, evolving further until time $2t$ and finally measuring $P_O$, which is only supported on $\mathcal{I}_{a,b}$. The intermediate $W$ depends on $P_a,P_b$ and can be chosen to be a single-qubit gate according to Table \ref{tab:W_gate}. Repeating this procedure for several $t$ yields a time trace $B_{W,(O,I)}(t)$. We similarly obtain the time trace $B_{\mathds{1}^{\otimes N},(O,I)}(t)$, this time without introducing any gate $W$. We estimate the $m$-th derivative at $t=0$ of these time traces using polynomial interpolation: these yield systems of equations $T_{\Id^{\otimes N}}^{(m)}$ ($m$ even), $T_{W}^{(m)}$ ($m$ odd) that are linear in the $(m-2)$-th kernel derivative, using Eq. \eqref{eq:time_traces_equations}. In order to obtain $T_{\Id^{\otimes N}}^{(2)}$ from $\partial_t^2B_{\Id^{\otimes},(O,I)}(t)|_{t=0}$ we need to use estimates $\lambda_a$ for the system coefficients. Tomography allows us to obtain an estimate for $K_{a,b}^{(0)}(0)$ from the linear map $T_{\Id^{\otimes N}}^{(2)}$ using Eq.\eqref{eq:xi_case1}. Similarly, we obtain $T_{W}^{(3)}$ from $\partial_t^3B_{W,(O,I)}(t)|_{t=0}$ and recover $K_{a,b}^{(1)}(0)$ using Eqs.\eqref{eq:xi_case2},\eqref{eq:xi_case3}. The fourth derivative $\partial_t^4B_{\Id^{\otimes N},(O,I)}(t)|_{t=0}$ will depend nonlinearly on $K_{a,b}^{(0)}(0)$, as well as the system parameters $\lambda_a$. Thus, we use our estimates obtained in previous steps to obtain the linear map $T_{\Id^{\otimes N}}^{(4)}$, which allows us to obtain an estimate for $K_{a,b}^{(2)}(0)$. Similarly, higher derivatives of the time trace depend linearly on the kernel derivatives not yet estimated, and nonlinearly on those kernel derivatives and system parameters already estimated.}
    \label{fig:inversion_schematic}
\end{figure*}

The first order term $\mathcal{F}^{(1)}_W(t)$ [Eq.~(\ref{eq:dyson_series}b)] would depend only on the system Hamiltonian coefficients $\lambda_a$, which have already been learned. Since $\text{Tr}(A_a(t) \gamma_E) = 0$, the second order term $\mathcal{F}^{(2)}_W(t)$ [Eq.~(\ref{eq:dyson_series}c)] can further be expressed as:
\begin{align}
    \mathcal{F}_{W}^{(2)}(t) =  \mathcal{F}_{W, SE}^{(2)}(t)+ \mathcal{F}_{W,S}^{(2)}(t), \label{eq:UW2_def}
\end{align}
with $\mathcal{F}_{W, SE}^{(2)}(t)$ is given by Eq.~(\ref{eq:dyson_series}c) with $H(t) \to V_{SE}(t)$ and $\mathcal{F}_{W, S}^{(2)}(t)$ is given by Eq.~(\ref{eq:dyson_series}c) with $H(t) \to H_{S}$. $\mathcal{F}_{W, S}^{(2)}(t)$ depends only on the system Hamiltonian coefficients $\lambda_a$, while $\mathcal{F}_{W,SE}^{(2)}(t)$ depends on the kernels $K_{a,b}(t)$.

The second and higher derivatives of $B_{W, (O, I)}(t)$ can then be expressed as
\begin{subequations}\label{eq:time_traces_equations}
\begin{align}
    &\partial_t^m B_{W, (O, I)}(t)\big|_{t=0} = \frac{1}{2^N}\text{Tr}[P_O T_W^{(m)}(P_I)] +f^{(m)}_{W, (O, I)},
\end{align}
where $m \geq 2$,
    \begin{align}\label{eq:TWm_def}
    T_W^{(m)}:=\partial_t^m \mathcal{F}_{W,SE}^{(2)}(t)|_{t=0},
\end{align}
and
\begin{align}
    &f_{W,(O,I)}^{(m)}= \frac{1}{2^N}\partial_t^m\tr\bigg(P_O\bigg(\mathcal{F}_{W,S}^{(2)}(t)(P_I)+\nonumber \\
    &\qquad \qquad \qquad \qquad \qquad \sum_{2 < n \leq m}^\infty \mathcal{F}_{W}^{(n)}(t)(P_I)\bigg)\bigg)\bigg|_{t=0}. \label{eq:f_definition_main}
\end{align}
\end{subequations}
We note that $T^{(m)}_W$ depends linearly on $K^{(m -2)}_{a,b}(0)$. On the other hand $f^{(m)}_{W, (O, I)}$ is a non-linear function of $K_{a,b}^{(m)}(0)$ --- 
however, $\mathcal{F}_{W, S}^{(2)}(t)$, and consequently $\partial_t^m \mathcal{F}_W^{(2)}(t) |_{t = 0}$, is independent of $K_{a,b}^{(m)}(0)$. Furthermore, since $\mathcal{F}_W^{(n)}(t)$ is obtained from $n^\text{th}$ term in the Dyson series expansion of the system-environment unitary, $\partial_t^m \mathcal{F}_W^{(n)}(t)|_{t = 0}$, for $2 < n \leq m$, depends only on $K_{a,b}^{(0)}(0), K_{a,b}^{(1)}(0)\dots K_{a,b}^{(m - 3)}(0) $. We provide an explicit expression for $f^{(m)}_{W, (O, I)}$ in Appendix \ref{app:compute_f}.

This structure in the dependence of $T^{(m)}_W$ and $f^{(m)}_{W, (O, I)}$ on the kernel derivatives $K_{a,b}^{(m)}(0)$ allows us to build a strategy that allows us to learn $K_{a,b}^{(m)}(0)$ from the experimentally measurable time-traces $B_{W, (O, I)}(t)$ [Eq.~\eqref{eq:experimental_time_traces}]: \emph{First}, construct a set of measurements specified by the input Pauli string $P_I$, output Pauli string $P_O$ as well as single-qubit unitaries $W$ such that $K_{a,b}^{(m - 2)}(0)$ can be uniquely determined from $\text{Tr}[P_O T_W^{(m)}(P_I)]/2^N$ --- note that this only involves analyzing the invertability of the linear map $K_{a,b}^{(m - 2)}(0) \to \text{Tr}[P_O T_W^{(m)}(P_I)]$. Then, experimentally measure the time-traces $B_{W, (O, I)}(t)$ --- from the second derivative of this time-trace at $t = 0$, using Eq.~\eqref{eq:time_traces_equations} for $m = 2$ and the inverse of the map $K_{a,b}^{(0)}(0) \to \text{Tr}[P_O T_W^{(2)}(P_I)]$, we immediately obtain $K_{a,b}^{(0)}(0)$. Similarly, the third derivative of the time-trace $B_{W, (O, I)}(t)$ at $t = 0$ yields $K_{a,b}^{(1)}(0)$ since $f^{(3)}_{W, (O, I)}$ depends only on $K_{a,b}^{(0)}(0)$. Proceeding similarly, we can sequentially determine $K_{a,b}^{(2)}(0), K_{a,b}^{(3)}(0) \dots$ from increasingly higher derivatives of $B_{W, (O, I)}(t)$ at $t = 0$.

It then remains to determine a set of measurements $(P_I, P_O, W)$ such that the linear map $K_{a,b}^{(m - 2)}(0) \to \text{Tr}[P_O T_W^{(m)}(P_I)]$ is invertible --- we do so in the next two subsections. We first provide a general local tomography procedure that works for any geometrically local open system model and then go onto optimize it for certain experimentally relevant settings.

 \subsubsection{\label{sec:tomography}Local tomography procedure}
Since the invertibility of the map 
\[
K_{a,b}^{(m-2)}(0) \to \text{Tr}[P_O T_W^{(m)}(P_I)] 
\]
determines the set of measurements (as specified by the input Pauli string $P_I$, output Pauli string $P_O$ and single-qubit unitary $W$), it will be convenient to obtain an explicit expression for the superoperator $T_W^{(m)}$. Note that the kernels satisfy $K_{a,b}(t)=K_{b,a}(-t)^*$ and consequently 
\begin{align*}
&\text{Re}[K_{a,b}^{(m)}(0)]=(-1)^m\text{Re}[K_{b,a}^{(m)}(0)] \text{ and},\nonumber\\
&\text{Im}[K_{a,b}^{(m)}(0)]=(-1)^{m+1}\text{Im}[K_{b,a}^{(m)}(0)].
\end{align*}
In particular, $K_{a,a}^{(m)}(0)$ is real for even $m$ and imaginary for odd $m$. Using Eq.~(\ref{eq:time_traces_equations}b), we then obtain
\begin{subequations}\label{eq:TW_explicit_expression}
\begin{align}
    &T_W^{(m)}= \sum_{c\in\mathcal{P}_{SE}}\text{Re}[K_{c,c}^{(m-2)}(0)] \tau_{W,c,c}^{(m),-}+\notag\\
    &\hspace{55pt}\text{Im}[K_{c,c}^{(m-2)}(0)]\tau_{W,c,c}^{(m),+}+\notag\\
    &\hspace{28pt}\sum_{(c,d)\in\mathcal{P}_{SE}}\text{Re}[K_{c,d}^{(m-2)}(0)](\tau_{W,c,d}^{(m),-}+(-1)^m\tau_{W,d,c}^{(m),-})+\notag\\
    &\hspace{55pt}\text{Im}[K_{c,d}^{(m-2)}(0)](\tau_{W,c,d}^{(m),+}+(-1)^{m+1}\tau_{W,d,c}^{(m),+}),\label{eq:TWm_sum}
\end{align}
where the sum over $(c, d) \in \mathcal{P}_{SE}$ runs over all unordered pairs of Paulis $c, d \in \mathcal{P}_{SE}$ appearing in the coupling Hamiltonian, $V_{SE}(t)$, and $\tau^{(m), \pm}_{W, c, d}$ are superoperators given by:
\begin{align}
    &\tau_{W,c,d}^{(m),\pm}:=-(W[P_c,[P_d,(\cdot)]_\pm]W^\dagger+\notag\\
    &\hspace{65pt}[P_c,W[P_d,(\cdot)]_\pm W^\dagger](2^m-2)+\notag\\
    &\hspace{65pt}[P_c,[P_d,W(\cdot) W^\dagger]_{\pm}])e^{i\frac{\pi}{4}(1\pm 1)}, \label{eq:tau_def}
\end{align}
\end{subequations}
with $[\cdot,\cdot]_+ := \{\cdot,\cdot\},[\cdot,\cdot]_-:=[\cdot,\cdot]$ being the anticommutator and commutator, respectively. At this point, we remark that if $W = \Id^{\otimes N}$, then $\tau_{\Id^{\otimes N}, c, c}^{(m), -} = 0$ since for any pauli string $P$, $[P, \{P, \cdot\}] = 0$. Furthermore, the coefficients of $\text{Re}[K_{c,d}^{(m-2)}(0)],\text{Im}[K_{c,d}^{(m-2)}(0)]$ in 
Eq.~\eqref{eq:tau_def} evaluate to
\begin{align}\label{eq:problematic_sum_no_W}
&\tau_{\Id^{\otimes N}, c, d}^{(m), \pm} \mp (-1)^{m} \tau_{\Id^{\otimes N}, d, c}^{(m), \pm} = -2^me^{i\frac{\pi}{4}(1\pm 1)}\times\notag\\
&\begin{cases}
    [[P_c, P_d]_\pm, (\cdot)] &,m\text{ odd,}\\
    \{[P_c, P_d]_\mp, (\cdot)\}+2(\pm P_c(\cdot) P_d-P_d(\cdot) P_c) &, m\text{ even}.
\end{cases}
\end{align}
Therefore, for $m$ odd, from Eqs.~\eqref{eq:tau_def} and \eqref{eq:problematic_sum_no_W} it follows that $T_{\Id^{\otimes N}, c, d}^{(m + 2)}$ depends only on the sums 
\begin{align}
    &\sum_{(c, d) \in \mathcal{P}_{SE}: [P_c, P_d] = Q} \text{Re}[K_{c,d}^{(m)}(0)],\notag\\
    &\sum_{(c, d) \in \mathcal{P}_{SE}: \{P_c, P_d\} = Q} \text{Im}[K_{c,d}^{(m)}(0)],\notag
    \end{align}
instead of individually on $\text{Re}[K_{c,d}^{(m)}(0)],\text{Im}[K_{c,d}^{(m)}(0)]$. Consequently, the usual measurement strategy involving short-time evolution of an initial product state and a Pauli measurement is not sufficient for learning $K_{c,d}^{(m)}(0)$. To overcome this issue, we need to additionally apply a non-identity intermediate layer of gates $W$ --- this is in contrast to the Hamiltonian or Lindbladian learning \cite{haah2024learning,stilck2024efficient}, where no intermediate gates are required. The gate $W$ that we introduce will depend on the kernel coefficient $K_{a,b}^{(m)}(0)$ that we want to learn, as summarized in three cases in Table \ref{tab:W_gate}. In the remainder of this section, we will analyze these cases one by one.

\begin{table}[]
    \centering
    \begin{tabular}{l p{4.2cm}}
        \toprule
        \textbf{Target coefficient} & \multicolumn{1}{c}{\textbf{W}} \\
        \midrule
        \textbf{Case 1:} &  \\
        \hspace{1em} $\text{Re}[K_{a,b}^{(m)}(0)]$, $m$ even & \multicolumn{1}{c}{$\Id^{\otimes N}$} \\
        \hspace{1em} $\text{Im}[K_{a,b}^{(m)}(0)]$, $m$ even &  \\
        \hspace{1em} $\text{Re}[K_{a,a}^{(m)}(0)]$, $m$ even &  \\
        \midrule
        \textbf{Case 2:} &  \\
        \hspace{1em} $\text{Re}[K_{a,b}^{(m)}(0)]$, $m$ odd & \multicolumn{1}{c}{Single-qubit Pauli $P_w$:} \\
        \hspace{1em} $\text{Im}[K_{a,b}^{(m)}(0)]$, $m$ odd & \multicolumn{1}{c}{$\{P_w,P_aP_b\}=0$}  \\
        \midrule
        \textbf{Case 3:} &  \\
        \hspace{1em} $\text{Im}[K_{a,a}^{(m)}(0)]$, $m$ odd & \multicolumn{1}{c}{Single-qubit gate $S \cdot H$ on $\mathcal{S}_a$} \\
        \bottomrule
    \end{tabular}

    \caption{Intermediate $W$ gate needed to learn $K_{a,b}^{(m)}(0)$. The choice of $W$ gate depends on if $a$ and $b$ are the same, on whether the real and imaginary part is the target coefficient and on the parity of $m$. In case 2 $P_w$ needs to commute with one of $P_a,P_b$ and anticommute with the other: since $P_a\neq P_b$ there is a site where they differ. If both are non-identity at that site, let $P_w$ be the Pauli of $P_b$ at that site. If, say, $P_b$ is identity at that site, let $P_w$ be a Pauli different from $P_a$ at that site. In case 3 the gate $SH$ can be applied on any site in the support of $P_a$. Note that since $K_{a, b}(t) = K_{b, a}^*(-t)$, $\text{Re}[K_{a, a}^{(m)}(0)] = 0$ if $m$ is odd and $\text{Im}[K_{a, a}^{(m)}(0)]= 0$ if $m$ is even, and thus do not appear in the cases listed above.}
    \label{tab:W_gate}
\end{table}

\textbf{Case 1, $m$ is even}. For $m$ even, we do not need the intermediate gate layer i.e., we can set $W = \Id^{\otimes N}$. From Eq.~\eqref{eq:TWm_sum}, it follows that
\[
T_{\Id^{\otimes N}}^{(m)} = \sum_{P, Q } \xi_{P,Q}^{\Id^{\otimes N},(m)} P (\cdot) Q,
\]
where either $P, Q \in \mathcal{P}_{SE}$, $P = Q = \Id^{\otimes N}$ or $(P, Q) = ([P_c, P_d]_\pm, \Id^{\otimes N})$ or $ (\Id^{\otimes N},[P_c, P_d]_\pm)$ for $P_c, P_d \in \mathcal{P}_{SE}$. Furthermore, for $P_a, P_b \in \mathcal{P}_{SE}$, 
\begin{align}\label{eq:xi_case1}
    &\xi_{P_a, P_b}^{\Id^{\otimes N},(m)} = 2^{m+1} (\text{Re}[K_{a,b}^{(m - 2)}(0)]-i\text{Im}[K_{a,b}^{(m - 2)}(0)]), \nonumber \\
    &\xi_{P_b, P_a}^{\Id^{\otimes N},(m)} = 2^{m+1} (\text{Re}[K_{a,b}^{(m - 2)}(0)]+i\text{Im}[K_{a,b}^{(m - 2)}(0)]),
\end{align}
We can thus recover $\text{Re}[K_{a,b}^{(m - 2)}(0)],\text{Im}[K_{a,b}^{(m - 2)}(0)]$ provided we can learn $\xi_{P_a, P_b}^{\Id^{\otimes N},(m)},\xi_{P_b, P_a}^{\Id^{\otimes N},(m)}$. To learn the latter from $\text{Tr}[P_O T^{(m)}_{\Id^{\otimes N}}(P_I)]$, a possible approach is to perform process tomography on $T^{(m)}_{\Id^{\otimes N}}$. We remind the reader that performing full process tomography on an $N-$qubit hermiticity-preseving superoperator $\mathcal{M}$ expressible as $\mathcal{M}(X)= \sum_{P, Q  \in \mathcal{P}_N} s_{P, Q} P X Q$ would allow us to unambiguously learn the coefficients $s_{P, Q}$ \cite{chuang1997prescription}. However, since $T^{(m)}_{\Id^{\otimes N}}$ is a $N-$qubit superoperator, full process tomography will require $\sim \text{exp}(N)$ measurements as well as $\sim \text{exp}(N)$ classical postprocessing time \cite{chuang1997prescription}. Instead, by using the locality of the Pauli strings appearing in $T^{(m)}_{\Id^{\otimes N}}$ we can reduce both the number of measurements as well as the classical postprocessing time to $O(N)$, as we discuss next.

\begin{figure}
    \centering
\includegraphics[width=1.0\linewidth]{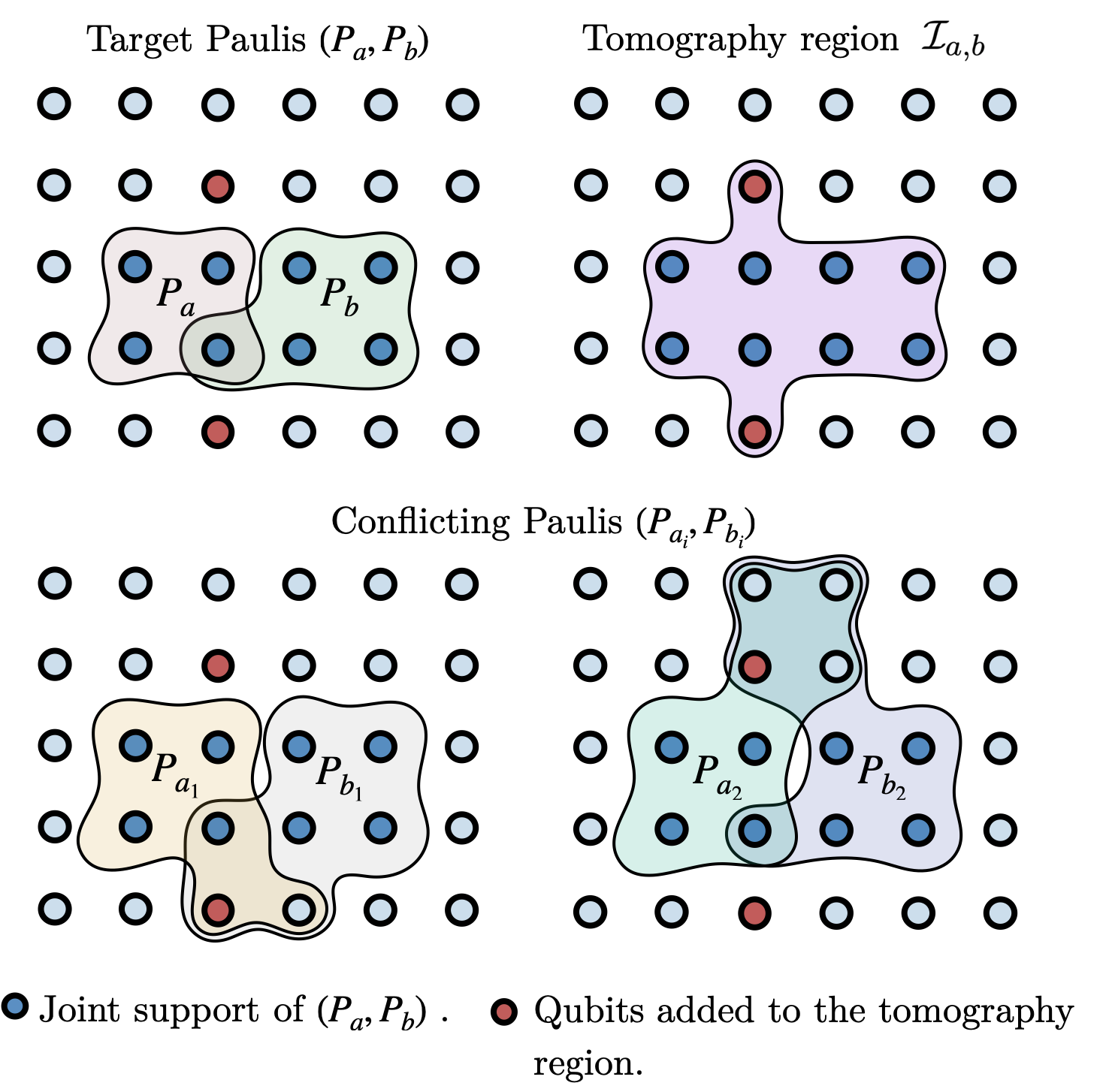}
    \caption{We show how to find a tomography support $\mathcal{I}_{a,b}$ that allows us learn the derivatives $K_{a,b}^{(m)}(0)$. The target Pauli strings $P_a,P_b$, shown in the top left panel, act on the joint support $\mathcal{S}_{a,b}$ (dark blue). The conflicting pairs of Pauli strings $(P_{a_1},P_{b_1}),(P_{a_2},P_{b_2})$, depicted in the bottom panels, satisfy Eqs.\eqref{eq:enlarged_support_cond_1},\eqref{eq:enlarged_support_cond_2}, so their kernel coefficients would mix with those of $K_{a,b}^{(m)}(0)$ if we only performed tomography on $\mathcal{S}_{a,b}$. We add one extra qubit from the region $\mathcal{S}_{a,b}^c$ that lies in the support of each pair (red), obtaining $\mathcal{Q}_{a,b}$. Our tomography region is $\mathcal{I}_{a,b}=\mathcal{S}_{a,b}\cup\mathcal{Q}_{a,b}$, the dark blue and red sites in the top right panel.}
    \label{fig:tomography_info_complete}
\end{figure}

Suppose we want to learn $\xi_{P_a, P_b}^{\Id^{\otimes N},(m)}$ for some distinct $P_a, P_b \in \mathcal{P}_{SE}$ without performing process tomography on all the qubits --- a first guess would be to perform process tomography on $\mathcal{S}_{a,b} = \text{supp}(P_a) \cup \text{supp}(P_b)$. However, there could be Pauli strings $P_c, P_d \in \mathcal{P}_{SE}$ such that 
\begin{align}\label{eq:enlarged_support_cond_1}
P_c |_{\mathcal{S}_{a, b}} = P_a|_{\mathcal{S}_{a, b}}, P_d|_{\mathcal{S}_{a, b}} = P_b|_{\mathcal{S}_{a, b}}, 
\end{align}

\noindent i.e., $P_c$ and $P_d$ are identical to $P_a$ and $P_b$ respectively on the region $\mathcal{S}_{a, b}$ and 
\begin{align}\label{eq:enlarged_support_cond_2}
P_c |_{\mathcal{S}^c_{a, b}} = P_d |_{\mathcal{S}^c_{a, b}},
\end{align}
i.e., $P_c$ and $P_d$ are identical to each other outside $\mathcal{S}_{a,b}$. Equivalently, these are the Pauli strings such that $\text{Tr}(P_O P_c P_I P_d)= \text{Tr}(P_O P_a P_I P_b)$ for any $P_I, P_O$ supported only on $\mathcal{S}_{a, b}$ and, consequently, we cannot distinguish the contribution of $P_c,P_d$ in $T^{(m)}_{\Id^{\otimes N}}$ from that of $P_a,P_b$ by tomography on $\mathcal{S}_{a,b}$. We will denote by $\mathcal{Y}_{\mathcal{S}_{a,b}}(a,b)$ the set of such pairs of Pauli strings i.e., Pauli strings in $\mathcal{P}_{SE}\setminus \{a,b\}$ that can't be distinguished from $P_a,P_b$ by tomography on $\mathcal{S}_{a,b}$. Indeed, performing process tomography on $\mathcal{S}_{a,b}$ can only unambiguously yield 
\begin{align}
    &\xi_{P_a,P_b}^{\Id^{\otimes N},(m)}+\sum_{(P_c, P_d) \in \mathcal{Y}_{\mathcal{S}_{a,b}}(a,b)}\xi_{P_c, P_d}^{\Id^{\otimes N},(m)},\notag
\end{align}
instead of $\xi_{P_a, P_b}^{\Id^{\otimes N},(m)}$. Since pairs of Pauli strings satisfying Eqs.\eqref{eq:enlarged_support_cond_1} and \eqref{eq:enlarged_support_cond_2} prevent us from estimating the desired coefficient directly, we call them \textit{conflicting} pairs of Pauli strings. Our strategy will be to slightly enlarge the support of the tomography region to $\mathcal{I}_{a,b}$ such that $|\mathcal{Y}_{\mathcal{I}_{a,b}}(a,b)|=\emptyset$, i.e. there are no further conflicting Pauli strings. To construct $\mathcal{I}_{a,b}$, we pick the smallest set of qubits $\mathcal{Q}_{a,b} $ outside $\mathcal{S}_{a, b}$ such that for every pair of Paulis $(P_c, P_d)\in \mathcal{Y}_{\mathcal{S}_{a,b}}(a,b)$, $\mathcal{Q}_{a, b}$ contains a qubit in the support of $P_c$ or $P_d$ but outside $\mathcal{S}_{a, b}$. Performing tomography on
\begin{align}
    &\mathcal{I}_{a,b}:=\mathcal{S}_{a, b} \cup \mathcal{Q}_{a, b }, \label{eq:Iab_def}
\end{align}
then allows us to uniquely obtain $\xi_{P_a,P_b}^{\Id^{\otimes N},(m)}$ since the additional qubits in $\mathcal{Q}_{a, b}$ allow us to distinguish the target Pauli pair $(P_a, P_b)$ from the Pauli string pairs that originally were conflicting with it. An example of constructing $\mathcal{I}_{a, b}$ is illustrated in Figure \ref{fig:tomography_info_complete}: here, the tomography region $\mathcal{S}_{a, b}$ is enlarged with one site from each conflicting pair of Pauli strings. We show in  Appendix \ref{app:measurement_supports} that the number of additional qubits $|\mathcal{Q}_{a, b}|\leq |\mathcal{Y}_{\mathcal{S}_{a,b}}(a,b)|$ is upper bounded by $\mathfrak{d}$. Since $|\mathcal{S}_{a,b}|=O(k_{SE})$ and we add one extra qubit from each pair in $\mathcal{Y}_{\mathcal{S}_{a,b}}(a,b)$, we get that the region $\mathcal{I}_{a,b}$ has size $|\mathcal{I}_{a,b}|=O(k_{SE})+O(\mathfrak{d})$.

Recall that for $m$ even, $K_{a,a}^{(m)}(0)$ is real, and we can learn it from $\xi_{P_a,P_a}^{\Id^{\otimes N},(m)} = 2^{m+1}\text{Re}[K_{a,a}^{(m)}(0)]$. In this case we only need to enlarge the support by $|\mathcal{Q}_{a,a}|\leq \mathfrak{d}$ qubits.

\textbf{Case 2, $m$ is odd and $a \neq b$}. For $m$ odd, we will introduce an intermediate unitary $W$. We will consider $W$ to be a single-site Pauli matrix, with the corresponding Pauli string $P_w$ that is identity everywhere except at that site. Suppose we want to learn $\text{Re}[K_{a,b}^{(m)}(0)],\text{Im}[K_{a,b}^{(m)}(0)], a\neq b$. Then we can pick $P_w$ that anticommutes with $P_a$ and commutes with $P_b$: if there is a site in $P_a$ that does not overlap with $P_b$, let $W$ be a Pauli matrix different from $P_a$ at that site; if $P_a,P_b$ overlap everywhere, since $P_a\neq P_b$ there is a site where they differ, so let $W$ be $P_b$ at that site. We check whether two Pauli strings $P_a,P_w$ commute or anticommute using the function $\chi(a,w)=0$ if $[P_a,P_w]=0$ and $\chi(a,w)=1$ if $\{P_a,P_w\}=0$. Notice that conjugating by $P_w$ only adds a phase: $P_wP_aP_w^\dagger = (-1)^{\chi(a,w)}P_a$. Then, if we use an input Pauli $P_I$, the coefficients of $\text{Re}[K_{a,b}^{(m)}(0)],\text{Im}[K_{a,b}^{(m)}(0)]$ in Eq.~\eqref{eq:tau_def} evaluate to
\begin{align}\label{eq:sum_W}
&(\tau_{P_w a, b}^{(m), \pm} \mp (-1)^{m} \tau_{P_w, b, a}^{(m), \pm})(P_I) = (-1)^{\chi(w,I)+\chi(w,b)}\times\notag\\
&\hspace{40pt}\frac{2^m-2}{2^{m-1}}(\tau_{\Id^{\otimes N}, a, b}^{(m-1), \pm} \mp (-1)^{m-1} \tau_{\Id^{\otimes N}, b, a}^{(m-1), \pm})(P_I).
\end{align}
Thus, introducing such a $W$ reduces the problem to {case 1}, up to phases that can be corrected in postprocessing. From Eq.~\eqref{eq:TWm_sum}, it follows that
\[
T_{P_w}^{(m)}(P_I) = (-1)^{\chi(w,I)}\sum_{P, Q } \xi_{P,Q}^{P_w,(m)} P P_I Q, \label{eq:T_Pw_simplified}
\]
where either $P, Q \in \mathcal{P}_{SE}$, $P = Q = \Id^{\otimes N}$ or $(P, Q) = ([P_c, P_d]_\pm, \Id^{\otimes N})$ or $ (\Id^{\otimes N},[P_c, P_d]_\pm)$ for $P_c, P_d \in \mathcal{P}_{SE}$. In this case, notice that before we invert this system of equations we need to correct the samples $\frac{1}{2^N}\tr(P_OT_{P_w}^{(m)}(P_I))\mapsto \frac{1}{2^N}\tr(P_OT_{P_w}^{(m)}(P_I))(-1)^{\chi(w,I)}$. Then we can obtain the coefficients for $P_a,P_b\in \mathcal{P}_{SE}$,
\begin{align}\label{eq:xi_case2}
    &\xi_{P_a, P_b}^{P_w,(m)} = (2^{m+1}-4) (-1)^{\chi(w,b)}\times \notag\\
    &\hspace{50pt}(\text{Re}[K_{a,b}^{(m - 2)}(0)]-i\text{Im}[K_{a,b}^{(m - 2)}(0)]),
\end{align}
and $\xi_{P_b, P_a}^{P_w,(m)} =  \xi_{P_a, P_b}^{P_w,(m),*}$. We can thus recover $\text{Re}[K_{a,b}^{(m - 2)}(0)],\text{Im}[K_{a,b}^{(m - 2)}(0)]$, for $a\neq b$ and $m$ odd, from $\xi_{P_a, P_b}^{P_w,(m)},\xi_{P_b, P_a}^{P_w,(m)}$. From Eq.\eqref{eq:T_Pw_simplified} one can see that the same analysis of the region of tomography carried out in {case 1} holds, so it suffices to do tomography on $\mathcal{I}_{a,b}$.
\begin{figure}
    \centering
    \includegraphics[width=0.9\linewidth]{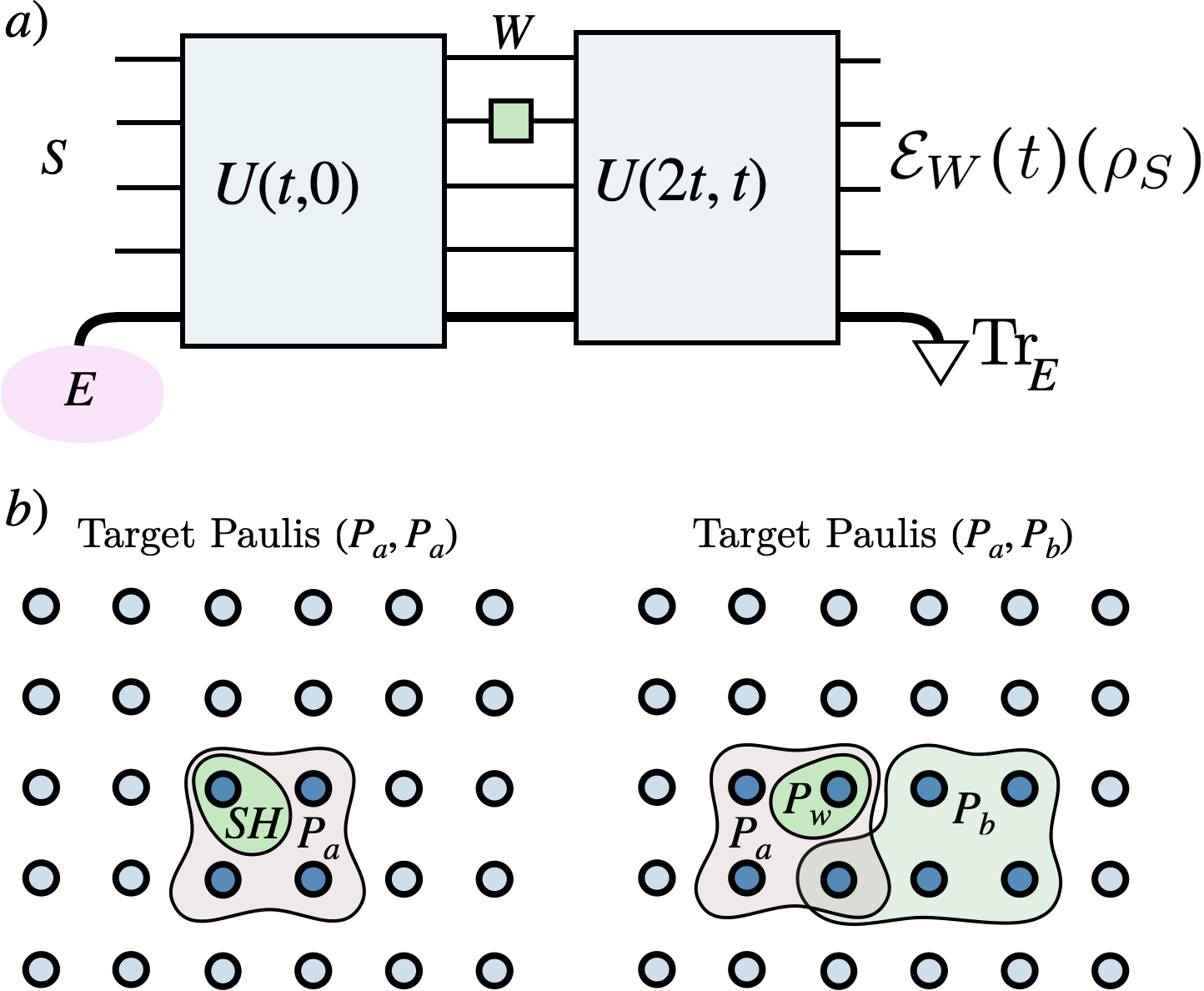}
    \caption{a) A single-qubit gate $W$ is applied in the middle of the time evolution. b) (Left) If $P_a=P_b$, $W$ can be chosen to be $S\cdot H$ on any site in the support of $P_a$. (Right) If $P_a\neq P_b$, $W$ can be chosen to be a Pauli string $P_w$ that is non-identity only on one site of $P_a$ or $P_b$, such that it commutes with one and anticommutes with the other.}
    \label{fig:site_W}
\end{figure}

\textbf{Case 3, $m$ is odd and $a = b$}. For $m$ odd, $\text{Im}[K_{a,a}^{(m)}(0)]$ vanishes from $T_{P_w}^{(m)}$ for the same reason as in $T_{\Id^{\otimes N}}^{(m)}$. Instead, we will pick a site on the support of $P_a$ and choose $W$ to be a single-qubit Clifford gate on that site of the form $S\cdot H$, a product of the phase gate $S$ and the Hadamard $H$. This has the nice property of cycling through non-identity Pauli matrices upon conjugation:
\begin{align}
    & \sigma^x \to \sigma^z \to \sigma^y \to \sigma^x,
\end{align}
where $\to$ applies $SH(\cdot)(SH)^\dagger$. The coefficient of $\text{Im}[K_{a,a}^{(m)}(0)]$ in Eq.\eqref{eq:TW_explicit_expression} evaluates to
\begin{align}
    &\tau_{SH,a,a}^{(m),+}(P_I)=-i(2^m-2)\times \notag\\
    &\hspace{40pt}(\{P_a\tilde{P}_a,\tilde{P}_I\} + P_a\tilde{P}_I\tilde{P}_a-\tilde{P}_a\tilde{P}_IP_a),
\end{align}
where we introduced the notation $\tilde{P}_a = WP_aW^\dagger$, which is a Pauli string, since $W$ is a Clifford gate. Our particular choice of $W$ means $\tilde{P}_a$ has no global phase: it is equal to $P_a$ except at one site, where the Pauli is shifted as above. We want to obtain $\text{Im}[K_{a,a}^{(m)}(0)]$ from the coefficient of $P_a\tilde{P}_I\tilde{P}_a$, but there might be other terms of this form in $T_{SH}^{(m)}(P_I)$. The rest of coefficients in Eq.\eqref{eq:TW_explicit_expression} evaluate to
\begin{align}
    &(\tau_{SH,c,d}^{(m),\pm}\pm\tau_{SH,d,c}^{(m),\pm} )(P_I) = -e^{i\frac{\pi}{4}(1\pm1)}\times \notag\\
    &\hspace{6pt}([[\tilde{P}_c,\tilde{P}_d]_\pm,\tilde{P}_I]+[[P_c,P_d]_\pm,\tilde{P}_I]+\notag\\
    &\hspace{10pt}(2^m-2)([P_c,[\tilde{P}_d,\tilde{P}_I]_\pm]\pm [P_d,[\tilde{P}_c,\tilde{P}_I]_\pm])),
\end{align}
so the terms $P_c\tilde{P}_I\tilde{P}_d,\tilde{P}_d\tilde{P}_IP_c$,$P_d\tilde{P}_I\tilde{P}_c,\tilde{P}_c\tilde{P}_IP_d$ might mix with $P_a \tilde{P}_I\tilde{P}_a$. The first and third term require $P_c =P_a,P_d=P_a$, but $P_c\neq P_d$, so this term does not appear. The second and third are satisfied for $P_c=WP_aW^\dagger,P_d=W^\dagger P_aW$ and vice versa (recall that the pair $P_c,P_d$ is unordered). From Eq.\eqref{eq:TW_explicit_expression}, it follows that
\begin{align}
    &T_{SH}^{(m)}(P_I)=\sum_{P,Q}\xi_{P,Q}^{SH,(m)}P\tilde{P}_IQ,
\end{align}
where either $P,W^\dagger QW\in\mathcal{P}_{SE}$, $W^\dagger PW,Q\in\mathcal{P}_{SE}$ $P=Q=\Id^{\otimes N}$ or $(P,Q)=(RS,\Id^{\otimes N})$ or $(\Id^{\otimes N},RS)$ where either $R,W^\dagger RW,S,W^\dagger SW$ are in $\mathcal{P}_{SE}$. Note that $\tilde{P}_I$ on the right hand side requires that either we modify the inversion procedure to include the extra map $W(\cdot)W^\dagger$ or we can simply prepare $W^\dagger P_I W$ instead of $P_I$. We recover $\text{Im}[K_{a,a}^{(m)}(0)]$ from the coefficients\\
\begin{align}\label{eq:xi_case3}
    &\xi_{P_a,\tilde{P}_a}^{SH,(m)}=(2^m-2)(-\text{Re}[K_{c,d}^{(m)}(0)]+\notag\\
    &\hspace{20pt}i(\text{Im}[K_{c,d}^{(m)}(0)]-\text{Im}[K_{a,a}^{(m)}(0)])),
\end{align}
and $\xi_{P_a,\tilde{P}_a}^{SH,(m)}=\xi_{\tilde{P}_a,P_a}^{SH,(m),*}$, where we let $(P_c,P_d)=(WP_aW^\dagger,W^\dagger P_a W)$. While we can remove $\text{Re}[K_{c,d}^{(m)}(0)]$ using the conjugate coefficient, we need to manually remove $\text{Im}[K_{c,d}^{(m)}(0)]$ in classical postprocessing. Since $P_c\neq P_d$, we will already have learned this coefficient from {case 2}. The tomography region to learn $\xi_{P_a,\tilde{P}_a}^{SH,(m)}$ is $\mathcal{I}_{a,\tilde{a}}$, where $\tilde{a}$ indicates Pauli string $\tilde{P}_a$. This, again, satisfies $|\mathcal{I}_{a,\tilde{a}}|\leq O(k_{SE})+O(\mathfrak{d})$, as detailed in Appendix \ref{app:measurement_supports}. 

We have seen that, in order to learn the derivatives of $K_{a,b}(t)$ at $t=0$ we need to introduce a gate $W$ halfway through time evolution and perform tomography on a region $\mathcal{I}_{a,b}$. While \textit{process} tomography on $\mathcal{I}_{a,b}$ allows us to learn the Hermiticity-preserving maps $T_W^{(m)}$\cite{chuang1997prescription}, it is sufficient to use a procedure with the sample requirement of \textit{state} tomography, as shown in Lemma \ref{lem:Ainverse}. The key is that we do not need to learn all the coefficients in $T_W^{(m)}$ because we are only interested in $\xi_{P,Q}^{W,(m)},\xi_{Q,P}^{W,(m)}$, for $P,Q$ as described in Eqs.\eqref{eq:xi_case1},\eqref{eq:xi_case2},\eqref{eq:xi_case3}. Thus, we only perform the necessary measurements to learn these coefficients, so the number of configurations $(\rho,O,W)$ that we need to prepare to learn the derivatives of $K_{a,b}(t)$ is at most $ 3^{|\mathcal{I}_{a,b}|}\cdot 2$, where $|\mathcal{I}_{a,b}|\leq O(k_{SE})+O(\mathfrak{d})$. In particular, this does not depend on the locality of the system Hamiltonian, $k_S$, which could be much larger.

\subsubsection{One- and two-qubit jump operators}
It is common for dissipation in experimental settings to be given by few-body interaction terms, so we consider now the case that $\mathcal{P}_{SE}$ contained at most two-qubit Pauli strings. Then, if either $P_a$ or $P_b$ was at least $2$-local, there would be no pairs of Pauli strings in $\mathcal{P}_{SE}$ satisfying Eqs.\eqref{eq:enlarged_support_cond_1}\eqref{eq:enlarged_support_cond_2}, since they would need to be at least $3$-local, so it would suffice to measure on $\mathcal{S}_{a,b}$. If $P_a,P_b$ are both single-qubit it is possible to have such pairs, as illustrated in Figure \ref{fig:example_single_qubit}.

\begin{figure}
    \centering
    \includegraphics[width=0.75\linewidth]{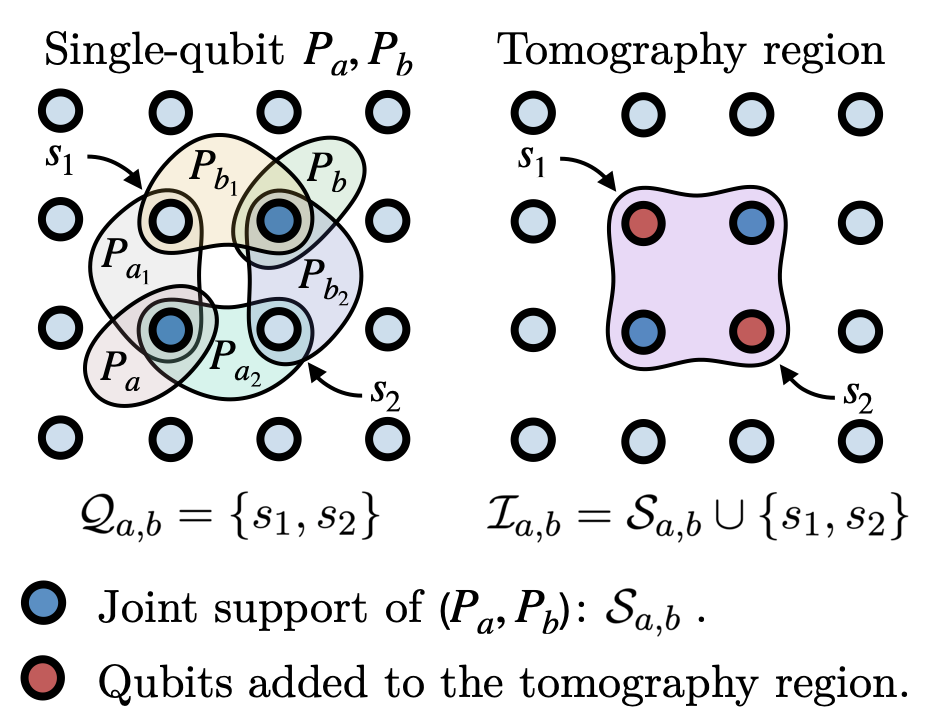}
    \caption{Measurement supports to learn the derivatives of a kernel $K_{a,b}(t)$ with $P_a,P_b$ single-qubit Paulis when $\mathcal{P}_{SE}$ have at most $2$-qubit jump operators. The pairs $(P_{a_1},P_{b_1}),(P_{a_2},P_{b_2})$ consist of $2$-qubit jump operators satisfying Eqs.\eqref{eq:enlarged_support_cond_1}, \eqref{eq:enlarged_support_cond_2}. The extra qubits that need to be added to the tomography region are $\mathcal{Q}_{a,b}=\{s_1,s_2\}$, so the tomography region is $\mathcal{I}_{a,b}=\mathcal{S}_{a,b}\cup \mathcal{Q}_{a,b}$.}
    \label{fig:example_single_qubit}
\end{figure}

\subsubsection{Parallelizing the measurements}
\begin{figure}
    \centering
    \includegraphics[width=0.99\linewidth]{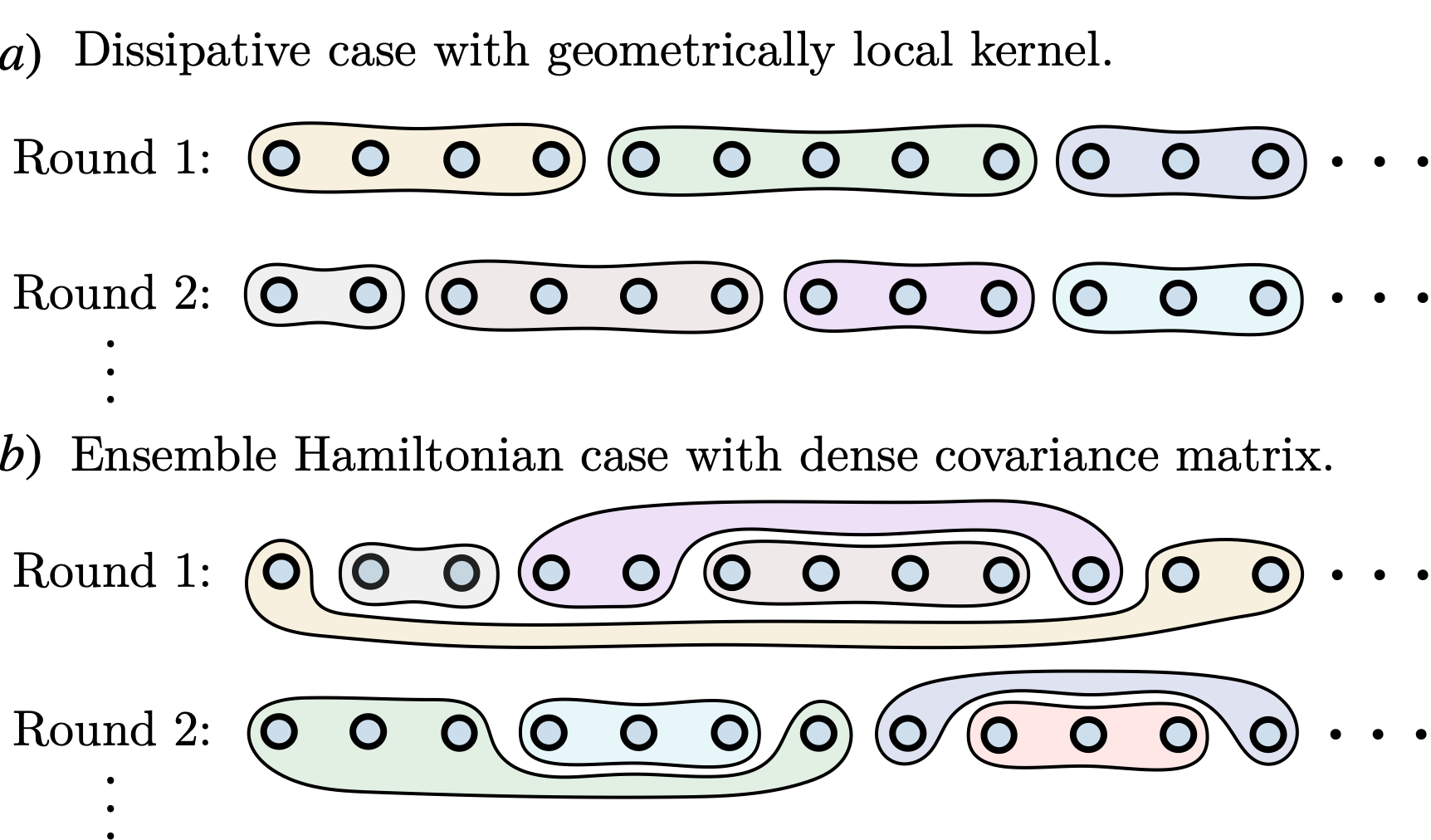}
    \caption{Parallelized measurements in a 1D system. a) Dissipative case, with geometrically-local kernels. Tomography on regions $\mathcal{I}_{a,b},\mathcal{I}_{c,d},...$ that do not overlap can be performed simultaneously. In order to perform tomography on all required regions, several rounds of measurement with different regions will be necessary. By geometric locality, each $\mathcal{I}_{a,b}$ only overlaps with a constant number of regions, so the total number of rounds is independent on the system size. b) Ensemble Hamiltonian case, with dense covariance matrix. Since the covariance between terms acting on far away regions of the lattice is nonzero, the resulting tomography regions within each round will not be geometrically local. 
    Instead of a sparse kernel we now have a dense covariance matrix with $O(N^2)$ parameters, so the number of rounds we need to perform in this case grows like $O(N)$.}
    \label{fig:parallelization}
\end{figure}
For tomography on a region $\mathcal{I}_{a,b}$ we only need to prepare states, introduce $W$ and measure observables all within the support $\mathcal{I}_{a,b}$. If we have several regions $\mathcal{I}_{a,b},\mathcal{I}_{c,d},\mathcal{I}_{e,f},...$ that don't overlap we can estimate their observables simultaneously, as they act on different supports and, thus, commute. As shown in Figure \ref{fig:parallelization} for the $1$D case we can perform several rounds of measurement, measuring several non-overlapping regions at each round, until we have performed tomography on all necessary regions. \\

In the dissipative case, with geometrically-local kernels, we know that each Pauli string $P_a,P_b\in \mathcal{P}_{SE}$ is supported on at most $k_{SE}$ sites and kernels with Pauli strings that are far away vanish, so we get an effective diameter $a_0=O(k_{SE})$ for $\mathcal{S}_{a,b}$ and for $\mathcal{I}_{a,b}$. Each $\mathcal{I}_{a,b}$ then overlaps with only a constant number of regions, so the number of necessary rounds of measurment is independent of $N$, as detailed in Appendix \ref{app:parallelization}.\\

In the ensemble Hamiltonian case, while the Pauli strings $P_a,P_b\in \mathcal{P}_{SE}$ are geometrically local, we consider the case where there might be all-to-all correlated errors, i.e. the covariance matrix $\Sigma$ is dense. This means that even if $P_a,P_b$ act on distant regions of the lattice, $\Sigma_{a,b}\neq 0$, so each row of $\Sigma$ has $O(N)$ nonzero terms. Then each
$\mathcal{I}_{a,b}$ overlaps with $O(N)$ other $\mathcal{I}_{c,d}$, so we need to perform $O(N)$ measurement rounds.

\section{\label{sec:numerics}Numerical study}

\subsection{\label{sec:numerics_nonmarkovian} Free fermions coupled to a bath with memory}
We consider a system $S$ of $N$ fermionic modes with annihilation operators $a_1,...,a_N$:
\begin{align}
    H_S&=-J\sum_{i=1}^{N-1} (a_i^\dagger a_{i+1}+ a_i^\dagger a_{i+1}^\dagger) +h.c.+2h\sum_{i=1}^Na_i^\dagger a_i. \label{eq:Ising_chain}
\end{align}
The system is bilinearly coupled to an environment $E$ consisting of $N$ fermionic modes with annihilation operators $b_{1},...,b_{N}$ through a Hamiltonian $V_{SE}$ and the environment dynamics is given by on-site dissipation $\mathcal{D}$:
\begin{align}
&V_{SE}=\sum_{i,j=1}^Nv_{ij}a_i(b_{j}^\dagger+b_{j}) +h.c.,\\
    &\mathcal{D}(\rho)=\sum_{i=1}^N\gamma_i(b_{i}\rho b_{i}^\dagger-\frac{1}{2}\{b_{i}^\dagger b_{i},\rho\}).
\end{align}
The bath is initially assumed to be in the vacuum, $\gamma = |0\rangle\langle 0|$. After tracing out the bath, the system dynamics follows a non-Markovian evolution \cite{tamascelli2018nonperturbative}, resulting in the memory kernels:
\begin{align}
    &K_{ij}(t)=\sum_{l=1}^Nv_{il}^*v_{jl}\tr(a_l^\dagger(t)a_l(0)\gamma)
\end{align}
\begin{figure}
    \centering
    \includegraphics[width=1\linewidth]{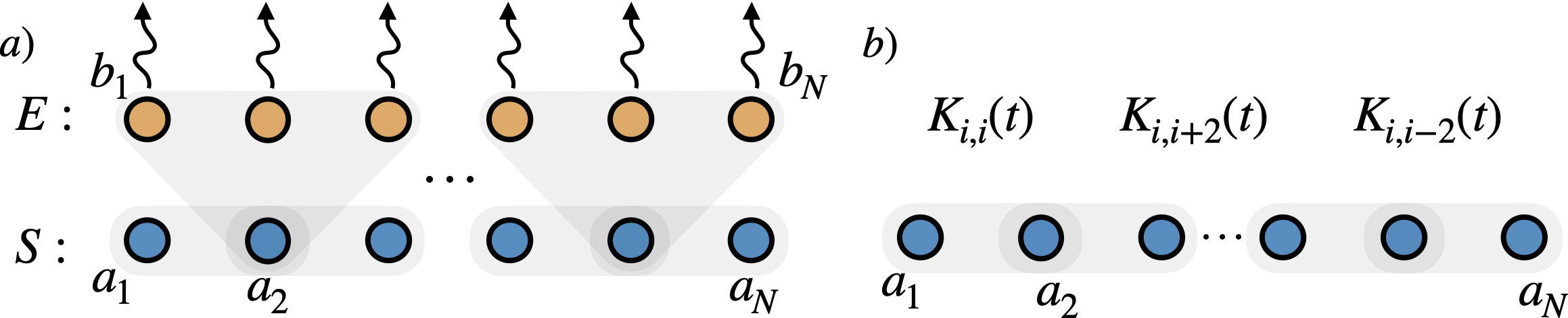}
    \caption{Fermionic model leading to non-Markovian dynamics on the system $S$. a) The fermionic modes of the system interact locally among themselves and with the environment modes $E$. The environment modes undergo dissipation. This leads to effective non-Markovian dynamics on the system, i.e. upon tracing out the environment. b) The environment modes have been traced out and their data is equivalently specified by the kernels $K_{ij}(t)$.}
    \label{fig:numerics_fig}
\end{figure}
We will choose $v_{ij}=\delta_{i+1,j}-\delta_{i,j+1}$. Thus, we obtain interactions between the system and the environment as given in Figure \ref{fig:numerics_fig} (left) and if we trace the environment, we can equivalently describe it by the kernels (right). Since the Hamiltonian is free-fermionic and the jump operators are linear in the creation/annihilation operators, time traces of observables can be simulated efficiently \cite{bravyi2011classical}. 

Following the learning algorithm outlined in the main text, we recover the coefficients of the system and kernel derivatives at $t=0$. We measure the observable for $10$ evenly-spaced timesteps between $t_{\min}=10^{-3}$ and $t_{\max}=10^{-1}$ for a certain number of shots $S$, fit a polynomial of degree $3$ and estimate the desired parameters. The absolute error in the estimated parameters decreases like $1/\sqrt{S}$ with the number of shots  per timestep $S$ in Figure \ref{fig:recovered_error_nonmarkovian} b). Since the scaling with $1/\sqrt{S}$ is the same regardless of the system size, as seen  in Figure \ref{fig:recovered_error_nonmarkovian} c), the absolute error in the estimate parameters does not increase with system size, as expected. The logarithmic scaling with the system size in Proposition \ref{prop:non_markovian_main} arises from the estimation of $O(N)$ parameters to a desired precision. In this numerical example we assume translational invariance for the system parameters and the kernels, so we only estimate one instance of each. The scaling we obtain predicts comparable accuracies in $h,J$ as obtained for Lindbladian learning \cite{stilck2024efficient}, as both Hamiltonian learning and Lindbladian learning are based on estimating the first derivative of the observable time trace. The constant term of the kernel and its first deriative at $t=0$ are obtained by estimating the second and third derivative of the observable time trace at $t=0$, which are harder to estimate with the same number of samples \cite{stilck2024efficient}, so we obtain coarser estimates. From these derivatives of the kernels we can estimate the environment bandwidth $\gamma$ and the system-environment coupling $v$ using Eq. \eqref{eq:K_ansatz_MFD}. As we can see in Figure \ref{fig:recovered_coefs_nonmarkovian}, the errors in the recovered coefficients do not scale with the system size $N$.

\begin{figure*}
    \centering
    \includegraphics[width=\linewidth]{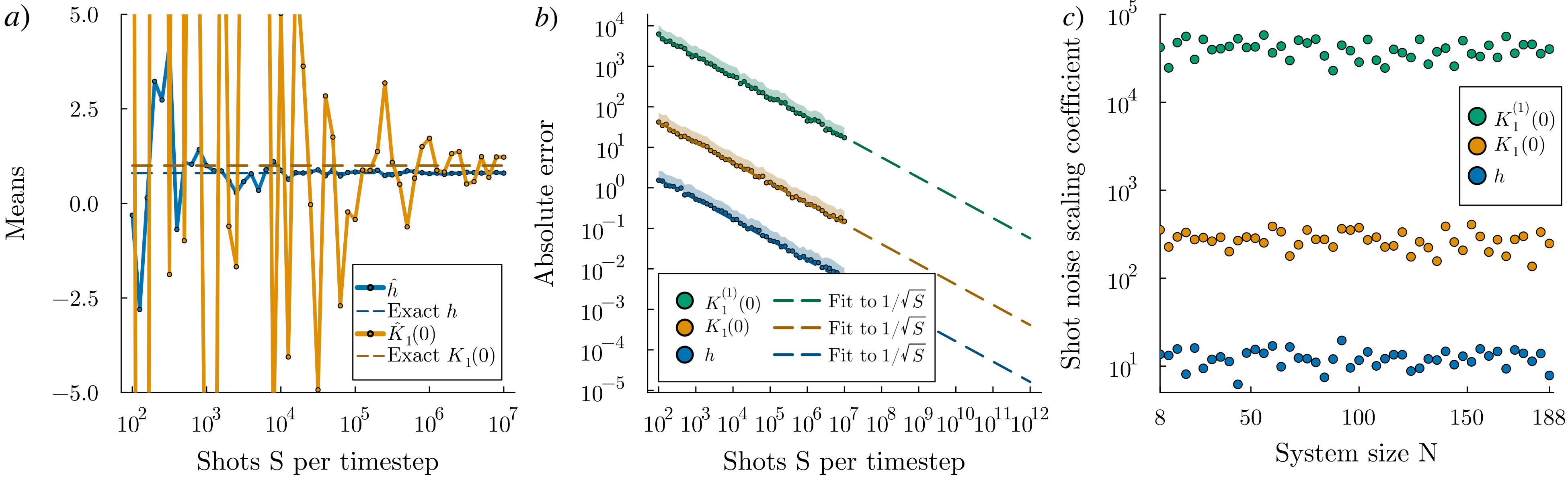}
    \caption{Error in the recovered coefficients for the transverse-field Ising model coupled to a non-Markovian bath, with coefficients $J=0.2,h=0.8, v_i=1,\gamma_i=0.9$. The values of the time traces are computed exactly for $10$ evenly spaced timesteps between $t_{\min}=10^{-3}$ and $t_{\max}=10^{-1}$. Then, for each timestep $S$ shots are drawn from the Bernoulli distribution, obtaining a noisy time trace with projection noise that is fit to a polynomial to estimate the desired parameter. The shaded regions indicate the variance in the absolute error for a fixed $S$, obtained by estimating the error in each parameter for $100$ different batches. a) Convergence of the estimate for $h$ and for the constant term of a kernel, $K_1(0)$, with the number of shots per timestep $S$. We fixed $N=120$. b) Absolute error in $h$ ($|h|=0.8$), the constant term, $K_1(0)$ ($|K_1(0)|=1$), and the first derivative of a kernel at $t=0$, $K_1^{(1)}(0)$ ($|K_1^{(1)}(0)|=0.45$) using $S$ shots per timestep and fixing $N=120$. The dashed lines represent fits to $1/\sqrt{S}$, as expected from Hoeffding's inequality. c) 
    More precisely, from plot $b)$ we expect that the absolute error will scale like $\sigma/\sqrt{S}$, where the standard deviation $\sigma$ is the shot noise scaling coefficient. We plot it for the  magnetic field, $h$, the constant term, $K_1(0)$, and the first derivative of the kernel $K_1^{(1)}(0)$ and see that it does not depend on the system size. Note that, since we assume translational invariance, we only need to estimate one instance of each desired parameter. In Proposition \ref{prop:non_markovian_main} we see a logarithmic scaling with $N$ because we want to estimate all $O(N)$ parameters correctly.}
    \label{fig:recovered_error_nonmarkovian}
\end{figure*}

\begin{figure}
    \centering
    \includegraphics[width=0.9\linewidth]{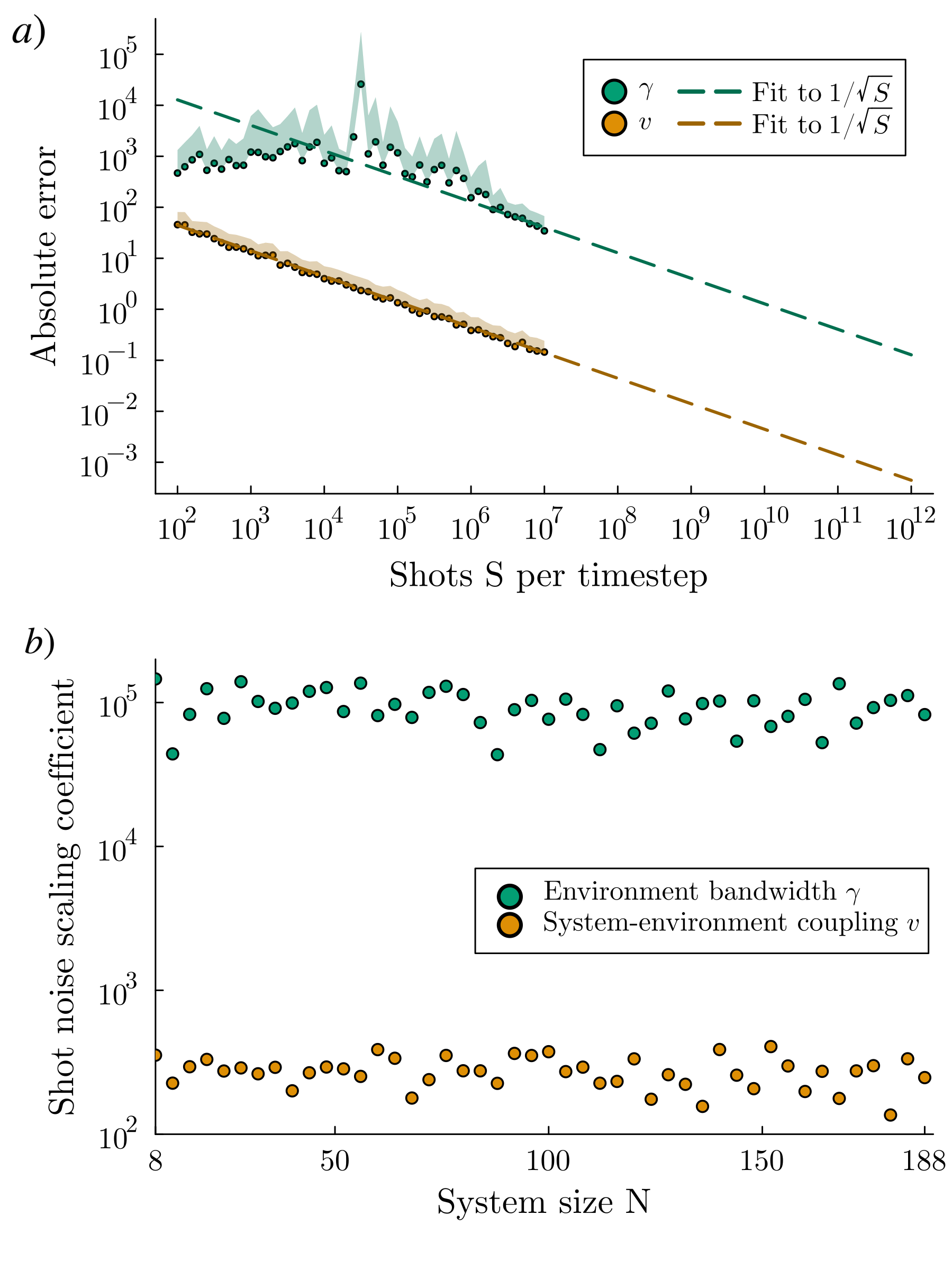}
    \caption{a) Error in the estimates of the coupling strength between system and environment, $v=1$, and the decay rate, $\gamma = 0.9$. The estimates are obtained from the translationally-invariant system in Figure \ref{fig:recovered_error_nonmarkovian} using Eq. \eqref{eq:K_ansatz_MFD}: $v$ is obtained from $K_1(0)$, while $\gamma$ is obtained from $K_1^{(1)}(0)$ using the estimate for $v$. The dashed lines represent fits to $1/\sqrt{S}$: while the error in $v$ decreases as expected, the error in $\gamma$ only decreases once we have an estimate for $v$ that is commensurate with its exact value. We thus fit the error for $\gamma$ once the error in $v$ goes below $\frac{v}{2}$, in this case $S\in[10^6,10^7]$, and see the expected decrease. The system size is $N=120$.
    b) The offset in the graphs of a), i.e. the scaling coefficients of the fit $1/\sqrt{S}$, for the bandwidth $\gamma$ and system-environment coupling $v$, does not depend on the system size. Although reducing the error in the bandwidth, as seen in a), requires a larger number of samples, the errors do not scale with the system size, consistent with Proposition \ref{prop:non_markovian_main}.} 
    \label{fig:recovered_coefs_nonmarkovian}
\end{figure}

\subsection{\label{sec:numerics_ensemble} Transverse field Ising model}
We consider a system $S$ of $N$ spin-$\frac{1}{2}$ particles in 1D evolving under the transverse-field Ising Hamiltonian:
\begin{align}
    H_S&=-J\sum_{i=1}^{N-1}\sigma_i^x\sigma_{i+1}^x-h\sum_{i=1}^N\sigma_i^z,
\end{align}
where $h,J$ are given by jointly Gaussian random variables:
\begin{align}
    &(h,J)^T = B\cdot (Z_1,Z_2)^T+\mu^T,
\end{align}
where $Z_i\sim \mathcal{N}(0,1)$, the covariance matrix is $\Sigma = BB^T$ and the means are $\mu = (\mathbb{E}(h), ...,\mathbb{E}(J))$.

Following the measurement strategy outlined in Section \ref{sec:numerics_nonmarkovian} for the constant term of the kernel, we obtain estimates for the covariance matrix. We evolve each observable for $10$ evenly-spaced timesteps between $t_{\min}=10^{-3}$ and $t_{\max}=10^{-1}$, each time drawing a new set of random parameters for the transverse field Ising model and obtaining a projective measurement outcome. We repeat this for $S$ shots per timestep and fit a polynomial of degree $2$ to obtain the required derivatives. The absolute error in estimating the system parameters and covariances decreases with $1/\sqrt{S}$ in Figure \ref{fig:ensemble_shots} b). In Figure \ref{fig:ensemble_shots} c) we see that this scaling does not increase with $N$, so the absolute error does not increase with system size. Recall that Proposition \ref{prop:hamiltonian_ensemble} is only valid for time $t=O((\log(N)\log(1/\epsilon))^{-1/2})$. In practice this is sufficient for system sizes accessible to current devices: using the same $t_{\max}=10^{-1}$ for all $N$ we do not see an increase in the shot noise scaling coefficient. The scaling we obtain predicts comparable accuracies in $h,J$ is similar to that obtained in the dissipative case. The covariance matrix is obtained from the second derivative of the observable, so we obtain similar precision to the constant term in the kernel for the dissipative case, as expected. 

\begin{figure*}
    \centering
    \includegraphics[width=\linewidth]{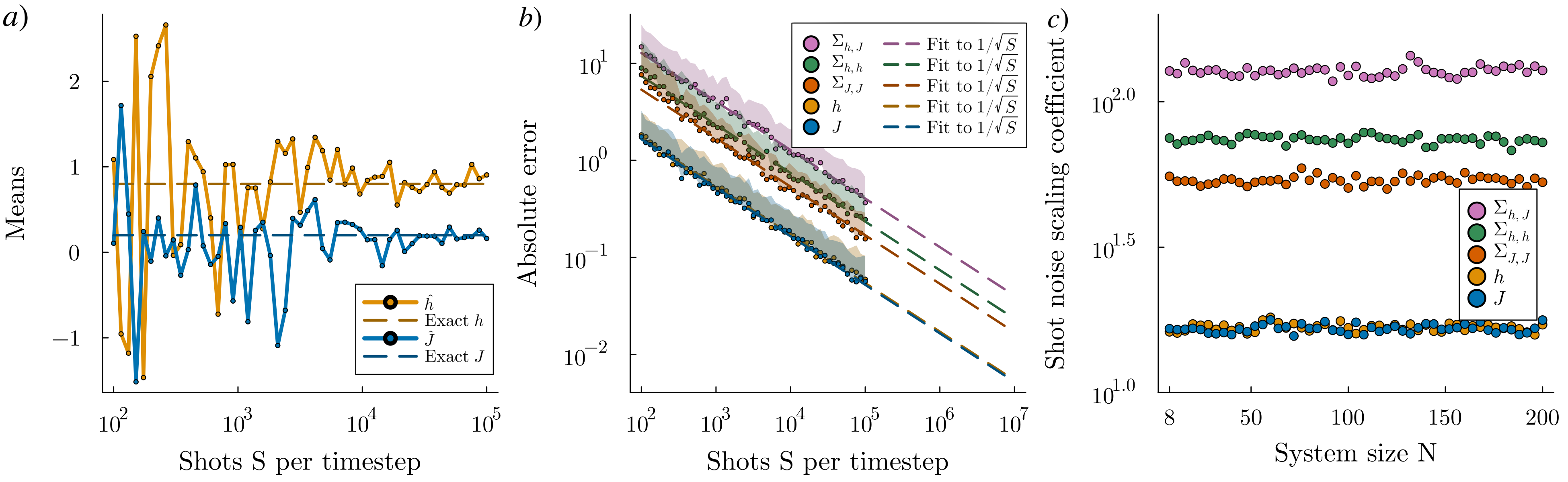}
    \caption{Recovery of the mean and covariance of the Hamiltonian coefficients for ensemble transverse field Ising Hamiltonians with dense $\Sigma$. We fix $J=0.2, h=0.8, \Sigma_{hh}=0.6, \Sigma_{hJ} = 0.3, \Sigma_{JJ}=0.7$. a) Convergence of $h,J$ towards the exact values in terms of the number of shots per timestep $S$, with fixed $N=200$. For each observable we measure $10$ timesteps between $t_{\min}=10^{-3}$ and $t_{\max}=10^{-1}$ to obtain the polynomial fit. 
    b) Absolute error in $h,J$ and the entries of the covariance matrix in terms of the number of shots per timestep $S$, with fixed $N=200$. The dashed lines are fits to $1/\sqrt{S}$, and we see the expected scaling from Hoeffding's inequality. c) More precisely, from plot $b)$ we expect that the absolute error will scale like $\sigma/\sqrt{S}$, where the standard deviation $\sigma$ is the shot noise scaling coefficient. We plot it for $h,J$ and their covariances and see that it does not depend on the system size. Note that, since we assume translational invariance, we only need to estimate one instance of each desired parameter. In this case the covariance matrix is dense (with many parameters being equal, by translational invariance), so in Proposition \ref{prop:hamiltonian_ensemble} we see a scaling with $N \log(N)$ because we want to estimate all $O(N^2)$ parameters correctly. }
    \label{fig:ensemble_shots}
\end{figure*}

\section{\label{sec:polynomial_fit}Sample complexity analysis}
From the measurement procedure we obtain samples for the time traces $B_{W,(O,I)}(t)=\frac{1}{2^N}  \tr(P_O \mathcal{E}_{W}(t)(P_I))$ for $t$ in an interval $[t_0,t_{\max}]$, which we determine below. We will use the robust interpolation method from \cite{stilck2024efficient} to obtain a polynomial fit to each time trace. They use a Lieb-Robinson bound for the Heisenberg evolution of local observables. In Appendix \ref{app:lieb_robinson_nm} we explain how it can be adapted to the non-Markovian case to obtain the following result.

\begin{lemma}[Adapted Theorem D.2, \cite{stilck2024efficient}] \label{lem:franca_liebrobinson}
    Let $\mathcal{H}_\Gamma(t)$ be a geometrically local Hamiltonian acting on a system with a $D$-dimensional regular lattice $\Gamma$ with constant $g$ and a free bosonic environment with finite memory. Moreover, let $t_{\max},\epsilon >0$ be given and $O_Y$ an observable such that $||O_Y||\leq 1$ and $O_Y$ is supported on a constant number of qubits on the system. Then there is a polynomial $p$ of degree
    \begin{align}
        &d=\tilde{O}\bigg(\log^D\bigg(\frac{\exp(vt_{\max})-1}{\epsilon}\bigg)t_{\max}\log\bigg(\frac{1}{\epsilon}\bigg)\bigg),\label{eq:bound_d}
    \end{align}
    such that for all $0\leq t\leq t_{\max}$:
    \begin{align}
        &|\tr(U^\dagger(t)O_YU(t)\rho_S\otimes\gamma_E)-p(t)|\leq \epsilon,
    \end{align} 
    and
    \begin{align}
        & p^{(m)}(0)=\tr(U^\dagger(t)O_YU(t)\rho_S\otimes\gamma_E)^{(m)}|_{t=0}, m=1,...,d,
    \end{align}
    where $U(t) = \mathcal{T}\exp(-it\int_0^tH(s)ds)$.
\end{lemma}

This shows that time traces of local observables can be approximated by a polynomial of controlled degree for constant time. For $W\neq \Id^{\otimes N}$, since $W(\cdot )W^\dagger$ is a tensor product of single-qubit unitary changes of basis, it does not change the support of local observables. Therefore, $B_{W,(O,I)}(t)$, can be approximated by a polynomial of controlled degree, $p_{W,(O,I)}(t)$, for constant time $t=O(1)$. Moreover, the derivatives of the truncation $p_{W,(O,I)}(t)$ at $0$ match those of the time trace $B_{W,(O,I)}(t)$. In order to estimate this polynomial from measurement samples we follow the robust interpolation method described in Appendix E of \cite{stilck2024efficient}, obtaining an estimate $\hat{p}_{W,(O,I)}(t)$ for $p_{W,(O,I)}(t)$ such that $|\hat{p}_{W,(O,I)}(t)-p_{W,(O,I)}(t)|<\epsilon$ for all $ t_0\leq t\leq t_{\max}$. Similarly, the higher derivatives of the estimate $\hat{p}_{W,(O,I)}(t)$ at $0$ are close to those of $p_{W,(O,I)}(t)$ in the following sense.

\begin{lemma}[Adapted Theorem E.1, Proposition E.1, Corollary E.1 \cite{stilck2024efficient}, Theorem 1.2 \cite{arora2024outlier}]\label{lem:samples}
    Let $p(t)$ be a polynomial of degree $d$ defined on an interval $[a,b]$ with $a= d^{-2}$ and $b=2+a$, $\epsilon_{S,M} >0$ a desired precision in the $M$-th derivative at $0$. Then for $\delta >0$, sampling 
    \begin{align}
        &O\bigg(d\log\bigg(\frac{d}{\delta}\bigg)\bigg),
    \end{align}
    i.i.d. points $(x_i,y_i)$ from the Chebyshev measure on $[a,b]$ satisfying
    \begin{align}
        &p(x_i)=y_i+w_i, |w_i|\leq \epsilon_{S}=  \epsilon_{S,M}\frac{(2M-1)!!}{3d^{2M}},
    \end{align}
    for at least a fraction $\alpha > \frac{1}{2}$ of the points is sufficient to obtain a polynomial $\hat{p}$ satisfying 
    \begin{align}
        &|p^{(m)}(0)-\hat{p}^{(m)}(0)|=\epsilon_{S,m} \leq \notag\\
        &\hspace{20pt}\epsilon_{S,M}\frac{(2M-1)!!}{d^{2M}}\frac{d^{2m}}{(2m-1)!!}, m=0,1,...,M.
    \end{align}
    
\end{lemma}
Note that, since $m\leq M\leq d$, all derivative estimates are at least as precise as the highest derivative: $\epsilon_{S,0}\leq \epsilon_{S,1}\leq ...\leq \epsilon_{S,M}$. Sampling from the uniform measure $t\in [\frac{1}{d^2},2+\frac{1}{d^2}]$ instead requires $O(d^2\log(\frac{1}{\delta}))$ samples \cite{stilck2024efficient}, where at least half the samples are estimated to precision $\epsilon_S$. 

We show in Appendix \ref{app:error_analysis} that, after taking into account the required precision to estimate the offsets $f_{W,(O,I)}^{(m)}$ accurately, the precision in learning the kernel derivatives $K_{a,b}^{(m)}(0), m=0,...,M$ to precision at least $\epsilon_{K,M}$ is bounded by the precision in the sampling and polynomial fit as follows:
\begin{align}
    &\epsilon_{K,M}\leq e^{O(M^2\log M)}\epsilon_{S,M+2}.
\end{align}
While this is independent of the system size, it is prohibitive for large $M$. The recursive structure of the linear systems of equations required to estimate the kernel derivatives up to order $M$, see Figure \ref{fig:inversion_schematic}, leads to a large overhead due to errors compounding from lower-order estimates.\\

In Appendix \ref{app:parallelization} we estimate the number of samples required to estimate all the time traces to precision $\epsilon_S$ using a parallelized measurement scheme. Plugging in the bounds above means that the total number of samples is upper bounded by:
\begin{align}
    &O\Big(\frac{((\mathfrak{d}+2)(s+\mathfrak{d})\mathfrak{d}^2+1)e^{ O(M^2\log M)}}{\epsilon_{K,M}^2}\times \notag \\
    &\hspace{80pt}\log\bigg(\frac{sN\cdot 3^{2k_{SE}+\mathfrak{d} }}{\delta}\bigg)\bigg).
\end{align}

\section{\label{sec:discussion} Conclusion and outlook}
We formulate the problem of learning open system models beyond the Born-Markov approximation. We work within the assumption of a Gaussian environment, where the memory effects can be captured by the environment memory kernels, i.e.~the two-point correlators of the environment operators coupling to the system. The Gaussian-environment assumption is physically reasonable for both free-space and cavity-based quantum simulators with atomic or molecular systems \cite{gardiner1985input,madsen2011observation,svendsen2023signatures,krinner2018spontaneous}, as well as for describing non-Markovian noise in superconducting devices \cite{caldeira1983path,malekakhlagh2016non,nakamura2024gate}. Under this assumption, we show that the derivatives of these memory kernels at $t = 0$ can be systematically learnt from the time evolution of local observables. Our proposed protocol has a favorable sample-complexity scaling with system size, although learning higher-order derivatives of the memory kernel still requires a substantially large amount of samples. While this makes it difficult to learn certain non-Markovian models -- either because the memory kernel contains many unknown parameters or because it exhibits sharply different short- and long-time behavior, both of which require very high-order derivatives to capture accurately -- in most experimentally relevant settings, a Gaussian non-Markovian environment can be well described by only a few lossy bosonic modes~\cite{tamascelli2018nonperturbative}. In such systems, we expect our protocol to perform comparably to established Hamiltonian or Lindbladian learning methods.  
% This makes it hard to learn non-Markovian models where either the memory kernel has a large number of unknown parameters, or if the memory kernel has drastically different short-time and long-time behavior, thus requiring very high-order derivatives for accurately capturing it. Nevertheless, in most experimentally relevant settings the Gaussian non-Markovian environment can be described by a small number of lossy bosonic modes \cite{tamascelli2018nonperturbative} and, in such systems, we expect our protocol to perform comparably to protocols for Hamiltonian or Lindbladian learning.

Furthermore, we also consider an ensemble Hamiltonian model where the noisy simulator is described by a geometrically local Hamiltonian with the coefficients of the local Hamiltonian terms being Gaussian random variables. Here, we allowed for possible all-to-all correlations between the coefficients of the different Hamiltonian terms. This model would capture noise in global controls applied in a quantum simulator which would introduce all-to-all correlated errors in the Hamiltonian. For instance, in Rydberg-atom arrays \cite{shaw2024benchmarking} and trapped-ion simulators \cite{guo2025hamiltonian}, both laser and microwave fields drive entire qubit arrays and lead to correlated frequency and intensity fluctuations that affect all qubits simultaneously \cite{jiang2023sensitivity}. We establish that, despite this all-to-all correlation, the mean and covariance matrix of the Hamiltonian coefficients can also be efficiently learnt by measuring the time-evolution of local observables in the model. 

Our work opens up several interesting questions --- the first and most immediate open question is to develop learning protocols that would be efficient for more complex memory kernels i.e., memory kernels whose parameters cannot be deduced from just a few derivatives. Furthermore, while experimental setups might be aware of their predominant source of noise and thus have prior information of the Paulis appearing in noise operators, it could also be useful to relax this assumption. For example, assuming that Pauli strings of weight larger than $k$ are $0$, in the Lindbladian setting one can learn all terms with Pauli strings less than $k$ using classical shadows \cite{stilck2024efficient}. Furthermore, for Hamiltonian learning, recently developed protocols require no prior assumption on the structure of the Hamiltonian \cite{bakshi2024structure}. Developing similar results for the non-Markovian regime would increase their experimental applicability and minimize modeling errors. Furthermore, while an emphasis of our work has been on developing rigorous protocols with provable sample complexity bounds, it would be important to develop more heuristic but practically efficient methods that can operate with a relatively small sample-budget. Finally, as quantum simulators scale up to $\gtrsim 500-1000$ qubits, even employing protocols with linear sample-complexity to learn a detailed noise-model for the simulator would also become inefficient. In this regime, it is likely that we would need protocols that simply estimate an upper-bound on the noise rate with a much smaller number of samples and bypass the task of learning a full-blown noise model.
\\

\section{\label{sec:acknowledgements}Acknowledgements}
We acknowledge useful discussions with Jeongwan Haah, Hsin-Yuan (Robert) Huang and Daniel Stilck França. R.T. acknowledges support from Center for Integration of Modern Optoelectronic Materials on Demand (IMOD) seed grant (DMR-2019444), from QuPIDC, an Energy Frontier Research Center, funded by the US Department of Energy (DOE), Office of Science, Basic Energy Sciences (BES), under the award number DE-SC0025620 and from the European Union’s Horizon Europe research and innovation program under grant agreement number 101221560 (ToNQS). J.A.M.L and R.T. acknowledge funding from the Munich Center for Quantum Science and Technology (MCQST), funded by the Deutsche Forschungsgemeinschaft (DFG) under Germany’s Excellence Strategy (EXC2111-390814868). The research is part of the Munich Quantum Valley, which is supported by the Bavarian State Government with funds from the High tech Agenda Bayern Plus. J.C. acknowledges support from AFOSR (YIP No.  FA9550-25-1-0147) and the Terman Faculty Fellowship at Stanford University.

\nocite{*}

\bibliography{references}% Produces the bibliography via BibTeX.

\clearpage
\onecolumngrid
\appendix
\pagenumbering{alph}

\section{\label{app:nonmarkovian_learning} Non-Markovian learning}

\begin{table}
    \centering
    \renewcommand{\arraystretch}{1.2} % more vertical space between rows
    \begin{tabular}{l p{2.5cm} p{12.2cm}}
    \toprule
    \textbf{Symbol} & \textbf{Defined in} & \textbf{Informal description} \\
    \midrule
    \hspace{1em} $\mathcal{P}_S$ & Eq.\eqref{eq:H(t)} & Index set for the Pauli strings in the system Hamiltonian $H_S$.\\ \midrule
    \hspace{1em} $\mathcal{P}_{SE}$ & Eq.\eqref{eq:H(t)} & Index set for the Pauli strings in the coupling Hamiltonian $V_{SE}(t)$.\\ \midrule
    \hspace{1em} $P_a$ & Below Eq.\eqref{eq:H(t)} & Pauli string acting on the system.\\ \midrule
    \hspace{1em} $A_a(t)$ & Below Eq.\eqref{eq:H(t)} & Hermitian operator  linear in the bosonic creation and annihilation operators.\\ \midrule
    \hspace{1em} $K_{a,b}(t)$ & Eq.\eqref{eq:kernel_def} & Non-Markovian kernel: two-point correlation function of environment modes.\\ \midrule
    \hspace{1em} $K_{a,b}^{(m)}(0)$ & Eq.\eqref{eq:Km_def} & $m$-th derivative of $K_{a,b}(t)$ at $t=0$.\\ \midrule
    \hspace{1em} $k_S, k_{SE}$ & Above Eq.\eqref{eq:total_variation_condition} & Size of the support of Pauli strings in $\mathcal{P}_S$ and $\mathcal{P}_{SE}$, respectively. \\ \midrule

    \hspace{1em} $\mathfrak{d}$ & Above Eq.\eqref{eq:total_variation_condition}  & Counts how many terms in $H(t)$ might not commute with a given term in $H(t)$. \\ \midrule 
    \hspace{1em} $\mathcal{Y}_{\mathcal{V}}(a,b)$ &  Eq.\eqref{eq:Yab_characterize} & Set of pairs of Pauli strings whose coefficients would mix with those of $K_{a,b}(t)$ if we were to only measure on the support $\mathcal{V}$. \\ \midrule 
    \hspace{1em} $\mathcal{S}_{a,b}$ & Above Eq.\eqref{eq:enlarged_support_cond_1} & Joint support of Pauli strings $P_a,P_b$. Satisfies $|\mathcal{S}_{a,b}|\leq O(k_{SE})$ \\ \midrule 
    
    \hspace{1em} $\mathfrak{d}_0$ & Eq. \eqref{eq:d_0_def} & Upper bound on $|\mathcal{Y}_{\mathcal{S}_{a,b}}(a,b)|$ for every $a,b\in\mathcal{P}_{SE}$. Satisfies $\mathfrak{d}_0 \leq \mathfrak{d}$. \\ \midrule
    \hspace{1em} $\mathcal{Q}_{a,b}$ & Above Eq.\eqref{eq:Iab_def} & Extra sites that need to be added to $\mathcal{S}_{a,b}$ to recover the coefficients of $K_{a,b}(t)$. Satisfies $|\mathcal{Q}_{a,b}|\leq \mathfrak{d}_0$.\\ \midrule 
    \hspace{1em} $\mathcal{I}_{a,b}$ &  Eq.\eqref{eq:Iab_def} & Tomography on support $\mathcal{I}_{a,b}$ is sufficient to learn $K_{a,b}(t)$. Satisfies $|\mathcal{I}_{a,b}|\leq O(k_{SE})+\mathfrak{d}_0$. \\ \midrule 
    
    \hspace{1em} $s$ & Above Eq.\ref{eq:total_variation_condition} & Sparsity of the kernel: for each $P_a$ there are at most $s$ Pauli strings $P_b$ such that $K_{a,b}(t)\neq 0$. In the ensemble Hamiltonian case, sparsity of the covariance matrix $\Sigma$. \\
    \midrule
    \hspace{1em} $a_0$ & Appendix \ref{app:parallelization} & Diameter of Pauli strings in $\mathcal{P}_{SE}$. \\ \midrule
    \hspace{1em} $M$ & Proposition \ref{prop:non_markovian_main} & Highest derivative $\partial_t^M K_{a,b}(t)|_{t=0}$ that we estimate. \\ \midrule 

    \hspace{1em} $\mathcal{E}_W(t)$ & Eq. \ref{eq:rho_W_t} & State preparation channel with an intermediate gate $W$. \\ \midrule 

    \hspace{1em} $\mathcal{F}_W^{(m)}(t)$ & Eq. \ref{eq:UWn_def} & $m$-th term in the Dyson series of the channel $\mathcal{E}_W(t)$. \\ \midrule 
    
    \hspace{1em} $\mathcal{F}_{W,S}^{(2)}(t),\mathcal{F}_{W,SE}^{(2)}(t)$ & Eq. \ref{eq:UW2_def} & Terms in $\mathcal{F}_{W}^{(2)}(t)$ that only involve $H_S$ and $V_{SE}(t)$, respectively. \\ \midrule 

    \hspace{1em} $T_W^{(m)}$ & Eq. \ref{eq:TWm_def} & Certain entries of this linear map yield the derivatives $K_{a,b}^{(m-2)}(0)$. \\ \midrule 

    \hspace{1em} $B_{W,(O,I)}(t)$ & Eq. \ref{eq:experimental_time_traces} & Time trace for initial Pauli string $P_I$, observable $P_O$ and intermediate gate $W$. \\ \midrule 

    \hspace{1em} $f_{W,(O,I)}^{(m)}$ & Eq. \ref{eq:f_definition_main} & Offset containing terms in $H_S$ and kernel derivatives of order $\leq m-3$. \\ \midrule 

    \hspace{1em} $B_a(t)$ & Above Eq. \ref{eq:trace_B} & Variable that can be $\lambda_a,A_a(t),\lambda_a+A_a(t)$ when $a$ is in $\mathcal{P}_S,\mathcal{P}_{SE}$, or both. \\ \midrule 

    \hspace{1em} $\hat{x}$ &   & The hat indicates an estimate for a variable $x$. \\ \midrule 

    \hspace{1em} d & Lemma \ref{lem:franca_liebrobinson} & Degree of the polynomial fit. \\ \midrule

    \hspace{1em} $\Sigma$ & Section \ref{sec:ensemble_model} & Covariance matrix of the ensemble Hamiltonian model. \\

    \bottomrule
    \end{tabular}
    \caption{Table of the symbols used in the non-Markovian and ensemble Hamiltonian models.}
    \label{tab:param_table}
\end{table}

\subsection{\label{app:measurement_supports} Tomography on a support of bounded size $\mathcal{I}_{a,b}$}
We now show which measurements need to be performed in order to invert the linear map:
\begin{align}
    &K_{a,b}^{(m-2)}(0)\to \tr[P_OT_W^{(m)}(P_I)].
\end{align}
In Section \ref{sec:tomography} we proposed to perform tomography on $\mathcal{I}_{a,b}$, which is an enlargement of the joint support of $P_a,P_b$,  $\mathcal{S}_{a,b}$, by at most a constant number of qubits. Here, we formally define the tomography region $\mathcal{I}_{a, b}$, provide bounds on its size and prove that it is sufficient to perform tomography on it to learn $K_{a, b}^{(m - 2)}(0)$. Recall the decomposition of the linear map $T_W^{(m)}$:
\begin{align}
    &T_W^{(m)}(X)=\sum_{P,Q}\xi^{W,{(m)}}_{P,Q}P XQ.
\end{align}
We recall from Eqs.\eqref{eq:enlarged_support_cond_1} and \eqref{eq:enlarged_support_cond_2} that for given Pauli strings $P_a, P_b$ and a support $\mathcal{V}$ we defined $\mathcal{Y}_{\mathcal{V}}(a, b)$ via:
\begin{align}
    &(P_c,P_d)\in \mathcal{Y}_\mathcal{V}(a,b)\text{ if and only if } P_c|_{\mathcal{V}}=P_a|_{\mathcal{V}},P_d|_{\mathcal{V}}=P_b|_{\mathcal{V}} \text{ and } P_c|_{\mathcal{V}^c}=P_d|_{\mathcal{V}^c}\neq \Id_{\mathcal{V}^c}.\label{eq:Yab_characterize}
\end{align}
As discussed in the main text, $\mathcal{Y}_\mathcal{V}(a,b)$ can be interpreted as the set of Pauli string pairs whose contribution to $T_W^{(m)}$ cannot be distinguished from that of $P_a, P_b$ by performing tomography on only $\mathcal{V}$. Indeed, if we restrict our measurements to a support $\mathcal{V}$ we obtain a map whose coefficients are sums of various $\xi_{c,d}^{(m)}$:
\begin{align}
    &\frac{1}{2^{|\mathcal{V}^c|}}\tr_{\mathcal{V}^c}(T_W^{(m)}(X))=\sum_{P,Q\in\mathcal{P}_{|\mathcal{V}|}}\Xi^{W,\mathcal{V},(m)}_{P,Q}PXQ, \quad\Xi_{P_c,P_d}^{W,\mathcal{V},(m)}=\xi^{W,(m)}_{P_c,P_d}+\sum_{(P_e,P_f)\in\mathcal{Y}_\mathcal{V}(c,d)}\xi_{P_e,P_f}^{W,(m)}
\end{align}
Our general strategy is to find a support $\mathcal{I}_{a,b}$ such that $\mathcal{Y}_{\mathcal{I}_{a,b}}(a,b)=\emptyset$, yielding:
\begin{align}
    &\Xi_{P_c,P_d}^{W,\mathcal{I}_{a,b},(m)}=\xi_{P_c,P_d}^{W,(m)}. \label{eq:Xi_xi}
\end{align}
We can obtain the kernel derivatives by estimating $\xi_{P,Q}^{W,(m)},\xi_{Q,P}^{W,(m)}$, for $P,Q$ as described in Eqs.\eqref{eq:xi_case1},\eqref{eq:xi_case2} and \eqref{eq:xi_case3}. Since we only need to learn particular coefficients of the super-operator $T_W^{(m)}$, we don't need to perform process tomography \cite{chuang1997prescription}. In Lemma \ref{lem:Ainverse} we outline the procedure to obtain the desired coefficients, which has the cost of state tomography.

\begin{lemma}[Enlarged support]\label{lem:enlarged_support}
    Let $P_a,P_b\in \mathcal{P}_{SE}$ be Pauli strings with nonzero kernel, $K_{a,b}\neq 0$. The derivative $K_{a,b}^{(m)}(0)$ can be recovered with the following procedure:
    \begin{align}
        \text{If } a\neq b&:\begin{cases}
            \text{Tomography on }\mathcal{I}_{a,b}\text{ using }W=\Id^{\otimes N}, & m \text{ even}.\\
            \text{Tomography on }\mathcal{I}_{a,b}\text{ using }W=P_w, & m \text{ odd}.
        \end{cases}\\
        \text{If } a= b&:\begin{cases}
            \text{Tomography on }\mathcal{I}_{a,a}\text{ using }W=\Id^{\otimes N}, & m \text{ even}.\\
            \text{Tomography on }\mathcal{I}_{a,\tilde{a}}\text{ using }W=(S\cdot H)_s\otimes \Id^{\otimes N-1}, & m \text{ odd}.
        \end{cases},
    \end{align}
    where $P_w$ is a single-qubit Pauli string that anticommutes with $P_a\cdot P_b$ and $(S\cdot H)_s\otimes \Id^{\otimes N-1}$ is identity everywhere except at site $s$, which can be any site in the support of $P_a$. These supports have bounded size:
    \begin{align}
        &|\mathcal{I}_{a,b}|\leq 2k_{SE}+\mathfrak{d}_0-2, \quad|\mathcal{I}_{a,\tilde{a}}|\leq k_{SE}+\mathfrak{d}_0-1, \quad \text{ with }\mathfrak{d}_0 :=\underset{K_{a,b}\neq 0}{\max}|\mathcal{Y}_{\mathcal{S}_{a,b}}(a,b)| \leq \mathfrak{d}. \label{eq:d_0_def}
    \end{align}
\end{lemma}
\begin{proof}
    Recall that we want to find $\mathcal{I}_{a,b}$ such that $\mathcal{Y}_{\mathcal{I}_{a,b}}(a,b)=\emptyset$ to learn the derivatives of the kernel from Eq. \eqref{eq:Xi_xi}. We will construct $\mathcal{I}_{a,b}$ as an enlargement of $\mathcal{S}_{a,b}$ that sequentially resolves the conflicting pairs in $\mathcal{Y}_{\mathcal{S}_{a,b}}(a,b)$. We begin by bounding this set. Consider a pair of Pauli strings $(P_c,P_d) \in \mathcal{Y}_{\mathcal{S}_{a,b}}(a,b)$. Since $P_c|_{\mathcal{S}_{a,b}}=P_a|_{\mathcal{S}_{a,b}}$, we know that $P_c$ overlaps with $P_a$. But $P_a$ can overlap with at most $\mathfrak{d}$ terms in $\mathcal{P}_{SE}$, so there are at most $\mathfrak{d}$ options for $P_c$. Since $(P_c,P_d) \in \mathcal{Y}_{\mathcal{S}_{a,b}}(a,b)$, $P_d$ is fixed by $P_a,P_b,P_c$ per Eq.\eqref{eq:Yab_characterize}. Therefore, the number of pairs in $\mathcal{Y}_{\mathcal{S}_{a,b}}(a,b)$ is at most $\mathfrak{d}$, so $\mathfrak{d}_0\leq \mathfrak{d}$.\\

    The pairs $(P_c,P_d)\in \mathcal{Y}_{\mathcal{S}_{a,b}}(a,b)$ can also be written as $P_c=Q\cdot P_a,P_d = Q\cdot P_b$, where $Q$ is a non-identity Pauli string that does not overlap with $\mathcal{S}_{a,b}$. We label the pairs in $\mathcal{Y}_{\mathcal{S}_{a,b}}(a,b)$ from $1$ to $\mathfrak{d}_0$, and consider the corresponding $Q_1,...,Q_{\mathfrak{d}_0}$. Pick one site $s_i$ from the support of each $Q_i$ and define the enlarged support $\mathcal{I}_{a,b}$ as follows:
\begin{align}
    &\mathfrak{\hspace{0pt}}\mathcal{Q}_{a,b}  :=\{s_1,...,s_{\mathfrak{d}_0}\},\quad \mathcal{I}_{a,b} := \mathcal{S}_{a,b}\cup \mathcal{Q}_{a,b}.
\end{align}
If two $Q_i,Q_j$ overlap, we can pick $s_i=s_j$ to minimize the number of extra sites we need to do tomography on. Since $\mathcal{S}_{a,b}\subseteq \mathcal{I}_{a,b}$, we have $\mathcal{Y}_{\mathcal{I}_{a,b}}(a,b)\subseteq \mathcal{Y}_{\mathcal{S}_{a,b}}(a,b)$. But for each pair $(P_c,P_d)\in \mathcal{Y}_{\mathcal{S}_{a,b}}(a,b)$, there is a site in $s\in \mathcal{Q}_{a,b}$ where $P_c|_s=P_d|_s\neq \Id$. Therefore, $P_c|_{\mathcal{I}_{a,b}}\neq P_a|_{\mathcal{I}_{a,b}}$ so $(P_c,P_d)\not\in\mathcal{Y}_{\mathcal{I}_{a,b}}(a,b)$, yielding $\mathcal{Y}_{\mathcal{I}_{a,b}}(a,b) = \emptyset$.\\

We now bound $|\mathcal{I}_{a,b}|$. If $|\mathcal{S}_{a,b}|=2k_{SE}$ or $2k_{SE}-1$, then one of the Pauli strings in each pair $(P_c,P_d)\in \mathcal{Y}_{\mathcal{S}_a\cup\mathcal
{S}_b}(a,b)$ would need to be at least $(k_{SE}+1)$-local, since $(P_c,P_d)=(Q\cdot P_a,Q\cdot P_b)$ and $Q$ is non-identity and supported outside of $\mathcal{S}_{a,b}$. Since Pauli strings in $\mathcal{P}_{SE}$ are at most $k_{SE}$-local, if $|\mathcal{S}_{a,b}|=2k_{SE}$ or $2k_{SE}-1$ we have $\mathcal{Y}_{\mathcal{S}_{a,b}}(a,b) = \emptyset$. Thus $\mathcal{I}_{a,b}$ is a support of at most $2k_{SE}-\mathfrak{d}_0+2$ qubits: $|\mathcal{I}_{a,b}|\leq 2k_{SE}-\mathfrak{d}_0+2$. Similarly we obtain $|\mathcal{I}_{a,\tilde{a}}|\leq k_{SE}+\mathfrak{d}_0-1$.\\

If we perform process tomography on $\mathcal{I}_{a,b}$ or, as we will see in Lemma \ref{lem:Ainverse}, a tomography procedure with the sample requirement of state tomography, by Eq. \eqref{eq:Xi_xi} we can recover the coefficients $\Xi_{P,Q}^{W,\mathcal{I}_{a,b},(m+2)}=\xi_{P,Q}^{W,(m+2)}$, for $P,Q$ as described in Eqs.\eqref{eq:xi_case1},\eqref{eq:xi_case2} and \eqref{eq:xi_case3}. Note that since $\mathcal{I}_{b,a}=\mathcal{I}_{a,b}$, tomography on this region also recovers $\xi_{Q,P}^{W,(m+2)}$, which allows us to recover the real and imaginary parts of $K_{a,b}^{(m)}(0)$.
\end{proof}

\begin{lemma}\label{lem:Ainverse}
    Let $T(\cdot) = \sum_{c,d}\xi_{c,d} P_c \cdot P_d$ be a Hermiticity-preserving linear map acting on $N$-qubit states. Suppose we can estimate $O_{a,b}=\frac{1}{2^N}\tr(P_aT(P_b))$ to precision $\epsilon$, where $P_a,P_b$ are phaseless Pauli strings of length $N$, satisfying $\tr(P_aP_b)= 2^N\delta_{a,b}$. Estimating all $4^N\times 4^N$ observables $O_{a,b}$ allows us to obtain estimates for every complex $\xi_{c,d}$ to precision $\epsilon$. Collecting the coefficients in vectors and using hats for the estimates, this is stated as:
    \begin{align}
        &||\hat{\vec{O}}-\vec{O}||_\infty \leq \epsilon \Rightarrow ||\hat{\vec{\xi}}-\vec{\xi}||_{\infty} \leq \epsilon.
    \end{align}
    Moreover, if we only want to obtain estimates for one particular $\xi_{ef}$ we only need estimates for $4^N$ different $O_{a,b}$. 
\end{lemma}
\begin{proof}
    The vector $\vec{O}$ for which we have good estimates is related to $\vec{\xi}$ by a linear equation:
    \begin{align}
        &O_{a,b}=\frac{1}{2^N}\tr(P_aT(P_b))=\sum_{c,d}\frac{1}{2^N}\tr(P_aP_cP_bP_d)\xi_{c,d}\Rightarrow \vec{O} = G\cdot \vec{\xi}, \text{ where }G_{(a,b),(c,d)}=\frac{1}{2^N}\tr(P_aP_cP_bP_d)
    \end{align}
    Introducing a SWAP operator and writing the vectorized superoperator as $\tilde{T}=\sum_{c,d}P_c\otimes P_d \xi_{c,d}$ we can invert $G$:
    \begin{align}
        &G_{(a,b),(c,d)}=\frac{1}{2^N}\tr(P_aP_cP_bP_d)= \frac{1}{2^N}\tr((P_a\otimes P_b)(P_c\otimes P_d)\text{SWAP})\Rightarrow\\
        &O_{a,b}=\frac{1}{2^N}\tr(P_a\otimes P_b\cdot  \tilde{T}\cdot\text{SWAP}) \Rightarrow \tilde{T}\cdot \text{SWAP} = \frac{1}{2^N}\sum_{a,b}P_a\otimes P_b O_{a,b} \Rightarrow \tilde{T} = \frac{1}{2^N}\sum_{a,b} (P_a\otimes P_b)\text{SWAP}\cdot  O_{a,b} \Rightarrow\\
        &\xi_{c,d} = \sum_{a,b} \frac{1}{2^{3N}}\tr((P_c\otimes P_d)(P_a\otimes P_b)\text{SWAP}) O_{a,b} \Rightarrow G^{-1}_{(c,d),(a,b)} = \frac{1}{2^{3N}}\tr(P_cP_aP_dP_b).\label{eq:Ginv}
    \end{align}
    We used the fact that the Pauli strings form a basis and $\tr(P_aP_b)=2^N\delta_{a,b}$ to write $\tilde{T}\cdot \text{SWAP}$ in terms of $O_{a,b}$ in the second line. In the third line we apply $P_c\otimes P_d$ on the left and take the trace. This means that we can obtain $\vec{\xi}$ from $\vec{O}$ as:
    \begin{align}
        &\vec{\xi} = G^{-1}\cdot \vec{O}.
    \end{align}
    Notice that $\vec{O}$ is a vector of $4^N$ real entries, while $\vec{\xi}$ is a vector of $4^N$ complex entries. However, since $T$ is a Hermiticity-preserving map we have that $\xi_{c,d}^*=\xi_{c,d}$, so the number of independent real entries is $4^N$. Suppose we want to estimate a particular $\xi_{ef}$.
    The row $G^{-1}_{(e,f),(\cdot,\cdot)}$ is largely filled with $0$: for each $P_a$ there is only one phaseless Pauli string $P_b$ such that $G^{-1}_{(e,f),(a,b)}\neq 0$, the one satisfying $P_b\propto (P_eP_aP_f)^\dagger$. Therefore, let $R_{a}^{e,f}:=G^{-1}_{(e,f),(a,b)}$ be the row entry corresponding to $P_b$ being fixed to the (phaseless) Pauli string satisfying $P_b \propto (P_eP_a P_f)^\dagger$ and similarly $O^{ef}_a:= O_{a,b}$. Thus from each row of the equation above we obtain:
    \begin{align}
        &\xi_{ef} = \vec{R}^{ef}\cdot \vec{O}^{ef},
    \end{align}
    where the dot product is between two vectors of length $4^N$. Since $|R^{ef}_a|=\frac{1}{4^N}$, by the triangle inequality we obtain the bound on the precision:
    \begin{align}
        &|\hat{\xi}_{ef}-
        \xi_{ef}|=|\vec{R}^{ef}\cdot (\hat{\vec{O}}^{ef}-\vec{O}^{ef})|\leq 4^N\cdot \frac{1}{4^N}||\hat{\vec{O}}^{ef}-\vec{O}^{ef}||_\infty \leq \epsilon.
    \end{align}

    Notice that if we are only interested in estimating one particular $\xi_{ef}$, it suffices to have estimates for only $4^N$ observables, $\vec{O}^{ef}$, instead of $4^N\times 4^N$ required for process tomography. If we initialize Pauli eigenstates and measure in the Pauli basis, we can estimate $\vec{O}^{ef}$ by preparing initial states in all $3^N$ distinct bases, measuring in the corresponding basis of $P_b$ and reconstructing each observable with the correct eigenvalue signs.
\end{proof}  

This Lemma means that in order to estimate $\xi_{a,b}^{W,\mathcal{I}_{a,b}}$
we only need to prepare $3^{|\mathcal{I}_{a,b}|}$ experimental configurations and measure to the desired precision. We summarize this Section in the following Lemma.

\begin{lemma}
Given a geometrically local $V_{SE}(t)$ with each term supported on at most $k_{SE}$ sites and not commuting with at most $\mathfrak{d}$ other terms in $\mathcal{P}_{SE}$, the derivatives at $t=0$ of kernel $K_{a,b}(t)$ can be learned by doing state tomography on $\mathcal{I}_{a,b}:=\mathcal{S}_{a,b}\cup\{s_1,...,s_{\mathfrak{d}_0}\}$, a region of size at most $O(k_{SE})+O(\mathfrak{d}_0)$, where $\mathfrak{d}_0\leq \mathfrak{d}$.
\end{lemma}

In order to learn all $\mathfrak{d}N$ kernels it suffices to estimate the $3^{2k_{SE}+\mathfrak{d}_0-2}\cdot \mathfrak{d}N$ configurations required for tomography on each $\mathcal{I}_{a,b}$. In the next Appendix we show that, since many $\mathcal{I}_{a,b}$ don't overlap, we can estimate them simultaneously.

\subsection{\label{app:parallelization}Initial states and measurements for parallelization}
We now show how to minimize the number of measurement rounds by measuring a large number of regions $\mathcal{I}_{a,b},\mathcal{I}_{c,d},...$  simultaneously in each round of measurements. Since we can only measure simultaneously those regions that don't overlap, we first characterize how many regions each $\mathcal{I}_{a,b}$ overlaps with. In the main text we argued that geometric locality yields a maximum diameter $a_0 = O(k_{SE})$ for Pauli strings $P_a,P_b$, which bounds the regions that $\mathcal{I}_{a,b}$ can overlap with. Here we provide a more detailed counting in terms of $\mathfrak{d}$ directly.

\begin{lemma} \label{lem:Iab_overlap}
    Let $\{\mathcal{I}_{a,b}\}_{(a,b)}$ be a collection of supports satisfying that $P_a,P_b\in\mathcal{P}_{SE}$ wach overlap with at most $\mathfrak{d}$ other terms in $\mathcal{P}_{SE}$ and that for a fixed $P_a$ there's at most $s$ terms $P_b\in\mathcal{P}_{SE}$ such that $K_{a,b}(t)\neq 0$. Then each $\mathcal{I}_{a,b}$ overlaps with at most $(\mathfrak{d}_0+2)(s+\mathfrak{d})\mathfrak{d}^2$ other $\mathcal{I}_{c,d}$
\end{lemma}
\begin{proof}
    Recall that $\mathcal{I}_{a,b}=\mathcal{S}_{a,b}\cup \mathcal{Q}_{a,b}$, where $|\mathcal{Q}_{a,b}|\leq \mathfrak{d}_0$. We first show that each $\mathcal{I}_{a,b}$ can overlap with $(\mathfrak{d}_0+2)\mathfrak{d}$ Pauli strings $P_c\in\mathcal{P}_{SE}$: $P_a,P_b$ each overlap with at most $\mathfrak{d}$ Paulis in $\mathcal{P}_{SE}$ and since each of the $\mathfrak{d}_0$ sites in $\mathcal{Q}_{a,b}$ is in the support of a pair of Paulis in $\mathcal{P}_{SE}$, that site can overlap with at most $\mathfrak{d}$ other Paulis. We now show that each $P_c\in\mathcal{P}_{SE}$ can overlap with at most $(s+\mathfrak{d})\mathfrak{d}$ regions $\mathcal{I}_{e,f}$: suppose $P_c$ overlaps with $P_e$, then there's $s$ choices of $P_f$ such that $P_c$ overlaps with $\mathcal{S}_{e,f}$ and $K_{e,f}(t)\neq 0$. Since $P_c$ overlaps with at most $\mathfrak{d}$ Paulis in $\mathcal{P}_{SE}$, we get that it overlaps with at most $s\mathfrak{d}$ supports $\mathcal{S}_{e,f}$ corresponding to nonzero kernels, each yielding one region $\mathcal{I}_{e,f}$. Now suppose that $P_c$ overlaps with $\mathcal{Q}_{e,f}$ but not with $P_e,P_f$, i.e. $P_c$ overlaps with some pair $(P_{e_1},P_{f_1})\in\mathcal{Y}_{\mathcal{S}_{e,f}}(e,f)$. For a fixed $(P_{e_1},P_{f_1})$, there's at most $\frac{1}{2}\mathfrak{d}$ regions $\mathcal{Q}_{g,h}$ that $(P_{e_1},P_{f_1})$ could be a part of: $P_{e_1},P_{f_1}$ overlap with $P_e,P_f$ and every other $(P_{e_i},P_{f_i})\in \mathcal{Y}_{\mathcal{S}_{e,f}}(e,f)$, respectively, and since $P_{e_j}\neq P_{e_i}\neq P_{f_i}$ having $\frac{1}{2}\mathfrak{d}$ pairs means that $P_{e_1}$ would overlap with  $\mathfrak{d}$ distinct Pauli strings. Therefore, $P_c$ can overlap with at most $\mathfrak{d}^2$ regions $\mathcal{Q}_{e,f}$. We thus see that $P_c$ can overlap with at most $(s+\mathfrak{d})\mathfrak{d}$ regions $\mathcal{I}_{e,f}$ and thus $\mathcal{I}_{a,b}$ can overlap with at most $(\mathfrak{d}_0+2)(s+\mathfrak{d})\mathfrak{d}^2$.
\end{proof}

In order to minimize the required number of rounds of measurements we need to maximize the number of supports $\mathcal{I}_{a,b}$ that we measure simultaneously. In the Hamiltonian learning setting this has been solved using a graph coloring algorithm \cite{haah2024learning}, which we adapt here. Let $\mathfrak{G}$ be a graph with a vertex $v_{a,b}=(a,b)$ for every nonzero kernel $K_{a,b}(t)$ and an edge between $v_{a,b},v_{c,d}$ if the regions $\mathcal{I}_{a,b},\mathcal{I}_{c,d}$ overlap. The degree $d$ of the graph $\mathfrak{G}$ is the number of neighbors that each vertex has, i.e. the number of edges connecting this vertex to others. By Lemma \ref{lem:Iab_overlap}, each vertex has edges with at most $(\mathfrak{d}_0+2)(s+\mathfrak{d})\mathfrak{d}^2$ of their neighbors. Therefore, the degree of graph $\mathfrak{G}$ satisfies $d\leq (\mathfrak{d}_0+2)(s+\mathfrak{d})\mathfrak{d}^2$. A graph of degree $d$ can be colored by a standard greedy algorithm using $d+1$ colors, $\mathcal{C}_1,...,\mathcal{C}_{d+1}$. Those vertices assigned to the same color $\mathcal{C}_i$ will not have any edges between them, i.e. they are not neighbors, so the corresponding regions don't overlap. For example, in Figure \ref{fig:parallelization} we see that for color $\mathcal{C}_1=\{v_{a,b},v_{c,d},v_{e,f}\}$, the regions $\mathcal{I}_{a,b},\mathcal{I}_{c,d},\mathcal{I}_{e,f}$ don't overlap, so they can be measured simultaneously in the same round. Since each round of measurements corresponds to a color, there are $(\mathfrak{d}_0+2)(s+\mathfrak{d})\mathfrak{d}^2+1$ different rounds.\\

Thus, for the round of measurement corresponding to color $\mathcal{C}_i$, we will prepare an initial state of the form:
\begin{align}
    &\rho_S = \bigotimes_{(a,b)\in \mathcal{C}_i}\rho_{S,a,b} \bigotimes_{s\in \mathcal{C}_i^c}\frac{\Id}{2},
\end{align}
where $\rho_{S,a,b}$ is a state supported only on $\mathcal{I}_{a,b}$ and $\mathcal{C}_i^c$ is the complement of the supports of all regions of the same color. Similarly, we will measure in each region in $\mathcal{C}_i$ and trace our the rest. In the main text we have reduced our problem to estimating expectation values of the form $B_{W,(O,I)}(t)=\frac{1}{2^N}\tr(P_O \mathcal{E}_W(t,P_I))$, for several $t$, with $P_O,P_I$ being Pauli strings that implement tomography on $\mathcal{I}_{a,b}$ following Lemma \ref{lem:Ainverse}. We can estimate many of these simultaneously by preparing $\rho_{S,a,b}$ according to Lemma \ref{lem:samples_singlePQ} and measuring on that support, for all the supports in $\rho_S$. This Lemma also tells us that, we can estimate one $B_{W,(O,I)}(t)$ to precision $\epsilon_{S}$ with success probability at least $1-\delta$ using a number of samples $2/\epsilon_{S}^2\log(2/\delta)$. However, we want to estimate all expectation values correctly with probability at least $1-\delta'$. Recall that, since $|\mathcal{P}_{SE}|=O(N)$, there are at most $s\cdot O(N)$ kernels. For each kernel we need to prepare at most $3^{2k_{SE}+\mathfrak{d}_0 -2}$ different basis pairs $(\rho,O)$ (Lemma \ref{lem:Ainverse}) and $2$ choices of $W$, as well as some $O(1)$ number of timesteps $t$, so we obtain $\delta'$ by dividing $\delta$ by all these cases. Since there are $(\mathfrak{d}_0+2)(s+\mathfrak{d})\mathfrak{d}^2+1$ colors, in order to estimate all necessary $B_{W,(O,I)}(t)$ to precision $\epsilon_S$, we use a number of samples:
\begin{align}
    &O\bigg(\frac{2((\mathfrak{d}_0+2)(s+\mathfrak{d})\mathfrak{d}^2+1)}{\epsilon_{S}^2}\cdot \log\bigg(\frac{s N\cdot 2^2\cdot 3^{2k_{SE}+\mathfrak{d}_0 -2}}{\delta}\bigg)\bigg). \label{eq:samples_epsilon_B}
\end{align}

\begin{lemma}\label{lem:samples_singlePQ}
    Let $P,Q$ be Pauli strings on $N$ qubits and $\Phi$ a quantum channel on $N$-qubit states. Suppose we prepare an intial state $\rho$ drawn uniformly from the eigenstates of $P$ and perform a projective measurement in the basis of $Q$, obtaining a measurement sample. In order to estimate $\frac{1}{2^N}\tr(Q\Phi(P))$ to additive precision $\epsilon$ with probability at least $1-\delta$ it suffices to use
    \begin{align}
        &\frac{2}{\epsilon^2}\log\bigg(\frac{2}{\delta}\bigg)
    \end{align}
    samples.
\end{lemma}
\begin{proof}
    Consider a Pauli string $P = \bigotimes_{i=1}^N \sigma^{\alpha_i}$ and let $|r_i,\alpha_i\rangle$ be the $r_i$-eigenstate of $\sigma^{\alpha_i}$, with $r_i \in \{\pm1\}, \alpha_i \in \{x,y,z\}$. Then if we consider initial states of the form $|r,\alpha\rangle:= \bigotimes_{i=1}^N |r_i,\alpha_i\rangle$ we can write the expectation value as:
    \begin{align}
        &\frac{1}{2^N}\tr(Q\Phi(P)) = \frac{1}{2^N}\sum_{r\in\{\pm 1\}^{ N}}\bigg(\prod_{i=1}^N r_i\bigg) \tr(P_O\Phi(|r,\alpha\rangle\langle r,\alpha|)).
    \end{align}
    We now produce i.i.d. samples $Z_j \in \{\pm1\}$, for $j=1,...,m$, by:
    \begin{enumerate}
        \itemsep0em 
        \item Drawing $r_j$ uniformly at random from $\{\pm 1\}^N$.
        \item Preparing the initial state $\rho_j = |r_j,\alpha\rangle\langle r_j,\alpha|$.
        \item Applying the quantum channel $\Phi$.
        \item Performing a projective measurement in the basis of $Q$, obtaining $X_j \in \{\pm 1\}$, which follows Born's rule: 
        \begin{align}
            &\mathbb{E}(X_j|r_j) = \tr(P_O \Phi(\rho_j)).
        \end{align}
        \item Letting $Z_j = X_j\cdot \prod_{i=1}^N r_i$.
    \end{enumerate}
    These are random variables with the following mean and variance:
    \begin{align}
        &\mathbb{E}(Z_j) = \frac{1}{2^N}\sum_{r\in\{\pm 1\}^{ N}}\bigg(\prod_{i=1}^N r_i\bigg) \mathbb{E}(X_j|r) = \frac{1}{2^N}\tr(Q\Phi(P)), \\
        &\text{Var}(Z_j) = 1 - \mathbb{E}(Z_j)^2 \leq 1,
     \end{align}
    where in the last line we used that $|\frac{1}{2^N}\tr(Q\Phi(P))|\leq 1$. Therefore, the sample mean $\hat{\mu}=\frac{1}{m}\sum_{j=1}^m Z_j$ satisfies: 
    \begin{align}
        &\mathbb{E}(\hat{\mu})=\frac{1}{2^N}\tr(Q\Phi(P)) ,\quad \text{Var}(\hat{\mu}) = \frac{\text{Var}(Z_j)}{m}\leq \frac{1}{m}.
    \end{align}
    Use Hoeffding's inequality to get sufficient number of samples to estimate $\mathbb{E}(\hat{\mu})$ to precision $\epsilon$ with success probability at least $1-\delta$. Recall that $Z_j \in [a,b]=[-1,1]$:
    \begin{align}
        & \text{Pr}(|\hat{\mu}-\mathbb{E}(\hat{\mu})|\geq \epsilon) \leq 2 e^{-\frac{2 m \epsilon ^2}{(b-a)^2}} \leq \delta \Rightarrow m \geq \frac{2}{\epsilon^2}\log\bigg( \frac{2}{\delta}\bigg).
    \end{align}
\end{proof}

\subsection{\label{app:compute_f}Computing $f_{W,(O,I)}^{(m)}$}
We first show that the offset function only involves at most $m$ orders in the Dyson series:

\begin{align}
    &f_{W,(O,I)}^{(m)}=\frac{1}{2^N}\tr\bigg(P_O\bigg(\partial_t^m \mathcal{F}_{W,S}^{(2)}(t)(P_I)|_{t=0} + \sum_{\substack{n=0\\ n\neq 2}}^\infty \partial_t^m\mathcal{F}_{W}^{(n)}(t)(P_I)|_{t=0}\bigg)\bigg)\label{eq:f_definition},
\end{align}
where the $n$-th term in the Dyson series is:
\begin{align}
    &\mathcal{F}_{W}^{(n)}(t)(P_I)=(-i)^n\sum_{l=0}^n\int_t^{2t}dt_1\int_t^{t_1}dt_2...\int_t^{t_{l-1}}dt_l\int_0^tdt_{l+1}...\int_{0}^{t_{n-1}}dt_n\cdot\notag\\
    &\hspace{150pt}\tr_E([H(t_1),\dots[H(t_{l}),W[H(t_{l+1}),\dots[H(t_n),P_I\otimes \gamma_E]]W^\dagger]]). \label{eq:UWn_def}
\end{align}
The integrand will be a sum of terms, each of which is a product of kernels and system variables. We assume that each kernel $K_{a,b}(t)$ is analytic at $t=0$. In particular there is no Markovian component given by a Lindbladian, whose kernel is $\delta(t)$. Then, for small enough $t$ we have a Taylor series for each kernel, so we can write each product of kernels as a Taylor series. Consider the coefficients of order $0$, i.e. the constants. After performing the integrals above they will have a time term $t^{n}$, so if $n<m$ this term is killed by the derivatives, while if $n>m$ the term $t^{n-m}|_{t=0}=0$ makes it vanish. Thus, only the coefficients with time power $n=m$ after integration remain in $f_{W,(O,I)}^{(m)}$. In particular:
\begin{align}
    &\partial_t^m\mathcal{F}_{W}^{(n)}(t)(P_I)|_{t=0}=0 \quad \text{ if }\quad n>m.
\end{align}
Since $m \geq 2$, the $n=0,1$ terms vanish upon taking the derivatives and we get the simpler expressions:
\begin{align}
    &f_{W,(O,I)}^{(2)}=\frac{1}{2^N}\tr\Big(P_O\Big(\partial_t^m \mathcal{F}_{W,S}^{(2)}(t)(P_I)|_{t=0}\Big)\Big) = \frac{-4}{2^N}\tr(P_O[H_S,[H_S,P_I]]), \label{eq:f2}\\
    &f_{W,(O,I)}^{(m)}=\frac{1}{2^N}\tr\bigg(P_O  \sum_{n=3}^m \partial_t^m\mathcal{F}_{W}^{(n)}(t)(P_I)|_{t=0}\bigg), \hspace{10pt} \text{for }m\geq 3.
\end{align}

The $m=2$ case only involves system Hamiltonian terms, so the integral can be evaluate easily. We now show how to compute the $m\geq 3$ case. We will compute it for the case $W=\Id^{\otimes N}$ for notational clarity: 
\begin{align}
    &\mathcal{F}_{\Id^{\otimes N}}^{(n)}(t)(P_I)=(-i)^n\int_0^{2t}dt_1...\int_0^{t_{n-1}}dt_n \tr_E([H(t_1),...[H(t_n),P_I\otimes \gamma_E]]).
\end{align}
Note that for general $W$ we only need to carry out the same computation for each term $l=0,...,n$ in the sum above, each with different Pauli strings being conjugated by $W$. The expectation value with $P_O$ then yields:
\begin{align}
    &\frac{1}{2^N}\tr_E(P_O\partial_t^m\mathcal{F}_{\Id^{\otimes N}}^{(n)}(t)(P_I)|_{t=0})=\frac{(-i)^n}{2^N}\partial_t^m\int_0^{2t}dt_1...\int_0^{t_{n-1}}dt_n \tr([[P_O,H(t_1)]...H(t_n)]P_I\otimes \gamma_E)\big|_{t=0},
\end{align}
where we used $\tr(A[B,C])=\tr([A,B]C)$ repeatedly. We rewrite the trace over nested commutators in the cluster expansion formalism to obtain a better bound on the number of nonvanishing terms: while a naive bound is $O(\mathfrak{d}^{\frac{n(n+1)}{2}})$, using the cluster expansion in Lemma \ref{lem:cluster_bound} we obtain the estimate $O((\mathfrak{d}_0+2)(\mathfrak{d}+1)(e\mathfrak{d})^n)\cdot e^{O(n\log n)}$. Define a cluster $\mathbf{V}$ as a set of tuples $\{(a,\mu(a))|a\in\mathcal{P}_S\cup \mathcal{P}_{SE}\}$, where $\mu(a)$ is the multiplicity of $a$ in cluster $\mathbf{V}$. The weight of the cluster is $|\mathbf{V}|:=\sum_a \mu(a)$ and its support is $\text{Supp}\mathbf{V} = \{a\in \mathcal{P}_S\cup\mathcal{P}_{SE}:\mu(a)\geq 1 \}$. We write $a\in\mathbf{V}$ if $\mu(a)\neq 0$ and let $\mathbf{V}!:=\prod_a \mu(a)!$. While $\mathbf{V}$ is a multiset, we can also fix an ordering and label its elements $\{\mathbf{V}_1,\mathbf{V}_2,...,\mathbf{V}_{|\mathbf{V}|}\}$. Using $B_a(t)$ for a variable that can be either $\lambda_a$, $A_a(t)$ or $\lambda_a+B_a(t)$, depending on whether $a$ is in $\mathcal{P}_S,\mathcal{P}_{SE}$ or both, Lemma \ref{lem:cluster_bound} tells us we can rewrite the trace as:
\begin{align}
     &\tr([[P_O,H(t_1)],...H(t_n)]P_I\otimes \gamma_E)=\sum_{\mathbf{V}\in \mathcal{G}_n^O}\frac{1}{\mathbf{V}!}\sum_{\sigma \in S_n}\tr(P_O,[P_{\mathbf{V}_{\sigma(1)}}B_{\mathbf{V}_{\sigma(1)}}(t_{1}),...,[P_{\mathbf{V}_{\sigma(n)}}B_{\mathbf{V}_{\sigma(n)}}(t_{n}),P_I\otimes \gamma_E]]), \label{eq:trace_B}
\end{align}
where $\mathcal{G}_n^O$ is a set of at most $|\mathcal{G}_n^O|\leq O((\mathfrak{d}_0+2)(\mathfrak{d}+1)(e\mathfrak{d})^n)$ clusters. We now expand each commutator. In the resulting terms, there will be $l$ choices to the left of $P_I\otimes \gamma_E $ and $n-l$ to its right, for $l=0,...,n$. Thus let $Q(n,l)=\{(q_1,...,q_l)|q_1< q_2<...<q_l; q_j \in \{1,...,n\}\}$ be the set all of choices of $l$ indices among a set of length $n$, arranging them in increasing order. Given $(q_1,...,q_l)\in Q(n,l)$, we fix $q_{l+1},...,q_{n}$ to be the indices in $\{1,...,n\}$ not chosen in $q_1,...,q_l$, in decreasing order: they are fixed by the conditions $\cup_{i=1}^n\{q_i\} = \{1,...,n\}$ and $q_{l+1}>...>q_n$. We abuse notation slightly and write $q=(q_1,...,q_n)\in Q(n,l)$ to say that the first $l$ indices are chosen and the rest are fixed by this choice. We can now expand the commutators as follows:
\begin{align}
    &\tr(P_O,[P_{\mathbf{V}_{\sigma(1)}}B_{\mathbf{V}_{\sigma(1)}}(t_{1}),...,[P_{\mathbf{V}_{\sigma(n)}}B_{\mathbf{V}_{\sigma(n)}}(t_{n}),P_I\otimes \gamma_E]])=\sum_{l=0}^n(-1)^{n-l} \times \notag\\
    &\hspace{20pt} \sum_{q\in Q(n,l)}\tr(P_OP_{\mathbf{V}_{\sigma(q_1)}}B_{\mathbf{V}_{\sigma(q_1)}}(t_{q_1})...P_{\mathbf{V}_{\sigma(q_l)}}B_{\mathbf{V}_{\sigma(q_l)}}(t_{q_l})(P_I\otimes \gamma_E) P_{\mathbf{V}_{\sigma(q_{l+1})}}B_{\mathbf{V}_{\sigma(q_{l+1})}}(t_{q_{l+1}})...P_{\mathbf{V}_{\sigma(q_{n})}}B_{\mathbf{V}_{\sigma(q_{n})}}(t_{q_n})).\label{eq:H_LR}
\end{align}
We will now study each of these traces term by term. For ease of notation we will use the indices $c = (c_1,...,c_n)$, where $c_i:=\mathbf{V}_{\sigma (i)}$. Let $c_S=\{a\in c|B_{a}(t)=\lambda_{a},\lambda_{a}+A_{a}(t)\}$ be the set of indices that involve system parameters, $\lambda_a$, and $c_E=\{a\in c|B_{a}(t)=A_{a}(t),\lambda_{a}+A_{a}(t)\}$ the set of indices that involve bath operators, $A_a(t)$. The trace in each of the terms above can be decomposed into a trace over the system and a trace over the environment:
\begin{align}
    &\tr(P_OP_{c_{q_1}}B_{c_{q_1}}(t_{q_1})...P_{c_{q_l}}B_{c_{q_l}}(t_{q_l})(P_I\otimes \gamma_E) P_{c_{q_{l+1}}}B_{c_{q_{l+1}}}(t_{q_{l+1}})...P_{c_{q_n}}B_{c_{q_n}}(t_{q_n}))=\notag\\
    &\hspace{80pt}\underset{i|c_{q_i}\in c_S}{\prod}\lambda_{c_{q_i}}\tr(P_O\prod_{i=1}^lP_{c_{q_i}}P_I \prod_{i=l+1}^nP_{c_{q_i}})\tr(\prod_{i=1|c_{q_i}\in c_E}^lA_{c_{q_i}}(t_{q_i})\gamma_E\prod_{i=l+1|c_{q_i}\in c_E}^nA_{c_{q_i}}(t_{q_i})) 
\end{align}

Here we used the notation that the products are ordered from left to right with increasing $i$. In the trace over the environment, we use the ciclicity of the trace to move those bath operators to the right of $\gamma_E$ to its left. Note that in the general case with $W\neq \Id^{\otimes N}$ we obtain the same trace over the environment, and the trace over the system will have $W(\cdot)W^\dagger$ applied to some product of Pauli strings, i.e. a Pauli string with a phase, which is mapped to a different Pauli string with the same support under $W$. Introducing the rolled over indices $q' = (q_{l+1},q_{l+2},...,q_n,q_1,q_2,...,q_l)$ we can write the trace over the environment in terms of the kernels using Wick's theorem:
\begin{align}
    \tr(\prod_{i=1|c_{q_i}\in c_E}^lA_{c_{q_i}}(t_{q_i})\gamma_E\prod_{i=l+1|c_{q_i}\in c_E}^nA_{c_{q_i}}(t_{q_i})) &= \tr(\prod_{i=l+1|c_{q_i}\in c_E}^nA_{c_{q_i}}(t_{q_i})\prod_{i=1|c_{q_i}\in c_E}^lA_{c_{q_i}}(t_{q_i})\gamma_E), \notag\\
    &=\tr(\prod_{i=1|c_{q_i'}\in c_E}^nA_{c_{q_i'}}(t_{q_i'})\gamma_E),\\
    &=\sum_{p\in\mathcal{P}_{2u_{c}}^2}\prod_{\substack{\{i_1,i_2\}\in p\\ i_1<i_2}} K_{c_{q_{i_1}'}c_{q_{i_2}'}}(t_{q_{i_1}'}-t_{q_{i_2}'}).
\end{align}
An odd number of bath operators makes the trace vanish, so we let $2u_{c} = |c_E|$ be the number of bath operators, and thus $u_c$ is the number of kernels. We introduce $\mathcal{P}_{2u_{c}}^2$, which consists of all partitions of $\{1,...,2u_{c}\}$ into pairs $p=\{p_1,p_2\}$. For numerical computations it is better to use Gaussian integration by parts:
\begin{align}
    &\tr(A_{c_1}(t_1)...A_{c_{2u_c}}(t_{2u_c})\gamma_E)= \sum_{i=2}^{2u_c}\tr(A_{c_1}(t_1)A_{c_i}(t_i)\gamma_E)\tr(A_{c_2}(t_2)...\widehat{A_{c_i}(t_i)}...A_{c_{2u_c}}(t_{2u_c})\gamma_E),
\end{align}
where the hat indicates that the term is missing. Since we know that for fixed $P_{c_1}$ there's at most $s$ $P_{c_i}$ such that $K_{c_1,c_i}(t_{1}-t_i)\neq 0$, the sum involves only $s$ terms. Applying this recursively we see that the total number of terms is $s^{u_c} << |\mathcal{P}_{2u_c}^2|=\frac{(2u_c)!}{2^{u_c}u_c!}$. Since we take $n$ integrals, $m$ derivatives of this expression, with $3\leq n \leq m$, and then set $t=0$ we obtain an expression of the form:
\begin{align}
    &\partial_t^m\int_0^{2t}dt_1...\int_0^{t_{n-1}}dt_n \sum_{p\in\mathcal{P}_{2u_{c}}^2}\prod_{\substack{\{i_1,i_2\}\in p\\i_1<i_2}}K_{c_{q_{i_1}'}c_{q_{i_2}'}}(t_{q_{i_1}'}-t_{q_{i_2}'})\Big|_{t=0} = \sum_{p\in\mathcal{P}_{2u_{c}}^2}\sum_{z_1+...+z_{u_c}=m-n}\tilde{v}^{z,p}_{c,q}\prod_{j=1}^{u_{c}}\frac{1}{z_j!}K_{c_{q_{p_{j,1}}'}c_{q_{p_{j,2}}'}}^{(z_j)}(0).
\end{align}
We ordered the partition $p$ into ordered pairs of indices $(p_{j,1},p_{j,2})$ labeled by $1\leq j\leq u$, with $p_{j,1}<p_{j,2}$. We collect combinatorial factors and signs arising from the integrals and derivatives into $\tilde{v}^{z,p}_{c,q}$ and write $v_{c,q}^{z,p}=\tilde{v}^{z,p}_{c,q}\prod_{j=1}^{u_c}(z_j!)^{-1}$. Therefore, for $m\geq 3$ the offset can be computed to be:
\begin{align}
    f_{\Id^{\otimes N},(O,I)}^{(m)} &= \frac{1}{2^N}\sum_{n=3}^m(-i)^m \sum_{\mathbf{V}\in \mathcal{G}_n^O}\frac{1}{\mathbf{V}!}\sum_{\sigma \in S_n}\sum_{l=0}^n(-1)^{n-l}\sum_{q\in Q(n,l)}\underset{i|c_{q_i}\in c_S}{\prod}\lambda_{c_{q_i}}\tr(P_O\prod_{i=1}^lP_{c_{q_i}}P_I \prod_{i=l+1}^nP_{c_{q_i}})\times \notag\\
    &\hspace{20pt}\sum_{p\in\mathcal{P}_{2u_{c}}^2}\sum_{z_1+...+z_{u_c}=m-n}v^{z,p}_{c,q}\prod_{j=1}^{u_{c}}K_{c_{q_{p_{j,1}}'}c_{q_{p_{j,2}}'}}^{(z_j)}(0),\label{eq:f_expansion}
\end{align}
where recall that $c_i := \mathbf{V}_{\sigma(i)}$. Since $n\geq 3$ we see that for the $m$-th derivative, the offset only contains kernel coefficients up to order $m-3$. 

\begin{lemma}\label{lem:cluster_bound}
    Let $H(t) = \sum_{a \in \mathcal{P}_S\cup\mathcal{P}_{SE}} B_a(t)$, where $B_a(t)$ can be $\lambda_a,A_a(t)$, or their sum, be a Hamiltonian such that each $P_aB_a(t)$ might not commute with at most $\mathfrak{d}$ other terms $P_bB_b(s)$ with $b\in \mathcal{P}_S\cup\mathcal{P}_{SE}$. Then if the support of $P_O$ is contained in the support of $ \mathcal{I}_{a,b}$:
    \begin{align}
    &\frac{1}{2^N}\tr([[P_O,H(t_1)],...H(t_n)],P_I\otimes \gamma_E) \notag\\
    &\hspace{30pt}=\frac{1}{2^N}\sum_{\mathbf{V}\in \mathcal{G}^O_n}\frac{1}{\mathbf{V}!}\sum_{\sigma \in S_n}\tr(P_O,[P_{\mathbf{V}_{\sigma(1)}}B_{\mathbf{V}_{\sigma(1)}}(t_{1}),...,[P_{\mathbf{V}_{\sigma(n)}}B_{\mathbf{V}_{\sigma(n)}}(t_{n}),P_I\otimes \gamma_E]]), 
\end{align}
    where $\mathcal{G}^O_n$ is a set of at most $|\mathcal{G}^O_n|\leq (\mathfrak{d}_0+2)\mathfrak{d}(e\mathfrak{d})^{n}$ clusters of weight $n$.
\end{lemma}
\begin{proof}
    We will follow the cluster expansion method outlined in \cite{haah2024learning,wild2023classical} to bound the number of terms in the sum. Let $\mathcal{C}_n$ be the set of all clusters of size $n$. Then, since we can view $\mathcal{C}_n$ as a collection of unordered indices in $(\mathcal{P}_n\cup \mathcal{P}_{SE})^n$ and each ordered tuple on the left hand side will appear $\mathbf{V}!$ times in the sum over $S_n$, we obtain:
    \begin{align}
    &\frac{1}{2^N}\tr([[P_O,H(t_1)],...,H(t_n)]P_I\otimes \gamma_E)=\frac{1}{2^N}\sum_{\mathbf{V}\in \mathcal{C}_n}\frac{1}{\mathbf{V}!}\sum_{\sigma \in S_n}\tr([[P_O,P_{\mathbf{V}_{\sigma(1)}}B_{\mathbf{V}_{\sigma(1)}}(t_{1})],...,P_{\mathbf{V}_{\sigma(n)}}B_{\mathbf{V}_{\sigma(n)}}(t_{n})]P_I\otimes \gamma_E) \label{eq:Hnest_proof},
\end{align}
    Notice that for the sum above to vanish it is enough for one of the nested commutators to commute with all the terms that it is nested around. Prior works have made this precise and counted the number of non-vanishing terms in sums of this form \cite{haah2024learning}. To each cluster $\mathbf{V}$ we associate a graph $G_\mathbf{V}$ whose vertices are the elements $a$ of $\mathbf{V}$, with $\mu(a)$ vertices for each $a$, and an edge between vertices $a,b$ if $[P_aB_a(t),P_bB_b(t)]\neq 0$. A cluster $\mathbf{V}$ is said to be connected if $G_{\mathbf{V}}$ is connected. Let $\mathcal{G}_n$ be the set of connected clusters of size $n$. Consider the union of two clusters: $\mathbf{W}=\mathbf{V}_1\cup\mathbf{V}_2$, satisfying $\mu_\mathbf{W}(a) = \mu_{\mathbf{V}_1}(a)+\mu_{\mathbf{V}_2}(a)$. We say that $\mathbf{V}$ is completely connected to $P_bB_a(t)$ if and only if $G_{\mathbf{V}\cup\{(b,1)\}}$ is connected. Let the set of clusters of weight $n$ that are connected to $P_bB_b(t)$ be $\mathcal{G}_n^b$. In the following Lemma we collect some results from previous works that quantify the number of nonzero terms in the sum above.
    
    \begin{lemma}[Adapted Proposition 3.6 \cite{haah2024learning}, Lemma 1, 2 \cite{wild2023classical}]
        For any cluster $\mathbf{V}\not\in \mathcal{G}_n^b$ and distinct indices $i_1,...,i_n\in \{1,...,n\}$:
        \begin{align}
            &\tr([[P_O,P_{\mathbf{V}_{i_1}}B_{\mathbf{V}_{i_1}}(t_{1})],...,P_{\mathbf{V}_{i_n}}B_{\mathbf{V}_{i_n}}(t_{n})]P_I\otimes \gamma_E )=0,.
        \end{align}
        Moreover, $|\mathcal{G}_n^b|\leq (e\mathfrak{d})^n$. 
    \end{lemma}

    Therefore, in Eq. \eqref{eq:Hnest_proof} the only nonvanishing terms in the sum over $\mathcal{C}_n$ are those in $\mathcal{G}_n^{b}$, where $P_bB_b(t)$ does not commute with with $P_O$. Since $P_O$ is supported in $\mathcal{I}_{a,b}$, which overlaps with at most $(\mathfrak{d}_0+2)(\mathfrak{d}+1)$ Pauli strings in $\mathcal{P}_S\cup\mathcal{P}_{SE}$, there are as many choices of $P_b$. Then, if we let $\mathcal{G}_n^O:= \bigcup_{[P_O,P_bB_b(t)]\neq 0} \mathcal{G}_n^b$, the sum over $\mathcal{C}_n$ is replaced by a sum over $|\mathcal{G}_n^O|\leq (\mathfrak{d}_0+2)(\mathfrak{d}+1)(e\mathfrak{d})^n$ terms:
    \begin{align}
    \eqref{eq:Hnest_proof}&=\frac{1}{2^N}\sum_{\mathbf{V}\in \mathcal{G}^O_n}\frac{1}{\mathbf{V}!}\sum_{\sigma \in S_n}\tr([[P_O,P_{\mathbf{V}_{\sigma(1)}}B_{\mathbf{V}_{\sigma(1)}}(t_{1})],...,P_{\mathbf{V}_{\sigma(n)}}B_{\mathbf{V}_{\sigma(n)}}(t_{n})]P_I\otimes \gamma_E),\\
    &=\frac{1}{2^N}\sum_{\mathbf{V}\in \mathcal{G}^O_n}\frac{1}{\mathbf{V}!}\sum_{\sigma \in S_n}\tr(P_O,[P_{\mathbf{V}_{\sigma(1)}}B_{\mathbf{V}_{\sigma(1)}}(t_{1}),...,[P_{\mathbf{V}_{\sigma(n)}}B_{\mathbf{V}_{\sigma(n)}}(t_{n}),P_I\otimes \gamma_E]]).
\end{align}
\end{proof}

\subsection{\label{app:error_analysis} Error analysis and sample complexity}

We now estimate the error accumulated in correcting the contributions from the offset. Recall that we estimate the $m-2$-th kernel derivative $K_{a,b}^{(m-2)}(0)$ from the $m$-derivative of the time traces $B_{W,(O,I)}^{(m)}(0)$. From Eq. \eqref{eq:time_traces_equations} we obtain the precision at which we perform tomography on $T_W^{(m)}$:
\begin{align}
    \epsilon_{T,m} &=\frac{1}{2^N}|\tr(P_O T_W^{(m)}(P_I)) - \tr(P_O \hat{T}_W^{(m)}(P_I))| = |B_{W,(O,I)}^{(m)}(0)- \hat{B}_{W,(O,I)}^{(m)}(0) + f_{W,(O,I)}^{(m)}-\hat{f}_{W,(O,I)}^{(m)}|\leq \epsilon_{S,m} + \epsilon_{f,m},
\end{align}

where $\epsilon_{S,m}$, $m=0,...,M+2$ is the precision in the derivative of the time trace $B_{W,(O,I)}^{(m)}(0)$ obtained from the polynomial fit in Lemma \ref{lem:samples} and $\epsilon_{f,m}$ is the precision in reconstructing $f_{W,(O,I)}^{(m)}$ in classical postprocessing, which we bound in Lemma \ref{lem:offfset_precision}. Moreover, Lemma \ref{lem:Ainverse} tells us that a precision $\epsilon_{T,m}$ in the tomography observables leads to estimates of the real and imaginary parts of $K_{a,b}^{(m-2)}(0)$ to the same precision: $\epsilon_{K,m-2} = \epsilon_{T,m}$. Therefore, we obtain:
\begin{align}
    &\epsilon_{K,m-2} \leq \frac{\epsilon_{S,M+2}}{d^{2(M+2-m)}}\frac{(2M+3)!!}{(2m+3)!!} + 3(1+\epsilon)^{m}(\mathfrak{d}_0+2)(\mathfrak{d}+1)(e\mathfrak{d})^{m}2^{3m}m^{2(m+1)} \cdot \epsilon_{K,m-3}.
\end{align}
Note that $m\leq M+2\leq d$, so the first term is less than $\epsilon_{S,M+2}$. This yields the following bound on the precision in the highest kernel derivative $\epsilon_{K,M}$:
\begin{align}
    &\epsilon_{K,M} = e^{O(M^2\log M)}(\epsilon_{S,M+2}+\epsilon_{K,0}).
\end{align}
Recall that estimating $K_{a,b}^{(0)}(0)$ requires computing the offset $f_{\Id^{\otimes N},(O,I)}^{(2)}$ from Eq.\eqref{eq:f2}:
\begin{align}
    &f_{\Id^{\otimes N},(O,I)}^{(2)} = \frac{-4}{2^N}\tr([P_O,H_S][H_S,P_I]). \label{eq:precision_K0}
\end{align}
Since $\mathcal{I_{a,b}}$ overlaps with at most $(2+\mathfrak{d_0})\mathfrak{d}$ terms in the Hamiltonian, we get that $P_O,P_I$ can each commute with at most that many terms. Therefore, $\epsilon_{K,0}$ is bounded by the precision in the Hamiltonian:
\begin{align}
    &\epsilon_{K,0} \leq \frac{\epsilon_{S,M+2}}{d^{2M}}\frac{(2M+3)!!}{7\cdot 5\cdot 3} + 4 (\mathfrak{d}_0+2)^2\mathfrak{d}^2 \epsilon_H
\end{align}
Recall from Eq.\eqref{eq:H_learning} that the precision for Hamiltonian learning is related to the sample precision by:
\begin{align}
    &\epsilon_H= \frac{1}{8}\epsilon_{S,1} \leq \frac{\epsilon_{S,M+2}}{8d^{2(M+1)}}\frac{(2M+3)!!}{5\cdot 3},
\end{align}
so we finally obtain the bound on $\epsilon_{K,M}$ depending only on sample precision:
\begin{align}
    \epsilon_{K,m} &\leq e^{ O(M^2\log M)}(\epsilon_{S,M+2}+\frac{\epsilon_{S,M+2}}{d^{2M}}\frac{(2M+3)!!}{105} + 4 (\mathfrak{d}_0+2)^2\mathfrak{d}^2 \epsilon_H),\notag\\
    &\leq e^{ O(M^2\log M)}\Big(1+\frac{(2M+3)!!}{105\cdot d^{2M}} + \frac{ (\mathfrak{d}_0+2)^2\mathfrak{d}^2(2M+3)!!}{30d^{2(M+1)}}\Big)\epsilon_{S,M+2},\notag \\
    &\leq O(e^{ O(M^2\log M)})\epsilon_{S,M+2},
\end{align}
where we used that, since $M\leq d$, the second and third summand decrease with $d$.

\begin{lemma}\label{lem:offfset_precision}
    Let $\vec{x}$ be a vector of length $M+2\mathfrak{d}N(m-2)$, containing the parameters of the Hamiltonian, as well as the real and imaginary parts of the kernel Taylor coefficients of degree $0,...,m-3$. Given estimates $\hat{\vec{x}}$ satisfying $||\hat{\vec{x}}-\vec{x}||_{\infty} <\epsilon$ we have:
    \begin{align}
        &|f_{\Id^{\otimes N},(O,I)}^{(m)}(\vec{x}) -f_{\Id^{\otimes N},(O,I)}^{(m)}(\hat{\vec{x}})|\leq 3(1+\epsilon)^m(\mathfrak{d}_0+2)^{2m}\mathfrak{d}^{m(m+1)} 2^{3m}m^{m+2} \cdot \epsilon .
    \end{align}
    The time complexity to compute the estimate is at most $O((\mathfrak{d}_0+2)^m\mathfrak{d}^{\frac{m(m+1)}{2}}m^{\frac{3}{2}m+6}2^{-4m})$.
\end{lemma}
\begin{proof}
Using the multivariate mean-value theorem, there is some $s\in [0,1]$ such that:
\begin{align}
    &f_{\Id^{\otimes N},(O,I)}^{(m)}(\vec{x}) -f_{\Id^{\otimes N},(O,I)}^{(m)}(\hat{\vec{x}}) = (\nabla_{\vec{x}}f_{\Id^{\otimes N},(O,I)}^{(m)}((1-s)\hat{\vec{x}}+s\vec{x}))\cdot (\vec{x}-\hat{\vec{x}})\Rightarrow \\
    &\Rightarrow |f_{\Id^{\otimes N},(O,I)}^{(m)}(\vec{x}) -f_{\Id^{\otimes N},(O,I)}^{(m)}(\hat{\vec{x}})|\leq ||\nabla f_{\Id^{\otimes N},(O,I)}^{(m)}((1-s)\hat{\vec{x}}+s\vec{x})||_1||\vec{x}-\hat{\vec{x}}||_\infty
\end{align}
Where we used Hölder's inequality. As shown in Eq.\eqref{eq:f_expansion}, $f^{(m)}_{\mathds{1}^{\otimes N},(O,I)}$ does not depend on all entries in $\vec{x}$, rather only those labelled by the index set $\mathcal{G}_{m}^O$, which is at most of size $|\mathcal{G}_m^O|\leq (\mathfrak{d}_0+2)(\mathfrak{d}+1)(e\mathfrak{d})^m$. Therefore, only as many entries in $\nabla f^{(m)}_{\Id^{\otimes N},(O,I)}$ are nonzero. We can then bound $||\nabla f_{\Id^{\otimes N},(O,I)}^{(m)}((1-d)\hat{\vec{x}}+s\vec{x})||_1\leq (\mathfrak{d}_0+2)(\mathfrak{d}+1)(e\mathfrak{d})^m ||\nabla f_{\Id^{\otimes N},(O,I)}^{(m)}((1-d)\hat{\vec{x}}+s\vec{x})||_\infty$. Using Lemma \ref{lem:bound_gradient} this can be further bounded by:
\begin{align}
    |f_{\Id^{\otimes N},(O,I)}^{(m)}(\vec{x}) -f_{\Id^{\otimes N},(O,I)}^{(m)}(\hat{\vec{x}})|&\leq (\mathfrak{d}_0+2)(\mathfrak{d}+1)(e\mathfrak{d})^m\cdot 3(1+\epsilon)^{m}(\mathfrak{d}_0+2)(\mathfrak{d}+1)(e\mathfrak{d})^m2^{3m}m^{2(m+1)}\cdot \epsilon,\notag\\
    &=3(1+\epsilon)^m(\mathfrak{d}_0+2)^{2}(\mathfrak{d}+1)^2(e\mathfrak{d})^{2m} 2^{3m}m^{2(m+1)} \cdot \epsilon .
\end{align}

In order to obtain a bound on the the time to compute $f_{\Id^{\otimes N},(O,I)}^{(m)}$ from Eq.\eqref{eq:f_expansion} we can follow a similar procedure to count the sums as done in Lemma \ref{lem:bound_gradient}. Taking into account the time required to compute $v_{c,q}^{z,p}$ and the products of system parameters and kernels  we obtain a time complexity $O((\mathfrak{d}_0+2)(\mathfrak{d}+1)(e\mathfrak{d})^m m^{\frac{3}{2}m+6}2^{-4m})$.
\end{proof}

\begin{lemma}\label{lem:bound_gradient}
    Given the real value of the Hamiltonian parameters and the kernel's derivatives up to order $m-3$, $\vec{x}$, and an estimate $\hat{\vec{x}}$ such that $||\vec{x}-\hat{\vec{x}}||_{\infty}<\epsilon$, the derivative $\nabla f_{\Id^{\otimes N},(O,I)}^{(m)}((1-s)\hat{\vec{x}}+s\vec{x})$ satisfies the following bound for all $s\in[0,1]$:
    \begin{align}
        &||\nabla f_{\Id^{\otimes N},(O,I)}^{(m)}((1-s)\hat{\vec{x}}+s\vec{x})||_\infty\leq 3(1+\epsilon)^{m}(\mathfrak{d}_0+2)(\mathfrak{d}+1)(e\mathfrak{d})^{m}2^{3m}m^{2(m+1)}.
    \end{align}
\end{lemma}
\begin{proof}
Let $x_*$ be the coordinate that attains the maximum such that $||\nabla f_{\Id^{\otimes N},(O,I)}^{(m)}((1-s)\hat{\vec{x}}+s\vec{x})||_\infty = |\partial_{x_*} f_{\Id^{\otimes N},(O,I)}^{(m)}((1-s)\hat{\vec{x}}+s\vec{x})|$. Since $f_{\Id^{\otimes N},(O,I)}^{(m)}(\vec{y})$ is a polynomial of degree at most $m$ in every variable $y_i$, the derivative $\partial_{x_*} f_{\Id^{\otimes N},(O,I)}^{(m)}(\vec{y})$ is a polynomial of degree at most $m-1$ in $x_*$ and degree $m$ in the other variables. Assume $x_* = \text{Re}[K_{a,b}^{(u)}(0)]$ is a kernel parameter (the cases with the imaginary part and where it is a system parameter follow similarly). Using Eq.\eqref{eq:f_expansion} and the triangle inequality we can bound the derivative as:
\begin{align}
    &|\partial_{x_*}(f_{\Id^{\otimes N},(O,I)}^{(m)}((1-s)\hat{\vec{x}}+s\vec{x}))| \leq \sum_{n=3}^m\sum_{\mathbf{V}\in \mathcal{G}_n^O}\frac{1}{\mathbf{V}!}\sum_{\sigma \in S_n}\sum_{l=0}^n\sum_{q\in Q(n,l)}\prod_{i|c_{q_i\in c_S}}|(1-s)\hat{\lambda}_{c_{q_i}}+s\lambda_{c_{q_i}}|\cdot\sum_{p\in\mathcal{P}_{2u_{c}}^2}\sum_{z_1+...+z_{u_c}=m-n}v^{z,p}_{c,q} \times \notag\\
    &\hspace{70pt} d_{*,c,q}^{z,p} \cdot s|(1-s)\hat{x}_*+sx_*|^{d_{*,c,q}^{z,p}-1}\prod_{\substack{j=1\\ c_{q_{p_{j,1}}'}\neq a,c_{q_{p_{j,2}}'}\neq b}}^{u_{c}}|(1-s)\hat{K}_{c_{q_{p_{j,1}}'}c_{q_{p_{j,2}}'}}^{(z_j)}(0)+sK_{c_{q_{p_{j,1}}'}c_{q_{p_{j,2}}'}}^{(z_j)}(0)|, 
\end{align}
where $c_i = \mathbf{V}_{\sigma (i)}$. Here we used $|\frac{1}{2^N}\tr(PQ)|\leq 1$ for Pauli strings $P,Q$ and let $d_{*,c,q}^{z,p}$ be the degree of the kernel coefficient $x_*$ in the product labeled by $c,q,z,p$. In order to bound the product of system coefficients we use $|(1-s)\hat{\lambda}_i+s\lambda_i| = |\lambda_i+(1-s)(\hat{\lambda}_i-\lambda_i)|\leq 1+\epsilon$. We similarly bound $|(1-s)\hat{K}_{c_i,c_j}^{(z_j)}+sK_{c_i,c_j}^{(z_j)}|\leq \sqrt{2}(1+\epsilon)$, accounting for the real and imaginary parts. So we obtain:
\begin{align}
    &|\partial_{x_*}(f_{\Id^{\otimes N},(O,I)}^{(m)}((1-s)\hat{\vec{x}}+s\vec{x}))| \leq \sum_{n=3}^m\sum_{\mathbf{V}\in \mathcal{G}_n^O}\frac{1}{\mathbf{V}!}\sum_{\sigma \in S_n}\sum_{l=0}^n\sum_{q\in Q(n,l)}\sum_{p\in\mathcal{P}_{2u_{c}}^2}\sum_{z_1+...+z_{u_c}=m-n}|v^{z,p}_{c,q}|d_{*,c,q}^{z,p}(\sqrt{2})^{u_{c}-1}(1+\epsilon)^{n-2u_{c}+u_{c}-1}
\end{align}

We used that if the number of kernels is $u_{c}$, then the number of Hamiltonian parameters is $n-2u_{c}$. The highest degree of a Hamiltonian term $d_{*,c,q}^{z,p}$ is $n-2u_{c}\leq n$ and the
maximum number of kernels is $\frac{n}{2},\frac{n-1}{2}$ if $n$ is even or odd, respectively, so $u_{c}\leq \lfloor \frac{n}{2}\rfloor$ we obtain:
\begin{align}
    &|\partial_{x_*}(f_{\Id^{\otimes N},(O,I)}^{(m)}((1-s)\hat{\vec{x}}+s\vec{x}))| \leq \sum_{n=3}^m n2^{\frac{1}{2}\lfloor \frac{n}{2}\rfloor}(1+\epsilon)^{n}\sum_{\mathbf{V}\in \mathcal{G}_n^O}\frac{m!}{\mathbf{V}!}\sum_{l=0}^n\sum_{q\in Q(n,l)}\sum_{p\in\mathcal{P}_{2u_{c}}^2}\sum_{z_1+...+z_{u_c}=m-n} |v^{z,p}_{c,q}|
\end{align}
As discussed in Lemma \ref{lem:cluster_bound} , $|\mathcal{G}_n^O| \leq (\mathfrak{d}_0+2)(\mathfrak{d}+1)(e\mathfrak{d})^{n}$. Since $Q(n,l)$ entails choosing $l$ distinct numbers among $\{1,...,n\}$, $|Q(n,l)| = {n\choose l}$.
We know that the number of partitions of $\{1,...,2u_c\}$ grows as $|\mathcal{P}_{2u_c}^2| = \frac{(2u_c)!}{2^{u_c} u_c!}$. Since this is an increasing function of $u_c \leq \lfloor \frac{n}{2}\rfloor$, we can upper bound it by $\frac{n!}{2^{\lfloor \frac{n}{2}\rfloor }\lfloor \frac{n}{2}\rfloor !}$. From Lemma \ref{lem:bound_v} we know that
$v_{c,q}^{z,p}$ is bounded by:
\begin{align}
    &|v_{c,q}^{z,p}|\leq \frac{2^{m}m!}{n!z_1!...z_u!}  ,
\end{align}
The sum over weak compositions can be simplified \footnote{This is just the multinomial theorem: $a^{b}=(\sum_{j=1}^{a}1)^{b}=\sum_{z_1+...+z_a=b}\frac{b!}{z_1!...z_a!}$.} using $\sum_{z_1+...+z_{u_c}=m-n}\prod_{j=1}^{u_c} \frac{1}{z_j!} = \frac{u_c^{m-n}}{(m-n)!}\leq \frac{ (\frac{n}{2})^{m-n}}{(m-n)!}$ . We will use $\frac{1}{\mathbf{V} ! }\leq 1$ and bound:
\begin{align}
    |\partial_{x_*}(f_{\Id^{\otimes N},(O,I)}^{(m)}((1-s)\hat{\vec{x}}+s\vec{x}))|&\leq m!\sum_{n=3}^m n2^{\frac{1}{2}\lfloor \frac{n}{2}\rfloor}(1+\epsilon)^{n}(\mathfrak{d}_0+2)(\mathfrak{d}+1)(e\mathfrak{d})^{n}\times \notag\\
    &\hspace{74pt}\sum_{l=0}^n{n\choose l}\cdot \frac{n!}{2^{\lfloor \frac{n}{2}\rfloor}\lfloor \frac{n}{2}\rfloor!}\cdot \frac{2^{m}m!}{n!}\cdot 2^{n-m}\frac{n^{m-n}}{(m-n)!},\\
    &\leq (1+\epsilon)^{m}(\mathfrak{d}_0+2)(\mathfrak{d}+1)(e\mathfrak{d})^{m}(m!)^2\cdot 3\cdot\sum_{n=3}^m  2^{2n} \frac{n^{m-n}}{\lfloor \frac{n}{2}\rfloor !(m-n)!},
\end{align}
where we used $\sum_{l=0}^n{n\choose l}=2^n$ and $n\leq 3\cdot 2^{\frac{1}{2}\lfloor \frac{n}{2}\rfloor}$ to simplify the factors of $2$. We simplify the last sum using ${m\choose n} \leq 2^m, m! \leq m^m$:
\begin{align}
    &(m!)^2\sum_{n=3}^m 2^{2n}\frac{n!}{\lfloor \frac{n}{2}\rfloor !}\frac{1}{n!(m-n)!}n^{m-n}\leq m!\cdot 2^{3m}\sum_{n=3}^m n^{\frac{n}{2}+1} n^{m-n}\leq m!\cdot 2^{3m}m^{m} \sum_{n=3}^m n^{1-\frac{n}{2}}\leq m!2^{3m}m^{m+2}\leq 2^{3m}m^{2(m+1)}.
\end{align}

\end{proof}

\begin{lemma}\label{lem:bound_v}
    Let $p_i(t)=\sum_{k=0}^\infty \frac{p_i^{(k)}(0)}{k!}t^k$ be Taylor series for $i=1,...,u$. Let $s_1,...,s_{2u}\in \{t_1,...,t_n\}$ be distinct time variables. Then we can evaluate:
\begin{align}
    &\partial_t^m \int_{t}^{2t}dt_1\int_{t}^{t_1}dt_2...\int_{t}^{t_{l-1}}dt_l\int_0^{t}dt_{l+1}\int_0^{t_{l+1}}dt_{l+2}...\int_0^{t_{n-1}}dt_n \prod_{i=1}^u p_i(s_{2i}-s_{2i-1})\big|_{t=0}\notag\\\
    &\hspace{240pt}=\sum_{z_1+...+z_u=m-n}v^{z,s}\prod_{i=1}^u p_i^{(z_i)}(0),
\end{align}
    where the coefficient $v^{z,s}$, with $z_1+...+z_u=m-n \geq 0$ satisfies the upper bound:
    \begin{align}
        &|v^{z,s}|\leq \frac{2^{m-n}m!}{(n-l)!l!z_1!...z_u!}.
    \end{align}
    The coefficient $v^{z,q}$ can be computed in time $O(m^{l+1}n^2l^{3-l}u2^{4m-2n-2l})$. 
\end{lemma}
\begin{proof}
We use the absolute convergence of the Taylor series to rewrite their product as:
\begin{align}
    \prod_{i=1}^u p_i(s_{2i}-s_{2i-1}) &= \prod_{i=1}^u \sum_{k=0}^{\infty} \frac{p_i^{(k)}(0)}{k!}(s_{2i}-s_{2i-1})^k,\notag\\
    &=\sum_{k=0}^{\infty} \sum_{z_1+...+z_u=k}\prod_{i=1}^u(s_{2i}-s_{2i-1})^{z_i}\prod_{i=1}^u\frac{p_i^{(z_i)}(0)}{z_i!},\notag\\
    &= \sum_{k=0}^{\infty} \sum_{z_1+...+z_u=k}\prod_{i=1}^u\sum_{j_i=0}^{z_i} {z_i\choose j_i}s_{2i}^{j_i}(-s_{2i-1})^{z_i-j_i}\prod_{i=1}^u\frac{p_i^{(z_i)}(0)}{z_i!}.\label{eq:v_2}
\end{align}
After taking $n$ integrals, we see that the $k$-th summand depends on time as $t^{\sum_{i=1}^u z_i +n}=t^{k+n}$. Therefore, if we take $m > k + n$ derivatives, the $k$-th summand vanishes; while if we take $m<k+n$ derivatives, the power of $t$ will be $\geq 1$, so it vanishes by setting $t=0$. Thus, only the summand $k=m-n$ remains, so we obtain the following expression for $v^{z,s}$:
\begin{align}
    v^{z,s} &= \partial_t^m\int_t^{2t}dt_1...\int_t^{t_{l-1}}dt_l\int_0^tdt_{l+1}...\int_0^{t_{n-1}}dt_n \prod_{i=1}^u\sum_{j_i=0}^{z_i}{z_i\choose j_i}\frac{1}{z_i!}s_{2i}^{j_i}(-s_{2i-1})^{z_i-j_i}\big|_{t=0},\notag\\
    &=\sum_{j_1=0}^{z_1}...\sum_{j_u=0}^{z_u}(-1)^{\sum_iz_i-j_i}\prod_{i=1}^u{z_i\choose j_i}\frac{1}{z_i!}\partial_t^m\int_t^{2t}dt_1...\int_t^{t_{l-1}}dt_l\int_0^tdt_{l+1}...\int_0^{t_{n-1}}dt_n  \prod_{i=1}^u s_{2i}^{j_i}s_{2i-1}^{z_i-j_i}\big|_{t=0}.\label{eq:v_3}
\end{align}
Using Lemma \ref{lem:I_n} and defining $d_{1},...,d_n$ by $\prod_{i=1}^u s_{2i}^{j_i}s_{2i-1}^{z_i-j_i} = \prod_{i=1}^n t_i^{d_i}$ we can write the integral as follows:
\begin{align}
    \int_t^{2t}dt_1...\int_t^{t_{l-1}}dt_l \prod_{i=1}^l t_i^{d_i}\int_0^tdt_{l+1}...\int_0^{t_{n-1}}dt_n  \prod_{i=l+1}^n t_i^{d_i} &=\int_t^{2t}dt_1...\int_t^{t_{l-1}}dt_l \prod_{i=1}^l t_i^{d_i} I_{n-l}(d^{[n-l]},t),\notag \\
    &=t^{n-l+||d^{[n-l]}||_1}\prod_{r = 1}^{n-l}\frac{1}{r+||d^{[r]}||_1}\int_t^{2t}dt_1...\int_t^{t_{l-1}}dt_l \prod_{i=1}^l t_i^{d_i},
\end{align}

and changing variables $v_i=t_i-t, i=1,...l$ gives us:
\begin{align}
    &=t^{n-l+||d^{[n-l]}||_1}\prod_{r = 1}^{n-l}\frac{1}{r+||d^{[r]}||_1}\int_0^{t}dv_1...\int_0^{v_{l-1}}dv_l \prod_{i=1}^l (v_i+t)^{d_i} ,\notag\\
    &=t^{n-l+||d^{[n-l]}||_1}\prod_{r = 1}^{n-l}\frac{1}{r+||d^{[r]}||_1}\sum_{w_1=0}^{d_1}...\sum_{w_l=0}^{d_l}{d_1\choose w_1}...{d_l\choose w_l}t^{\sum_{i=1}^l  d_i-w_i}\int_0^{t}dv_1...\int_0^{v_{l-1}}dv_l \prod_{i=1}^l v_i^{w_i},
\end{align}
Using $w = (w_1,...,w_l)$, the integrals are now $I_l(w,t)$ and we can collect the powers of $t$:
\begin{align}
    &t^{n-l+||d^{[n-l]}||_1}\cdot t^{\sum_{i=1}^l d_i -w_i}\cdot I_l(w,t) = t^{n-l+||d||_1-||w||_1}\cdot t^{l+||w||_1}\prod_{r=1}^l \frac{1}{r+||w^{[r]}||_1}= t^{n+||d||_1}\prod_{r=1}^l \frac{1}{r+||w^{[r]}||_1}.
\end{align}
Recall that while we defined $d_1,...,d_n$ implicitly, we know that $||d||_1=||z||_1=m-n$. After taking $m$ derivatives with respect to $t$ and setting $t=0$ we obtain:
\begin{align}
    &v^{z,s}=m!\sum_{j_1=0}^{z_1}...\sum_{j_u=0}^{z_u}(-1)^{\sum_{i=1}^uz_i-j_i}\prod_{i=1}^u{z_i\choose j_i}\frac{1}{z_i!} \prod_{r=1}^{n-l}\frac{1}{r+||d^{[r]}||_1}\cdot\prod_{r=1}^l \frac{1}{r+||w^{[r]}||_1}.\label{eq:v_zs}
\end{align}
We take the absolute value of this expression and upper bound it using the triangle inequality. The last two products can be upper bounded by the case $d=(0,...,0)$:
\begin{align}
    &\prod_{r=1}^{n-l}\frac{1}{r+||d^{[r]}||_1}\cdot\prod_{r=1}^l \frac{1}{r+||w^{[r]}||_1} \leq \prod_{r=1}^{n-l}\frac{1}{r}\cdot\prod_{r=1}^l \frac{1}{r} \leq \frac{1}{(n-l)!l!}.
\end{align}
The remaining sums are upper bounded as follows, using that $\sum_{i=1}^u z_i =m-n$:
\begin{align}
    &\sum_{j_1=0}^{z_1}...\sum_{j_u=0}^{z_u} \prod_{i=1}^u{z_{i} \choose j_i}\frac{1}{z_{i}!}= \prod_{i=1}^{u}\frac{\sum_{j_i=0}^{z_i}{z_i\choose j_i}}{z_i!}  = \prod_{i=1}^{u}\frac{2^{z_{i}}}{z_i!} = \frac{2^{m-n}}{z_1!...z_u!}.
\end{align}
This yields the final bound:
\begin{align}
    &|v^{z,s}|\leq \frac{2^{m-n}m!}{(n-l)!l!z_1!...z_u!}
\end{align}
From Eq. \eqref{eq:v_zs} we can compute $v^{z,s}$. The factorial $m!$ can be computed in time $O(m)$. Similarly, ${d_1\choose w_1}$ can be computed in time $O(d_1)$ and using the arithmetic mean-geometric mean inequality yields that the product of binomial coefficients can be computed in time $O(l\cdot (\frac{m-l}{l})^{l})$. The sum $\sum_{i=1}^u z_i-j_i$ is computed in time $O(u)$.
The product $\prod_{r=1}^l\frac{1}{r+||w^{[r]}||_1}$ can be computed in time $O(l^2)$. The number of terms in the sums over $z_i$ can be bounded as above by $4^{m-n}$, and similarly the number of terms in sums over $d_i$ is bounded by $4^{m-l}$. Therefore,we can compute $v^{z,s}$ in time $O(m^{l+1}n^2l^{3-l}u2^{4m-2n-2l})$.

\end{proof}

\begin{lemma}\label{lem:I_n} Let $d=(d_1,...,d_n)$ be a vector of $n$ nonnegative integers, and denote $||d||_1=\sum_{i=1}^n|d_i|$ the sum of its entries. Consider the integral
    \begin{align}
    &I_n(d,t):=\int_0^tdt_1...\int_0^{t_{n-1}}dt_n \prod_{s=1}^n t_s^{d_s}.
\end{align}
Let $d^{[r]}$ be the vector of the last $r$ entries of $d$, then we have:
\begin{align}
    &I_n(d,t)= t^{n+||d||_1}\prod_{r=1}^n\frac{1}{r+||d^{[r]}||_1}
\end{align}
\end{lemma}
\begin{proof}
We use induction on $n$. Base case, $n=1$:
\begin{align}
    I_1((d_1),t)&=\int_0^tdt_1 t_1^{d_1} = \frac{1}{d_1+1}t^{d_1+1}.
\end{align}
Inductive step: we use the fact that $(d_2,...,d_n)=d^{[n-1]}$.
\begin{align}
    I_n(d,t) &= \int_0^tdt_1t_1^{d_1}I_{n-1}(d^{[n-1]},t_1)=\int_0^tdt_1t_1^{d_1}t_1^{n-1+||d^{[n-1]}||_1}\prod_{r=1}^{n-1}\frac{1}{r+||(d^{[n-1]})^{[r]}||_1},
\end{align}
In the exponent of $t_1$ we have $d_1+||d^{[n-1]}||_1=||d||_1$, and in the denominator we get $(d^{[n-1]})^{[r]}=d^{[r]}$, since $r\leq n-1$. Therefore, using $d=d^{[n]}$ we obtain:
\begin{align}
    &I_n(d,t)=\frac{t^{n+||d||_1}}{n+||d||_1}\prod_{r=1}^{n-1}\frac{1}{r+||d^{[r]}||_1}=t^{n+||d||_1}\prod_{r=1}^{n}\frac{1}{r+||d^{[r]}||_1}.
\end{align}

\end{proof}

\subsection{\label{app:numerics}Numerical observables}
The $N$ fermionic modes of the system and $N$ modes of the bath are mapped to $4N$ Majorana modes, for $i=1,...,N$:
\begin{align}
    c_{2i-1}&=a_i+a_i^\dagger, \hspace{30pt}c_{2i}=i(a_i-a_i^\dagger).\\
    c_{2(i + N)-1}&=b_{i}+b_{i}^\dagger, \hspace{13pt}c_{2(i+N)}=i(b_{i}-b_{i}^\dagger) .
\end{align}
We will compute time traces of the form $O_{a,b,c,d}(t) = \tr(ic_ac_b \rho_{c,d}(t))$, where $\rho_{c,d}(0)=\frac{1}{2}(\Id+ic_cc_d)\otimes \gamma$ and $\gamma = (|0\rangle\langle 0|)^{\otimes 2N}$. The time traces required to learn the derivatives of the kernel $K_{i,j}(t)$ are:
\begin{align}
    &\text{If }i\neq j: O_{2i-1,l,2j-1,l}, \text{ where }l\neq i,j\\
    &\text{If }i= j: O_{2i-1,l,2i-1,l}, O_{l,k,l,k},O_{k,2i-1,k,2i-1} \\
    &\hspace{92pt}\text{ where }l,k\neq i, l\neq k
\end{align}

\subsection{\label{app:lieb_robinson_nm} Time evolution of local observables for constant time}

In Appendix D of \cite{stilck2024efficient}, they show that given a Lieb-Robinson bound on the Heisenberg evolution of local observables and a maximum evolution time $t_{\max}$, the time evolution of local observables at constant times can be approximated by polynomials of controlled degree. In particular, in Proposition D.1 in \cite{stilck2024efficient} they show that for $t=O(1)$, the expectation value of observables under the global evolution can be approximated by that of a local evolution. From this they derive their Theorem D.2, which also bounds the closeness of the derivatives. In order to obtain Lemma \ref{lem:franca_liebrobinson} we only need to change Proposition D.1 in their proof for an equivalent Lemma in the non-Markovian setting. In the non-Markovian case with unbounded system-environment Hamiltonian it is not expected that a Lieb-Robinson bound on operators holds, as there are initial system-environment states for which the Hamiltonian interaction terms are unbounded \cite{trivedi2024lieb}. However, a Lieb-Robinson bound has been obtained for the initial states of the form $\rho_S\otimes \gamma_E$ \cite{trivedi2024lieb}.

We introduce some notation to state this bound. Let $a[l]$ be the set of sites that are within distance $l$ of the Pauli string $P_a$ and $H_{a[l]}=\sum_{b:\mathcal{S}_b\cap a[l]\neq \emptyset}\lambda_bP_b + \sum_{b:\mathcal{S}_b\cap a[l]\neq \emptyset}P_bA_b(t)$. Consider the Heisenberg evolution of an observable $P_a$ under $H(t)$ and the restricted Hamiltonian $H_{a[l]}$: $P_a(t)=U^\dagger(t)P_aU(t), P_a(t;l)=U_{a[l]}^\dagger(t)P_aU_{a[l]}(t)$, where $U_{a[l]}(t) = \mathcal{T}\exp(-i\int_0^t H_{a[l]}(s)ds)$. The physically relevant error for an initial state $\rho(0)$ is $\Delta_{P_a}(t;l)=|\tr((P_a(t)-P_a(t;l))\rho(0))|$, this measures the deviation in the observables for a specific initial state. In our case we only need to bound the error when the environment state is a Gaussian state. 

\begin{lemma}[Adapted Proposition 1,\cite{trivedi2024lieb}]
    Given $P_a\subseteq \Lambda$, a $d$-dimensional lattice, there exists $v_{\text{LR}}>0$ such that, for all initial states $\rho_S\otimes \gamma_E$, with $\gamma_E$ a Gaussian state, the physically relevant error in the expectation value is bounded:
    \begin{align}
        |\Delta_{P_a}(t;l)|&=|\tr((P_a(t)-P_a(t;l))\rho_S\otimes \gamma_E)|,\notag \\
        &\leq f(l)\exp(-l/a_0)(\exp(v_{\text{LR}}|t-t'|/a_0)-1),
    \end{align}
    where, for large $l$, $f(l)\leq O(l^{d-1})$ and $a_0\leq k$. \label{lem:lieb-robinson_nonmarkovian}
\end{lemma}
This Lieb-Robinson bound applies to our setting: the environment modes are geometrically local, can be described as a free bosonic field interacting only with the system and we assume that the kernel has finite memory in the sense of Eq. \eqref{eq:total_variation_condition}.

\begin{lemma}[Adapted Theorem E.1, Proposition E.1, Corollary E.1 \cite{stilck2024efficient}, Theorem 1.2 \cite{arora2024outlier}]
    Let $p(t)$ be a polynomial of degree $d$ defined on an interval $[a,b]$ with $a= d^{-2}$ and $b=2+a$, $\epsilon_{S,M} >0$ a desired precision in the $M$-th derivative at $0$. Then for $\delta >0$, sampling 
    \begin{align}
        &O\bigg(d\log\bigg(\frac{d}{\delta}\bigg)\bigg),
    \end{align}
    i.i.d. points $(x_i,y_i)$ from the Chebyshev measure on $[a,b]$ satisfying
    \begin{align}
        &p(x_i)=y_i+w_i, |w_i|\leq \epsilon_{S}=  \epsilon_{S,M}\frac{(2M-1)!!}{3d^{2M}},
    \end{align}
    for at least a fraction $\alpha > \frac{1}{2}$ of the points is sufficient to obtain a polynomial $\hat{p}$ satisfying 
    \begin{align}
        &|p^{(m)}(0)-\hat{p}^{(m)}(0)| \leq \epsilon_{S,M}\frac{(2M-1)!!}{d^{2M}}\frac{d^{2m}}{(2m-1)!!}, m=0,1,...,M.
    \end{align}
    
\end{lemma}
\begin{proof}
    Recall from Eq. E6 of \cite{stilck2024efficient} the following version of Markov brother's inequality: let $q:[0,b]\to \mathbb{R}$ be a polynomial of degree $d$ and let $q^{(k)}$ be its $k$-th derivative. Then for $0\leq a\leq b$:
    \begin{align}
        &\max_{x\in[a,b]}\Big|q^{(k)}(x)\Big|\leq \bigg|\frac{2}{b-a}\bigg|^k C_M(d,k)\max_{x\in[a,b]}|q(x)|, \text{ where }C_M(d,k)=\frac{d^2(d^2-1)(d^2-2^2)...(d^2-(k-1)^2)}{(2k-1)!!},\label{eq:q_ab}
    \end{align}
    and $(2k-1)!! = (2k-1)(2k-3)\cdot ...\cdot 3\cdot 1$ is the double factorial. By the Taylor expansion on $x\in[0,a]$ we obtain:
    \begin{align}
        &q(x)=\sum_{k=0}^d q^{(k)}(a)\frac{(x-a)^k}{k!}\Rightarrow q^{(m)}(0)=\sum_{k=m}^d q^{(k)}(a)\frac{(-a)^{k-m}}{(k-m)!}.
    \end{align}
    We bound $|q^{(m)}(0)|$  using the triangle inequality and the bound on $|q^{(m)}(a)|$ from Eq.\eqref{eq:q_ab}:
    \begin{align}
        &\Big|q^{(m)}(0)\Big|\leq \sum_{k=m}^d \Big|q^{(k)}(a)\Big|\frac{a^{k-m}}{(k-m)!}\leq \sum_{k=m}^d \bigg|\frac{2}{b-a}\bigg|^k C_M(d,k)\max_{x\in[a,b]}|q(x)|\frac{a^{k-m}}{(k-m)!}.
    \end{align}
    We use $b=2+a$, $C_M(d,k)\leq \frac{d^{2k}}{(2k-1)!!}$ and $a=1/d^2$:
    \begin{align}
        &\Big|q^{(m)}(0)\Big|\leq \max_{x\in[a,b]}|q(x)|\sum_{k=m}^d\frac{d^{2k}}{(2k-1)!!}\frac{d^{-2(k-m)}}{(k-m)!}\leq \max_{x\in[a,b]}|q(x)| \frac{d^{2m+1}}{(2m-1)!!},
    \end{align}
    where in the last inequality we bound the sum by the largest summand, $k=m$, times $d$.\\

    By Theorem E.1 in \cite{stilck2024efficient}, originally Corollary 1.5 in \cite{kane2017robust}, for $\delta >0$, sampling $O(d\log(\frac{d}{\delta}))$ i.i.d. points $(x_i,y_i)$ from the Chebyshev measure on $[a,b]$ satisfying $p(x_i)=y_i+w_i, |w_i|\leq \sigma$ for at least a fraction $\alpha \geq 1/2$ of the points is sufficient to obtain a polynomial estimate $\hat{p}$ satisfying $\max_{x\in[a,b]}|p(x)-\hat{p}(x)|\leq 3\sigma$ with success probability at least $1-\delta$.\\
    
    Now let $q(x)=p(x)-\hat{p}(x), \epsilon_S := \sigma, \epsilon_{S,M}:=3\frac{d^{2M+1}}{(2M-1)!!}\epsilon_S$, so we obtain:
    \begin{align}
        &|p^{(m)}(0)-\hat{p}^{(m)}(0)|\leq 3\epsilon_S \frac{d^{2m+1}}{(2m-1)!!}=\epsilon_{S,M}\frac{(2M-1)!!}{d^{2M}}\frac{d^{2m}}{(2m-1)!!},
    \end{align}
    and note that, since $m\leq M\leq d$ we have
    \begin{align}
        &\frac{d^{2m}}{(2m-1)!!}\leq \frac{d^{2M}}{(2M-1)!!},
    \end{align}
    so all derivative estimates are at least as precise as the $M$-th derivative.
\end{proof}

\section{\label{app:RandomHam} Ensemble Hamiltonians}
Recall that $H_\Lambda =\sum_{a\in \mathcal{P}_S}\Lambda_aP_a$ is a geometrically local Hamiltonian on $N$ spin-$\frac{1}{2}$ particles with Pauli strings $P_a$ supported on at most $k$ sites. The vector of coefficients $\Lambda$ consists of jointly Gaussian random variables $\Lambda = B\cdot Z + \lambda$. We assume that the means are bounded $|\lambda_a|\leq 1$ and the covariance matrix $\Sigma = BB^T$ satisfies $\Sigma_{aa}\leq 1$. Given an observable $O$ with $||O||\leq 1$, let $O_\Lambda(t)$ be its time evolution in the Heisenberg picture under $H_\Lambda$ for time $t$ and a specific draw of $\Lambda$.

\subsection{Truncated series for local operators}
Recall that, even though the measurement strategy is the same as in the non-Markovian case, we can not use the non-Markovian Lieb-Robinson bounds because the covariance matrix is not necessarily sparse and the kernels don't satisfy the finite memory condition. Moreover, we can not apply Hamiltonian Lieb-Robinson bounds, since the random variables $\Lambda_a$ are not bounded and lead to unbounded Lieb-Robinson velocity. We show that the expected operator time evolution $\mathbb{E}_\Lambda(O_\Lambda(t))$ can still be approximated by a low-degree polynomial. 

\begin{lemma}
    Let $H_\lambda =\sum_{a=1}^L\Lambda_aP_a$ be a geometrically local Hamiltonian on $N$ spin-$\frac{1}{2}$ particles with Pauli strings $P_a$ supported on at most $k$ sites. $\{\Lambda_a\}_{a=1}^L$ are jointly Gaussian random variables with mean $\{\lambda_a\}_{a=1}^L$ satisfying $|\lambda_a|\leq 1$ and covariance matrix $\Sigma$ satisfying $|\Sigma_{aa}|\leq 1$ for $i,j=1,...,L$. Given an observable $O$ with $||O||\leq 1$, let $\mathbb{E}_\Lambda(O_\Lambda(t))$ be its expected time evolution in the Heisenberg picture for time $t$. This has a convergent Taylor series for $t \leq\Big(\frac{6ec}{a}\max\Big(1+\sqrt{2}\sqrt{\log(\frac{2}{\epsilon})+(d+2)\log(L)},10\Big)\Big)^{-1}$ and can be approximated to precision $\epsilon >0$ by a truncation of degree $d \geq \log(\frac{1}{\epsilon})$, $\mathbb{E}_{\Lambda,d}(O_\Lambda(t))$:
    \begin{align}
        &||\mathbb{E}_\Lambda(O_\Lambda(t))-\mathbb{E}_{\Lambda,d}(O_\Lambda(t))||\leq \epsilon.
    \end{align}
    Asymptotically, $t=O((\log (N)\cdot \log(1/\epsilon))^{-\frac{1}{2}})$.
\end{lemma}
\begin{proof}
Let $O_\Lambda(t)=e^{itH_\Lambda}Oe^{-itH_\Lambda}$ be its time-evolution under $H_\Lambda$ in the Heisenberg picture. Using the Baker-Campbell-Hausdorff formula and the remainder theorem, we can write the Taylor series for $O_{\Lambda}(t)=e^{itH_\Lambda}Oe^{-itH_\Lambda}$ as a truncation to degree $d$, $O_{d,\Lambda}(t)$, and an integral remainder $R_{d,\Lambda}(t)$:
\begin{align}
    &O_\Lambda(t) = O_{d,\Lambda}(t) + R_{d,\Lambda}(t), \text{ with } O_{d,\Lambda}(t):=  \underset{r=0}{\overset{d}{\sum}}\frac{(it)^r}{r!}[(H_\Lambda)^r,O], R_{d,\Lambda}(t) :=\int_0^tO_\Lambda^{(d+1)}(\tau)\frac{(t-\tau)^d}{d!}d\tau .\label{eq:OTruncated}
\end{align}

Where $O_\Lambda^{(d+1)}(t)$ is the $d+1$-th derivative of $O_{\Lambda}(t)$. This way we can find a truncated series for the expected operator $\mathbb{E}_\Lambda(O_\Lambda(t)) = \mathbb{E}_\Lambda(O_{d,\Lambda}(t)) + \mathbb{E}_\Lambda(R_{d,\Lambda}(t))$. Our goal is to obtain a bound on $t$ such that the Taylor series converges and the truncation incurs an error at most $\epsilon >0$:
\begin{align}
    & ||\mathbb{E}_\Lambda(O_\Lambda(t))-\mathbb{E}_\Lambda(O_{d,\Lambda}(t))|| = || \mathbb{E}_\Lambda(R_{d,\Lambda}(t))|| < \epsilon .\label{eq:resbound}
\end{align}

We can bound the expectation of the remainder in two regimes: with high probability, all random variables $\Lambda_b$ are bounded by a constant ($|\Lambda_b|\leq \beta,\forall b$, for some $\beta\geq 1$). Then, since $H_\Lambda$ is geometrically local we can use Lieb-Robinson bounds to bound the remainder. With low probability, some variable will be larger than $\beta$ and the remainder will be large. Thus we need to pick $\beta$ large enough that the probability is small to compensate. Let $\beta>0$ be a threshold and consider the event $E:|\Lambda_b|\leq \beta,\forall b$, i.e. the event that all $|\Lambda_b|$ are bounded by $\beta$, and let the complementary event be $E^c$. Using Bayes' rule and the triangle inequality we obtain:
\begin{align}
    &|| \mathbb{E}_\Lambda(R_{d,\Lambda}(t))||\leq ||\mathbb{E}_\Lambda(R_{d,\Lambda}(t)|E)||\mathbb{P}(E)+||\mathbb{E}_\Lambda(R_{d,\Lambda}(t)|E^c)||\mathbb{P}(E^c).
\end{align}
We first bound the probabilities of a single $\Lambda_b$ being outside the threshold $\beta\geq 1$. Recall that $\Lambda_b\sim \mathcal{N}(\lambda_b,\Sigma_{bb})$, $\Sigma_{bb}\leq 1, |\lambda_b|\leq 1,  \forall b$, so using the complementary error function bound $\text{erfc}(x)\leq e^{-x^2}$ we obtain:
\begin{align}
    &\mathbb{P}(|\Lambda_b|> \beta) = \text{erfc}\bigg(\frac{\beta-|\lambda_b|}{\sqrt{2\Sigma_{bb}}}\bigg) \leq  \exp\bigg(-\frac{(\beta-|\lambda_b|)^2}{2\Sigma_{bb}}\bigg)\leq  \exp\bigg(-\frac{(\beta-1)^2}{2}\bigg).
\end{align}

We can use the bound $\mathbb{P}(E)\leq 1$, while the probability of $E^c$ is bounded by:
\begin{align}
    &\mathbb{P}(E^c)
    = \mathbb{P}\big(\exists b:\,|\Lambda_b|>\beta\big)
    \leq \sum_{b=1}^L \mathbb{P}(|\Lambda_b|>\beta)
    \leq L\,e^{-\frac{(\beta-1)^2}{2}}. \label{eq:Ec_prob}
\end{align}

From Lemmas \ref{lem:bound_E}, \ref{lem:bound_Ec} we will consider $\beta\geq 10$ and obtain the following bounds on the conditional expectations:
\begin{align}
    &||E_{\Lambda}(R_{d,\Lambda}(t)|E)|| \leq \frac{1}{1-\frac{2ec\beta t}{a}}\Big(\frac{2ec\beta t}{a}\Big)^{d+1}\text{ if }t<\frac{a}{2ec\beta},\\
    &||E_{\Lambda}(R_{d,\Lambda}(t)|E^c)|| \leq \frac{(6tL)^{d+1}}{(d+1)\cdot d \cdot... \cdot (\lfloor \frac{d}{2}\rfloor+1) }\Big\lfloor\frac{d+1}{2}\Big\rfloor^{\frac{1}{2}}.
\end{align}

We find $t,d,\beta$ such that both $||E_\Lambda(R_{d,\Lambda}(t)|E)||$ and $||E_\Lambda(R_{d,\Lambda}(t)|E)||L\cdot \text{exp}(-\frac{(\beta-1)^2}{2})$ are smaller than $\frac{\epsilon}{2}$. We set $t=\frac{1}{3}\frac{a}{2ec\beta}<\frac{a}{2ec\beta}$, which is required for convergence in event $E$, and $d \geq\log(\frac{1}{\epsilon})$. This suffices to show $||E_\Lambda(R_{d,\Lambda}(t)|E)||\leq \frac{\epsilon}{2}$:
\begin{align}
    &||E_\Lambda(R_{d,\Lambda}(t)|E)|| \leq \frac{1}{1-\frac{2ec\beta t}{a}}\bigg(\frac{2ec\beta t}{a}\bigg)^{d+1} \leq \frac{1}{2}3^{-d}\leq \frac{1}{2}e^{-d} = \frac{\epsilon}{2}.
\end{align}

To bound $||E_\Lambda(R_{d,\Lambda}(t)|E^c)||\mathbb{P}(E^c)$ it suffices to use $\beta = \max(1+\sqrt{2}\sqrt{\log(\frac{2}{\epsilon})+(d+2)\log(L)},10)$:
\begin{align}
    &||E_{\Lambda}(R_{d,\Lambda}(t)|E^c)||\mathbb{P}(E^c)<\frac{(6t)^{d+1}}{(d+1)\cdot d \cdot... \cdot (\lfloor \frac{d}{2}\rfloor+1) }\Big\lfloor\frac{d+1}{2}\Big\rfloor^{\frac{1}{2}} \cdot (L^{d+2}\cdot e^{-\frac{1}{2 }(\beta-1)^2})<\frac{\epsilon}{2}.
\end{align}
The last product is $\frac{\epsilon}{2}$ by the choice of $\beta$, and the rest is less than $1$ by the choice of $t$.  Thus for $t \leq \Big(\frac{6ec}{a}\max\Big(1+\sqrt{2}\sqrt{\log(\frac{2}{\epsilon})+(d+2)\log(L)},10\Big)\Big)^{-1}$ the truncation converges and Eq. \ref{eq:resbound} holds.
\end{proof}

Since the expected observable can be well approximated by a polynomial of low degree, we can find a polynomial fit to the time trace $\frac{1}{2^N}\mathbb{E}_\Lambda(\tr(P_a e^{-itH_\Lambda}P_be^{itH_\Lambda}))$ using the analysis from the non-Markovian case. Since we need to do tomography on the second order of the Dyson series, i.e. two nested commutators, to obtain estimate $\Sigma_{a,b}$, the same measurement protocol as in the non-Markovian case allows us to obtain estimates for $\lambda_a,\Sigma_{a,b}$. By estimating $\mathbb{E}_\Lambda(\Lambda_a\Lambda_b)= \Sigma_{a,b}+\lambda_a\lambda_b$ and $\lambda_a,\lambda_b$ to precision $\frac{\epsilon}{3}$, we can learn $\Sigma_{a,b}$ to precision $\epsilon$:
\begin{align}
    &|\hat{\Sigma}_{a,b}-\Sigma_{a,b}|=|\widehat{\lambda_a\lambda_b}-\mathbb{E}_\Lambda(\Lambda_a\Lambda_b) -\hat{\lambda}_a\hat{\lambda}_b+\lambda_a\lambda_b|\leq |\widehat{\lambda_a\lambda_b}-\mathbb{E}_\Lambda(\Lambda_a\Lambda_b)| +|\hat{\lambda}_a\hat{\lambda}_b-\lambda_a\lambda_b|\leq \frac{\epsilon}{3}+\frac{2\epsilon}{3}=\epsilon 
\end{align}
We used the following bound for the second term:
\begin{align}
    &|\hat{\lambda}_a\hat{\lambda}_b-\lambda_a\lambda_b|= |\hat{\lambda}_a\lambda_b- \hat{\lambda}_a\lambda_b+\hat{\lambda}_a\hat{\lambda}_b-\lambda_a\lambda_b| = |\hat{\lambda}_a(\lambda_b-\hat{\lambda}_b)+\lambda_b(\lambda_a-\hat{\lambda}_a)|\overset{(1)}{\leq}|\hat{\lambda}_a|\frac{\epsilon}{3} +|\lambda_b|\frac{\epsilon}{3}\overset{(2)}{\leq} \frac{2\epsilon}{3}
\end{align}
In $(1)$ we used the triangle inequality and in $(2)$ we used $|\lambda_b|,|\hat{\lambda}_a|\leq 1$.\\

Therefore, our total sample complexity is the same as for the non-Markovian case, except that the overhead due to the offset is only $\epsilon \mapsto \frac{\epsilon}{3}$:
\begin{align}
    &O\Big(\frac{2\cdot 9((\mathfrak{d}_0+2)(s+\mathfrak{d})\mathfrak{d}^2+1)}{\epsilon_{S}^2}\cdot \log\bigg(\frac{L\cdot 2\cdot 3^{2k_{SE}+\mathfrak{d}_0 -2}}{\delta}\bigg)\bigg). 
\end{align}
And since we only need to estimate up to the $2$-nd derivative of the time trace using Eq.\eqref{eq:precision_K0}:
\begin{align}
    &\epsilon_{K,0}\leq \Big(\frac{1}{d^{4}}+\frac{4(\mathfrak{d}_0+2)^2\mathfrak{d}^2}{8d^{6}}\Big)\epsilon_{S,4}=\Big(1+\frac{4(\mathfrak{d}_0+2)^2\mathfrak{d}^2}{8d^{2}}\Big)3\epsilon_{S}.
\end{align}
Finally, in order to estimate the covariance to precision $\epsilon_{K,0}$ we use a number of samples:
\begin{align}
    &O\Bigg(\frac{2\cdot 9^2((\mathfrak{d}_0+2)(s+\mathfrak{d})\mathfrak{d}^2+1) \Big(1+\frac{4(\mathfrak{d}_0+2)^2\mathfrak{d}^2}{8d^{2}}\Big)^2}{2\epsilon_{K,0}^2}\cdot \log\bigg(\frac{L\cdot 2\cdot 3^{2k_{SE}+\mathfrak{d}_0 -2}}{\delta}\bigg)\Bigg).
\end{align}

\begin{lemma}\label{lem:bound_E}
    In event $E$ all the random variables satisfy $|\Lambda_a|\leq \beta$. For time $t<\frac{a}{2ec\beta}$, where $a,c$ are constants from the Lieb-Robinson bound, the Taylor series for $O_\Lambda(t)$ converges. A truncation to degree $d$, $O_{\Lambda,d}(t)$ will yield an error:
    \begin{align}
        &||\mathbb{E}_\Lambda(R_{d,\Lambda}(t)|E)||\leq \frac{1}{1-\frac{2ec\beta t}{a}}\bigg(\frac{2ec\beta t}{a}\bigg)^{d+1}.
    \end{align}
\end{lemma}
\begin{proof}
Since $H_\Lambda$ is geometrically local and all variables are bounded, the following Lieb-Robinson bound holds:
\begin{thm} \label{thm:lieb-robinson_com}
    (Lieb-Robinson bound) Let $A,B$ be observables acting on subsystems $X,Y$. Their commutator is bounded:
    \begin{align}
        &||[A,B]||\leq 2||A||||B|| c\cdot exp(- a\cdot d(X,Y)),
    \end{align}
    where $a,c$ are constants and $d(X,Y)$ is the distance between subsystems $X,Y$.
\end{thm}

The single-commutator Lieb-Robinson bound leads to the following Lieb-Robinson bound on nested commutators:
\begin{align}
    &\underset{a_1,...,a_d}{\sum}||[P_{a_1},...[P_{a_d},A]]||\leq \underset{a_1,...,a_d}{\sum}(2c)^d||A|| \underset{a_j}{\min}(e^{-a\cdot d(P_{a_j},A)})\leq (2c)^d||A||\underset{a_1,...,a_d}{\sum} \big(e^{-a\cdot d(P_{a_1},A)}\cdot...\cdot e^{-a\cdot d(P_{a_d},A)}\big)^\frac{1}{d}, \\
    &= (2c)^d||A||\bigg(\underset{a_1}{\sum}e^{-\frac{a}{d}\cdot d(P_{a_1},A)}\bigg)\cdot ...\cdot \bigg(\underset{a_d}{\sum}e^{-\frac{a}{d}\cdot d(P_{a_d},A)}\bigg) \overset{(1)}{\leq} (2c)^d ||A||\Big(\frac{d}{a}\Big)^d= \Big(\frac{2c}{a}\Big)^dd^d||A|| .\label{eq:lr_nested_commutator}
\end{align} 
In the first line we used that the harmonic mean upper bounds the minimum function, $\min(x_1,...,x_d)\leq (x_1\cdot...\cdot x_d)^\frac{1}{d}$, and in $(1)$ we used Lemma 12 from \cite{kashyap2024accuracy}.\\

We can finally bound $||\mathbb{E}_\Lambda(R_{d,\Lambda}(t)|E)||$ using that $ \mathbb{E}_\Lambda(|\prod_{j=1}^r\Lambda_{a_j}|||\Lambda_b|\leq \beta, \forall b)\leq \beta^r$:
\begin{align}
    &R_{d,\Lambda}(t) = O_{\Lambda}(t)-O_{d,\Lambda}(t)= \underset{r=m+1}{\overset{\infty}{\sum}}\frac{(it)^r}{r!}[(H)^r,O]\Rightarrow \\
    &\Rightarrow ||\mathbb{E}_\Lambda(R_{d,\Lambda}(t)|E)||= \bigg|\bigg|\mathbb{E}_\Lambda\bigg(\underset{r=m+1}{\overset{\infty}{\sum}}\frac{(it)^r}{r!}\underset{i_1,...,i_r}{\sum}\Lambda_{i_1}...\Lambda_{i_r}[P_{i_r},...[P_{i_1},O]...]\Big|E\bigg)\bigg|\bigg|\leq\\
    &\leq \underset{r=m+1}{\overset{\infty}{\sum}}\frac{t^r}{r!}\underset{i_1,...,i_r}{\sum}|\mathbb{E}_\Lambda(\Lambda_{i_1}...\Lambda_{i_r}|E)|\cdot||[P_{i_r},...[P_{i_1},O]...]|| \leq \underset{r=m+1}{\overset{\infty}{\sum}}\frac{t^r}{r!}\beta^r\underset{i_1,...,i_r}{\sum}||[P_{i_r},...[P_{i_1},O]...]||\leq \\
    &\leq \underset{r=m+1}{\overset{\infty}{\sum}}\frac{t^r}{r!}\beta^r\bigg(\frac{2c}{a}\bigg)^rr^r||O||< \underset{r=m+1}{\overset{\infty}{\sum}}\bigg(\frac{2ec\beta t}{a}\bigg)^r, \text{ converges if }t<\frac{a}{2ec\beta}
\end{align}
In the last line we used Eq.\eqref{eq:lr_nested_commutator} and $\frac{r^r}{r!}<e^r$. When the last expression converges, $||\mathbb{E}_\Lambda(R_{d,\Lambda}(t)|E)||\leq \frac{1}{1-\frac{2ec\beta t}{a}}(\frac{2ec\beta t}{a})^{m+1}$.
\end{proof}

\begin{lemma}\label{lem:bound_Ec} In event $E^c$ there is at least one of the random variables $\Lambda$ that is larger than $\beta$ in norm. Suppose that $\beta\geq 12$, then a truncation to degree $d$ of the Heisenberg evolution $O_\Lambda(t)$ will yield an error:
    \begin{align}
        &||\mathbb{E}(R_{d,\Lambda}(t)|E^c)||\leq \frac{(6tL)^{d+1}||O||}{(d+1)\cdot d \cdot... \cdot (\lfloor \frac{d}{2}\rfloor+1) }\Big\lfloor\frac{d+1}{2}\Big\rfloor^{\frac{1}{2}}.
    \end{align}
\end{lemma}
\begin{proof}
Recall that $R_{d,\Lambda}(t) =\int_0^tO_\Lambda^{(d+1)}(\tau)\frac{(t-\tau)^d}{d!}d\tau$ is the integral remainder of the truncation of the Heisenberg time evolution $O_\Lambda(t)$ to degree $d$. By the triangle inequality we get:
\begin{align}
    &||\mathbb{E}_\Lambda(R_{d,\Lambda}(t)|E^c)\leq \int_0^tE_\Lambda(||O_\Lambda^{(d+1)}(\tau)|||E^c)\frac{(t-\tau)^d}{d!}d\tau.
\end{align}

We first bound $||O_\Lambda^{(d+1)}(t)||$ using $O^{(d+1)}_\Lambda (t)=i^{d+1}[(H_\Lambda)^{d+1},O_\Lambda(t)]$:
\begin{align}
    ||O_\Lambda^{(d+1)}(t)|| &= ||i^{d+1}[(H_\Lambda)^{d+1},O_\Lambda(t)]||= \bigg|\bigg|\underset{i_1,...,i_{d+1}}{\sum}\Lambda_{i_1}...\Lambda_{i_{d+1}} [P_{i_{d+1}},...[P_{i_1},O_{\Lambda}(t)]...]\bigg|\bigg|, \notag\\
    &\leq \underset{i_1,...,i_{d+1}}{\sum}|\Lambda_{i_1}...\Lambda_{i_{d+1}}|\cdot|| [P_{i_{d+1}},...[P_{i_1},O_{\Lambda}(t)]...]||,\notag \\
    &\leq \underset{i_1,...,i_{d+1}}{\sum}|\Lambda_{i_1}...\Lambda_{i_{d+1}}| \cdot 2^{d+1}||O_\Lambda(t)||=2^{d+1}||O||\underset{i_1,...,i_{d+1}}{\sum}|\Lambda_{i_1}...\Lambda_{i_{d+1}}|.
\end{align}

We upper bound the commutators by $2^{d+1}$, since $||P_a||=1$, and use that $||O_{\Lambda}(t)||=||O||$ for all $ \Lambda,t$ by unitarity. We now bound the expecation of the random variables. Using Hölder's inequality we obtain the upper bound:
\begin{align}
    &\mathbb{E}(|\Lambda_{i_1}...\Lambda_{i_{d+1}}||E^c) \leq \prod_{k=1}^{d+1}\mathbb{E}(|\Lambda_{i_k}|^{d+1}|E^c)^{\frac{1}{d+1}}\leq \max_k \mathbb{E}(|\Lambda_{i_k}|^{d+1}|E^c).
\end{align}

In the event $E^c$ there is at least one random variable $\Lambda_a$ satisfying $|\Lambda_a|\geq \beta$. We now consider event $F$, where all of them satisfy it: $|\Lambda_a|>\beta$ for all $a$.
\begin{align}
    \mathbb{E}(|\Lambda_{i_k}|^{d+1}|E^c)&= \mathbb{E}(|\Lambda_{i_k}|^{d+1}|F)P(F|E^c)+\mathbb{E}(|\Lambda_{i_k}|^{d+1}|F^c)P(F^c|E^c),\notag\\
    &\leq \mathbb{E}(|\Lambda_{i_k}|^{d+1}||\Lambda_{i_k}|>\beta)P(F|E^c)+\beta^{d+1}P(F^c|E^c).
\end{align}
The right term drops out, since $P(F^c|E^c)=0$ because $F^c\cap E^c = \emptyset$: a random variable can't be both less than or equal to $\beta$ and strictly larger than $\beta$. We bound the left term using that the marginal probability distribution of $\Lambda_{i_k}$ is $\mathcal{N}(\lambda_{i_k},\Sigma_{i_ki_k})$, with $|\lambda_{i_k}|\leq 1,\Sigma_{i_k,i_k}\leq 1$ and Lemma \ref{lem:EZ>Z_0}:
\begin{align}
    &\mathbb{E}(|\Lambda_{i_k}|^{d+1}||\Lambda_{i_k}|>\beta)\leq  2 \frac{((\beta+1)^2+1)}{\beta^2-1}    \exp\bigg(-\frac{(\beta-1)^2}{4} +2\beta\bigg)\Big\lfloor \frac{d+1}{2}\Big\rfloor ^{\frac{1}{2}}\Big\lfloor\frac{d}{2}\Big\rfloor! \cdot 3^{d+1}.
\end{align}
Using $P(F|E^c)\leq 1$ and considering $\beta\geq 10$ we can simplify the dependence on $\beta$:
\begin{align}
    &\mathbb{E}(|\Lambda_{i_1}...\Lambda_{i_{d+1}}|E^c) \leq 2 \Big\lfloor \frac{d+1}{2}\Big\rfloor ^{\frac{1}{2}}\Big\lfloor\frac{d}{2}\Big\rfloor! \cdot 3^{d+1}. \label{eq:bound_Lambdas}
\end{align}

We can now perform the integral above:
\begin{align}
    &||\mathbb{E}_\Lambda(R_{d,\Lambda}(t)|E^c)||\leq \int_0^tE_\Lambda(||O_\Lambda^{(d+1)}(\tau)|||E^c)\frac{(t-\tau)^d}{d!}d\tau, \notag\\
    &\leq \int_0^t2^{d+1}||O||\underset{i_1,...,i_{d+1}}{\sum}\mathbb{E}_\Lambda(|\Lambda_{i_1}...\Lambda_{i_{d+1}}||E^c)\frac{(t-\tau)^d}{d!}d\tau = \frac{(2t)^{d+1}}{(d+1)!}||O||\underset{i_1,...,i_{d+1}}{\sum}\mathbb{E}_\Lambda(|\Lambda_{i_1}...\Lambda_{i_{d+1}}||E^c), \notag\\
    &\leq \frac{(2t)^{d+1}}{(d+1)!}||O||L^{d+1}\Big\lfloor\frac{d+1}{2}\Big\rfloor^{\frac{1}{2}}\Big\lfloor\frac{d}{2}\Big\rfloor!\cdot 3^{d+1}\leq \frac{(6tL)^{d+1}||O||}{(d+1)\cdot d \cdot... \cdot (\lfloor \frac{d}{2}\rfloor+1) }\Big\lfloor\frac{d+1}{2}\Big\rfloor^{\frac{1}{2}}
\end{align}
In the last line we used that there are $L^{d+1}$ terms in the sum in the second line and Eq.\eqref{eq:bound_Lambdas}.
\end{proof}
\begin{lemma}\label{lem:EZ>Z_0,mu} Let $Z\sim \mathcal{N}(\mu,\sigma^2)$ be a Gaussian random variable, $m \geq 1$ an integer, $\beta\geq 1$ a threshold. Then: 
\begin{align} 
    &\mathbb{E}(|Z|^m||Z|\geq \beta) \leq 2 \frac{((\beta+|\mu|)^2+\sigma^2)\sigma^2}{\beta^2-|\mu|^2}    \exp\bigg(-\frac{(\beta-|\mu|)^2}{4\sigma^2} +\frac{2 \beta |\mu|}{\sigma^2}\bigg)\Big\lfloor \frac{m}{2}\Big\rfloor ^{\frac{1}{2}}\Big\lfloor\frac{m-1}{2}\Big\rfloor! (2\sigma+|\mu|)^{m}.
\end{align}
\end{lemma} 
\begin{proof} 
Notice that we can write $Z = X + \mu$, where $X\sim \mathcal{N}(0,\sigma^2)$:
\begin{align}
    &\mathbb{E}(|Z|^m||Z|\geq \beta) \leq \sum_{n=0}^m{m\choose n}|\mu|^{m-n}\mathbb{E}(|X|^n||Z|\geq \beta) = \sum_{n=0}^m{m\choose n}|\mu|^{m-n}\mathbb{E}(|X|^n||X|\geq \beta-|\mu|)\frac{\mathbb{P}(|X|\geq \beta-|\mu|)}{\mathbb{P}(|Z|\geq \beta)},
\end{align}
where we used the definition of conditional expectation in the last step. We will bound the fraction using the two-sided bounds on the Gaussian tail given by the Mills ratio:
\begin{align}
\frac{t}{t^2 + \sigma^2}\,\frac{1}{\sqrt{2\pi}\,\sigma}\, e^{-\frac{t^2}{2\sigma^2}}
\leq \mathbb{P}(|X|\geq t)\leq \frac{2\sigma}{\sqrt{2\pi}\,t} \, e^{-\frac{t^2}{2\sigma^2}}, \qquad t>0.
\end{align}

Since $\mathbb{P}(|Z|\geq \beta)\geq \mathbb{P}(|X|-|\mu|\geq \beta)$ by the triangle inequality, it suffices to bound:
\begin{align}
\frac{\mathbb{P}(|X|\geq \beta-|\mu|)}{\mathbb{P}(|X|\geq \beta+|\mu|)}
\leq 2  \frac{((\beta+|\mu|)^2+\sigma^2)\sigma^2}{\beta^2-|\mu|^2} \exp\Big(\frac{(\beta+|\mu|)^2 - (\beta-|\mu|)^2}{2\sigma^2}\Big)
= 2 \frac{((\beta+|\mu|)^2+\sigma^2)\sigma^2}{\beta^2-|\mu|^2}  \exp\Big(\frac{2 \beta |\mu|}{\sigma^2}\Big).
\end{align}

We can now use Lemma \ref{lem:EZ>Z_0} to bound the expectation value:
\begin{align}
    &\mathbb{E}(|X|^n||X|\geq \beta-|\mu|) \leq \Big\lfloor \frac{n}{2}\Big\rfloor ^{\frac{1}{2}}\Big\lfloor\frac{n-1}{2}\Big\rfloor! (2\sigma)^{n}e^{-\frac{(\beta-|\mu|)^2}{4\sigma^2}}.
\end{align}

Finally, this yields the result:
\begin{align}
    \mathbb{E}(|Z|^m||Z|\geq \beta) &\leq 2 \frac{((\beta+|\mu|)^2+\sigma^2)\sigma^2}{\beta^2-|\mu|^2}  \exp\bigg(\frac{2 \beta |\mu|}{\sigma^2}\bigg) \exp\bigg(-\frac{(\beta-|\mu|)^2}{4\sigma^2}\bigg)\sum_{n=0}^m{m\choose n}|\mu|^{m-n}\Big\lfloor \frac{n}{2}\Big\rfloor ^{\frac{1}{2}}\Big\lfloor\frac{n-1}{2}\Big\rfloor! (2\sigma)^{n},\\
    &\leq 2 \frac{((\beta+|\mu|)^2+\sigma^2)\sigma^2}{\beta^2-|\mu|^2}  \exp\bigg(\frac{2 \beta |\mu|}{\sigma^2}\bigg) \exp\bigg(-\frac{(\beta-|\mu|)^2}{4\sigma^2}\bigg)\Big\lfloor \frac{m}{2}\Big\rfloor ^{\frac{1}{2}}\Big\lfloor\frac{m-1}{2}\Big\rfloor! (2\sigma+|\mu|)^{m}.
\end{align}
\end{proof}

\begin{lemma}\label{lem:EZ>Z_0} Let $Z\sim \mathcal{N}(0,\sigma^2)$ be a Gaussian random variable, $m \geq 1$ an integer, $\beta\geq 0$ a threshold. Then:
\begin{align}
    &\mathbb{E}(|Z|^m||Z|\geq \beta)\leq \Big\lfloor \frac{m}{2}\Big\rfloor ^{\frac{1}{2}}\Big\lfloor\frac{m-1}{2}\Big\rfloor! (2\sigma)^{m}e^{-\frac{\beta^2}{4\sigma^2}}
\end{align}
\end{lemma}
\begin{proof}
    We write the conditional moment in terms of the incomplete Gamma function:
    \begin{align}
        &\mathbb{E}(|Z|^m||Z|\geq \beta) = \int_{-\infty}^{\infty}|z|^{m} \mathbb{P}(z||z|\geq \beta)dz= 2\int_{\beta}^{\infty}z^{m} \mathbb{P}(z)dz= 2\int_{\beta}^\infty z^{m}\frac{e^{-\frac{z^2}{2\sigma^2}}}{\sqrt{2\pi\sigma^2}}dz=\\ 
        &=\frac{2\cdot 2^{\frac{m-1}{2}}\cdot\sigma^m}{\sqrt{2\pi}}\int_{\frac{\beta^2}{2\sigma^2}}^\infty t^{\frac{m-1}{2}}e^{-t}dt = \frac{2^{\frac{m}{2}}\sigma^m}{\sqrt{\pi}}\Gamma\Big(\frac{m+1}{2},\frac{\beta^2}{2\sigma^2}\Big)
    \end{align}
    In the second line we used the change of variables $z = \sigma(2t)^{\frac{1}{2}}, dz = \sigma(2t)^{\frac{-1}{2}}dt$ and the definition of the upper incomplete Gamma function $\Gamma(s,x)=\int_x^\infty t^{s-1}e^{-t}dt$. For positive integer $s$ we have:
    \begin{align}
        &\Gamma(s,x) = (s-1)!e^{-x}\underset{k=0}{\overset{s-1}{\sum}}\frac{x^k}{k!} \leq (s-1)!e^{-x}2^{s-1}e^{\frac{x}{2}} = 2^{s-1}(s-1)!e^{-\frac{x}{2}}
    \end{align}
    
    If $m$ is even, then $s=\frac{m+1}{2}$ is not an integer. Let $s = k + r$ be the decomposition of a real number $s$ into its integer, $k$, and decimal, $0\leq r<1$, parts. The ratio $\frac{\Gamma(s,x)}{\Gamma(s)}$ is monotonically increasing in $s$ for $s > 0,x>0$ \cite{tricomi1950sulla}, so $\Gamma(k+r,x)\leq \frac{\Gamma(k+r)}{\Gamma(k+1)}\Gamma(k+1,x)$ and we can use Gautschi's inequality $\frac{\Gamma(k+r)}{\Gamma(k+1)}<k^{r-1}$ to obtain the bound $\Gamma(r+d,x)<r^{d-1}\Gamma(r+1,x)$. Using $s=\frac{m+1}{2}$ with $m$ even, $k=\frac{m}{2},r=\frac{1}{2}$ and we get $\Gamma(\frac{m+1}{2},x)<(\frac{m}{2})^{-\frac{1}{2}}\Gamma(\frac{m}{2}+1,x)$. Therefore, we obtain the bounds:
    \begin{align}
        &m \text{ odd}: \mathbb{E}(Z^m||Z|\geq \beta) = \frac{2^{\frac{m}{2}}\sigma^m}{\sqrt{\pi}}\Gamma\Big(\frac{m+1}{2},\frac{\beta^2}{2\sigma^2}\Big) \leq \frac{2^{\frac{m}{2}}\sigma^m}{\sqrt{\pi}}\Big(\frac{m-1}{2}\Big)!2^{\frac{m-1}{2}}e^{-\frac{\beta^2}{4\sigma^2}}< (2\sigma)^m\Big(\frac{m-1}{2}\Big)!e^{-\frac{\beta^2}{4\sigma^2}}\\
        &m \text{ even}: \mathbb{E}(Z^m||Z|\geq \beta) < \frac{2^{\frac{m}{2}}\sigma^m}{\sqrt{\pi}}\Big(\frac{m}{2}\Big)^{-\frac{1}{2}}\Gamma\Big(\frac{m}{2}+1,\frac{\beta^2}{2\sigma^2}\Big)\leq \frac{2^{\frac{m}{2}}\sigma^m}{\sqrt{\pi}}\Big(\frac{m}{2}\Big)^{-\frac{1}{2}}\Big(\frac{m}{2}\Big)!2^{\frac{m}{2}}e^{-\frac{\beta^2}{4\sigma^2}}<(2\sigma)^m\Big(\frac{m}{2}\Big)^{\frac{1}{2}}\Big(\frac{m}{2}-1\Big)!e^{-\frac{\beta^2}{4\sigma^2}}
    \end{align}
    Using the floor function we see that $\lfloor \frac{m-1}{2}\rfloor$ is $\frac{m-1}{2}$ for $m$ odd and $\frac{m}{2}-1$ for $m$ even, so we get the global bound:
    \begin{align}
    &\mathbb{E}(|Z|^m||Z|\geq \beta)\leq \Big\lfloor \frac{m}{2}\Big\rfloor ^{\frac{1}{2}}\Big\lfloor\frac{m-1}{2}\Big\rfloor! (2\sigma)^{m}e^{-\frac{\beta^2}{4\sigma^2}}
\end{align}
\end{proof}

\end{document}